\documentclass[11pt]{article}
\usepackage[utf8]{inputenc}

\usepackage{graphicx}
\usepackage{grffile}

\usepackage{geometry}
\usepackage{color}
\usepackage{amsmath}
\usepackage{amsfonts}
\usepackage{stackrel}
\usepackage{graphicx}
\usepackage{verbatim}

\usepackage[section]{placeins}
\usepackage{fancyhdr}

\usepackage{amssymb}
\usepackage{amsthm}

\usepackage{pgf,tikz}
\usetikzlibrary{arrows}
\usepackage{tikz-cd}

\usepackage{cite}
\usepackage[arrow,curve,matrix,arc,2cell]{xy}
\xyoption{rotate}
\UseAllTwocells

\usepackage[unicode]{hyperref}
\hypersetup{
    pdftoolbar=true,        
    pdfmenubar=true,        
    pdffitwindow=false,     
    pdfstartview={FitH},    
    pdftitle={Analysis of sojourn time distributions for semi-Markov models},    
    pdfauthor={Kelli Francis-Staite}{Langford White},     
    pdfsubject={Analysis of sojourn time distributions for semi-Markov models},   
    pdfcreator={Kelli Francis-Staite}{Langford White},   
    pdfproducer={Kelli Francis-Staite}{Langford White}, 
    pdfkeywords={Geometric Distribution}{Families of distributions}{Discrete Distributions}{Finite Distributions}{Semi-Markov models}{Maximum likelihood estimators}, 
    pdfnewwindow=true,      
    colorlinks=true,       
    linkcolor=black,          
    citecolor=black,        
    filecolor=black,      
    urlcolor=black 
    linkcolor=red,          
    citecolor=black,        
    filecolor=magenta,      
    urlcolor=cyan           
}
\usepackage{caption}
\usepackage{subcaption}
\usepackage[sort=none,numberedsection]{glossaries-extra} 
\newcommand{\Rbb}{\mathbb{R}}

\newcommand{\Ac}{\mathcal{A}}
\newcommand{\Bc}{\mathcal{B}}
\newcommand{\Pc}{\mathcal{P}}

\newcommand{\Er}{\mathbf{E}}
\renewcommand{\Pr}{\mathbf{P}}
\newcommand{\lbr}{\left\{}
\newcommand{\rbr}{\right\}}
\newcommand{\lb}{\left(}
\newcommand{\rb}{\right)}

\newcommand{\ra}{\rightarrow}

\newcommand{\bcb}{\begin{color}{blue}}
\newcommand{\bcr}{\begin{color}{red}}
\newcommand{\bcg}{\begin{color}{cyan}}
\newcommand{\ec}{\end{color}}

\usepackage{cleveref}

\theoremstyle{plain}
\newtheorem{thm}{Theorem}[section]
\newtheorem{lem}[thm]{Lemma}
\newtheorem{prop}[thm]{Proposition}
\newtheorem{cor}[thm]{Corollary}
\theoremstyle{definition}
\newtheorem{dfn}[thm]{Definition}
\newtheorem{ex}[thm]{Example}

\numberwithin{figure}{section}
\numberwithin{equation}{section}

\DeclareMathOperator{\Var}{Var}
\DeclareMathOperator{\E}{E}

\usepackage{bookmark}

\title{Analysis of sojourn time distributions for semi-Markov models}
\author{Kelli Francis-Staite \and Langford B. White}

\begin{document}
\pdfbookmark[section]{Title and Abstract}{TitleAndAbstract}
\maketitle

\begin{abstract}
    This report aims to characterise certain sojourn time distributions that naturally arise from semi-Markov models. To this end, it describes a family of discrete distributions that extend the geometric distribution for both finite and infinite time. We show formulae for the moment generating functions and the mean and variance, and give specific examples. We consider specific parametrised subfamilies;  the linear factor model and simple polynomial factor models. We numerically simulate drawing from these distributions and solving for the Maximum Likelihood Estimators (MLEs) for the parameters of each subfamily, including for very small sample sizes. The report then describes the determination of the bias and variance of the MLEs, and shows how they relate to the Fisher information, where they exhibit appropriate concentration effects as the sample size increases. Finally, the report addresses an application of these methods to experimental data, which shows a strong fit with the simple polynomial factor model. \end{abstract}

\pagebreak
\pdfbookmark[section]{\contentsname}{toc}
\tableofcontents
\pagebreak

\section{Introduction}
This report describes and analyses sojourn time distributions which can then be applied to parametrise the semi-Markov model (SMM). Soujourn time distributions are the distributions that govern how long a semi-Markov model remains in each state. Unlike usual Markov-models where the sojourn time distribution is geometric, semi-Markov models allow any distribution to be used.

Such sojourn time distributions have been modelled before by using known distributions such as the gamma, Poisson, multinomial, more general exponential family members and Coxian phase-type distributions particularly in the context of hidden semi-Markov models (HSMMs) \cite{Mitchell1993,Yu2010,Sansom2001,Bulla2006,Guedon2003,Duong2008,BullaThesis}. An R-package is available for modelling some of these distributions \cite{Bulla2013}. Non-parametric versions have also been studied in the context of expectation maximisation for HSMMs \cite{Bulla2006,Yu2010}, however the underlying properties of these distributions have not been studied in the literature and we seek to do this \Cref{sec:genGeometric}. Further, the two subfamilies of the distributions we study have not previously appeared in the literature. 

In general, the sojourn time distributions we study can be considered discrete generalisations of the geometric distribution in finite or infinite time. There are many approaches for generalisations of the geometric distribution \cite{Rattihalli2021,Tripathi2016,Gomez2010}, although they usually only consider infinite time. For example, considering the exponential as the continuous analog of the geometric distribution and the gamma distribution as generalising the exponential distribution, then discretisations of the gamma distribution can be considered generalisations of the geometric distribution as in \cite{Chakraborty2012}. 

Unlike other geometric generalisations, ours will involve the parameters $\rho$. These parameters can be considered as the marginal probabilities that a semi-Markov chain stays in the same state for the next time step, as in \Cref{sec: SemiMarkovRelationship}. For a Markov chain, the geometric sojourn time distribution corresponds to constant $\rho$ for each state, and here we extend this to consider linear and simple polynomial factor models for $\rho$. We are not aware of any other research on generalisations of the geometric distribution that has such an approach.

Markov models are often applied to model behaviours where only small amounts of data may be available, so we require parsimonious models which nevertheless have sufficient generality to capture different kinds of qualitative sojourn time behaviours. We also require models for which statistically efficient estimation methods can be defined. At this stage, we make the obvious point that whilst maximum likelihood estimators (MLEs) are statistically unbiased and efficient for large sample sizes \cite{DeGroot}; this may not be the case for small sample sizes such in our experimental data. So the study of the performance of MLEs forms an important part of this study which we undertake in \Cref{sec:MLEsandCRLBs}.

Statistical inference methods for hidden SMMs (HSMMs) are typically more computationally demanding than the equivalents for HMMs, so it is also of interest to determine whether an observed process has Markovian state transitions or the more general semi-Markov ones. We show in \Cref{sec:experidata} how we approached this for our experimental data, which we can use as a methodology for testing this kind of hypothesis in general. We also see a strong fit between this experimental data and the simple polynomial factor model, suggesting the theory performs well in experimental settings.

Overall, the contents of this report gives a novel approach to modelling sojourn time distributions that will be of use for (hidden) semi-Markov models, and already holds appeal for modelling experimental data. The results concerning the sojourn time distributions may be of independent interest for the discrete distribution literature.

\subsection{Layout of the report}

The contents of this report is as follows. In \Cref{subsec:defnOfGeneralisation} we define the family of distributions we aim to study, which generalises the geometric distribution. We give general properties including the mean, variance, moment and probability generating functions as well as results that determine when the moments exist in general. These results are new. In \Cref{subsec:linearex} and \Cref{subsec:simplePolyexample} we define the subfamilies of interest, the linear factor model and simple polynomial factor models, and describe their distributions, which is also new research. In \Cref{sec: SemiMarkovRelationship} we show the relationship to semi-Markov models and discuss further results in this area.

In \Cref{sec:MLEsandCRLBs} we study how to determine the MLEs of the subfamilies of interest, studying the linear factor model in \Cref{subsec:linearMLEs} and the polynomial factor models in \Cref{subsec:polyMLE}. In each case, we determine numerically the MLEs from various samples and show how the variance of the MLEs closely follows the inverse Fisher information which is related to the Cram\'{e}r-Rao lower bound \cite{DeGroot} suggesting our numerical methods are behaving efficiently. We study how this changes with different sample sizes. We also discuss methods to determine the size of the support of the distribution, i.e. the maximum sojourn time. 

In \Cref{sec:experidata} we analyse our experimental data. We show how the data  has sojourn times which do not appear to follow a geometric distribution, and why a simple polynomial factor model is a better fit. We then find the MLEs for such a model for each of the three tasks, including determining the maximum sojourn time, and discuss how well these models fit. We are satisfied with the model for tasks 1 and 3, and suggest some improvements for task 2. In \Cref{sec:conclusion} we conclude and discuss future work.

In \Cref{appendix} the second author expands on the work for the linear factor model. He confirms the results in \Cref{subsec:linearMLEs} with separate simulations, as well as shows how the expected log-likelihood explains some of the behaviour of the MLEs and how to further understand estimation when the maximum sojourn time is unknown. In \Cref{appn:ExperimentalData} we include the experimental data analysed in \Cref{sec:experidata}.

\section{Generalising the geometric distribution}\label{sec:genGeometric}
Here we define a family of discrete distributions that generalise the geometric distribution in both finite and infinite state cases. We give properties of this family, including formulae for the moments, and prove results relating to their existence in the infinite support case. We note that this family is particularly general as it can be used to characterise all discrete distributions that are supported on $1,2,3,\ldots, T$ with $T$ a positive integer and possibly infinite. 

We then focus on two specific subfamilies, the linear factor model in \Cref{subsec:linearex} and some simple polynomial factor models in \Cref{subsec:simplePolyexample}. Finally we detail the relationship to semi-Markov models in \Cref{sec: SemiMarkovRelationship} and the relationship to matrix analytic methods.

As far as we are aware, all results in sections \Cref{subsec:defnOfGeneralisation},  \Cref{subsec:linearex} and \Cref{subsec:simplePolyexample} are new contributions to the literature on discrete distributions, which we expect will also lead to new contributions in Markov processes research.

\subsection{Definition and properties} \label{subsec:defnOfGeneralisation}
We define the family of discrete distributions that we intend to study then study properties of the distributions. 
\begin{dfn}\label{defn:discreteRho}
Take any $\rho(1),\rho(2),\ldots, \rho(T-1)\in (0,1]$ with $\rho(1)>0$, and set $\rho(T)=0$, where $T$ is a positive integer (which we may allow to go to $\infty$). Then we can define a discrete distribution with support contained in $\{1,2,\ldots T\}$ having probability mass function (PMF)
\begin{align} 
\label{eqn:PMF}
f(k)=(1-\rho(k))\rho(k-1)\rho(k-2)\ldots \rho(1) \\\nonumber
= (1-\rho(k))\prod_{t=1}^{k-1}\rho(t) \quad \text{for }k=1,2,\ldots, T,
\end{align}
where we use the convention that product $\prod_{t=1}^{0}\rho(t)$ is equal to 1. 
\end{dfn}
Note that we could allow some $\rho(k)=0$ for $1 \le k<T$, but then $f(k) =1$ and then $f(k+i) = 0$ for all $i=1,2,\ldots,T$, reducing the support of $f$. Also, whenever $\rho(k)=1$ we have $f(k) = 0$, so as we see below, this is sufficiently general to capture all discrete distributions with support contained in $\{1,2,\ldots \infty\}$. We can also shift the distribution of $f$ to distribution $f_t$ with support contained within $\{1+t,2+t,\ldots, T+t\}$ for any integer $t$ by setting $f_t(t+k) = f(k)$. We will consider this in \Cref{sec:experidata}.

It is clear that each $f(k)\in [0,1]$ as products of numbers in $[0,1]$ remains in $[0,1]$. We have 
\begin{align*}
    \sum_{k=1}^Tf(k) &= \sum_{k=1}^T(1-\rho(k))\rho(k-1)\rho(k-2)\ldots \rho(1)\\
    &= 1 + \sum_{k=2}^T\prod_{t=1}^{k-1} \rho(t) - \sum_{k=1}^T\prod_{t=1}^{k} \rho(k)\\
    &= 1 + \sum_{k=1}^{T-1}\prod_{t=1}^{k} \rho(t) - \sum_{k=1}^T\prod_{t=1}^{k} \rho(k)\\
    &= 1-\prod_{t=1}^{T} \rho(k) = 1,
\end{align*}
as $\rho(T)=0$. So $f$ gives a valid PMF.

If we let $\rho(k)=p$ for some fixed $p\in (0,1)$ and take $T\to \infty$ then this is gives the geometric distribution with PMF for $k=1,2,\ldots$ \begin{equation} f(k) = p^{k-1}(1-p).\label{eqn:PMFgeometric} \end{equation}

In general, we can interpret the $\rho(k)$ as the probability of an event occurring at time $k$ given it also occurred at time $k-1,k-2,\ldots,1$. Then $f(k)$ is the probability of the event not occurring at time $k$ given it occurred at time $1,2,\ldots,k-1$. In the geometric case, the time steps are independent and $\rho(k)$ does not change, however the general $f$ described above allows for dependence between each time. We expand on this in \Cref{sec: SemiMarkovRelationship} where we show how this relates to a semi-Markov model.

\begin{lem}\label{lem:relationship}
We have the relationship 
\begin{equation} \label{eqn:relationship}
\sum_{t=k}^{T}f(t)=1-\sum_{t=1}^{k-1}f(t) = \frac{f(k)}{1- \rho(k)} = \prod_{t=1}^{k-1}\rho(t).
\end{equation}
\end{lem}
\begin{proof}
By induction, we can show 
\begin{equation} \label{eqn:rhoInTermsOff}
\rho(k) = \frac{1- f(1) - f(2) - \ldots - f(k)}{1- f(1) - f(2) - \ldots - f(k-1)}=1-\frac{f(k)}{\sum_{t=k}^{T}f(t)}=1-\frac{f(k)}{1-\sum_{t=1}^{k-1}f(t)},\end{equation}
and the result follows.
\end{proof}
Note that \Cref{eqn:rhoInTermsOff} shows that any discrete PMF supported $k=1,2,\ldots,T$ (with $T$ possibly infinite) can be defined in terms of $\rho(k)$, so in fact \Cref{defn:discreteRho} is an alternative way to describe all such distributions. If $F$ is the corresponding cumulative distribution function to $f$, then 
\[ \rho(k) =\frac{1 - F(k)}{1-F(k-1)}\] and also 
 \begin{align}  F(k) &=\sum_{t=1}^kf(k) = \sum_{t=1}^k(1-\rho(t))\prod_{s=1}^{t-1} \rho(s) \nonumber\\
 &= 1+ \sum_{t=1}^{k-1}\prod_{s=1}^{t-1}\rho(s) - \sum_{t=1}^k\prod_{s=1}^{t-1}\rho(s)\nonumber\\
 &= 1-\prod_{t=1}^k\rho(t).
 \label{eqn:FinTermsOfrho}\end{align}
Note that if $F(k) = 1$ then $\rho(k) = 0$ and this means the distribution has reached the end of its support. We use this in \Cref{sec:experidata}. The next proposition characterises all moments of the distribution in terms of the $\rho(k)$.

\begin{prop} \label{prop:MGFPGF}
For a random variable $X$ with PMF $f$ as in \Cref{eqn:PMF} we have moment generating function (MGF) 
\[M_X(t) = \sum_{k=1}^{T-1} (e^{t(k+1)}-e^{tk})\prod_{t=1}^k \rho(k)+e^t.\]
The $n$-th derivatives are
\[M_X^{(n)}(t)  = \sum_{k=1}^{T-1} ((k+1)^{n}e^{t(k+1)}-k^{n}e^{tk})\prod_{t=1}^k\rho(t) +e^t.\]
 So we have 
\[\E(X^n) =M_X^{(n)}(0) = \sum_{k=1}^{T-1} ((k+1)^{n}-k ^{n})\prod_{t=1}^k\rho(t) +1,\] and 
then 
\begin{align*} \E(X) &= \sum_{k=1}^{T-1} \prod_{t=1}^k\rho(t) +1,\\
\Var(X) &= \E(X) +2 \sum_{k=1}^{T-1}k \prod_{t=1}^k\rho(t)- \E(X)^2.
\end{align*}

The probability generating function is 
\[P_X(z) = \sum_{k=1}^{T-1} (z^{k+1}-z^k)\prod_{t=1}^k\rho(t) + z\]

with $n$-th derivatives 

\[P_X^{(n)}(z)= \sum_{k=n}^{T-1} \left(\frac{(k+1)!z^{k+1-n}}{(k+1-n)!}-\frac{k!z^{k-n}}{(k-n)!}\right)\prod_{t=1}^k\rho(t) + \frac{dz}{dz^n}.\]

The factorial moments are 
\begin{align*}\E(X(X-1)\cdots (X-n+1))& = P_X^{(n)}(1)= \sum_{k=n}^{T-1} \frac{k!n}{(k+1-n)!}\prod_{t=1}^k\rho(t) + \frac{dz}{dz^n}.\end{align*}
\end{prop}

\begin{proof}
The moment generating function is 
\begin{align*}
    M_X(t) &= \sum_{k=1}^{T}e^{kt}f(k)\\
    &= \sum_{k=1}^{T}e^{kt}(1-\rho(k))\prod_{t=1}^{k-1}\rho(t)\\
    &=\sum_{k=1}^{T}e^{kt}\prod_{t=1}^{k-1}\rho(t)-\sum_{k=1}^{T}e^{kt}\prod_{t=1}^{k}\rho(t)\\
    &= e^t+ \sum_{k=2}^{T}e^{kt}\prod_{t=1}^{k-1}\rho(t)-\sum_{k=1}^{T}e^{kt}\prod_{t=1}^{k}\rho(t)\\
    &= e^t+ \sum_{k=1}^{T-1}e^{(k+1)t}\prod_{t=1}^{k}\rho(t)-\sum_{k=1}^{T-1}e^{kt}\prod_{t=1}^{k}\rho(t)\\    &=\sum_{k=1}^{T-1} (e^{t(k+1)}-e^{tk})\prod_{t=1}^k \rho(k)+e^t
\end{align*}
The formulas for the derivatives $M_X^{(n)}(t)$ and $E(X)$ follow directly. We then have the variance
\begin{align*}\Var(X) &= \E(X^2)-\E(X)^2 = \sum_{k=1}^{T-1} (2k +1)\prod_{t=1}^k\rho(k) +1 - \E(X)^2 \\
&= \E(X) +2 \sum_{k=1}^{T-1}k \prod_{t=1}^k\rho(k)- \E(X)^2.\end{align*}

The probability generating function is 
\[P_X(t) = \sum_{k=1}^{T}z^kf(k) = M_X(\log(z))= \sum_{k=1}^{T-1} (z^{k+1}-z^{k})\prod_{t=1}^k \rho(k)+z.\]
The $n$-th derivatives $P_X^{(n)}(z)$ are straightforward to calculate, so the factorial moments are 
\begin{align*}\E(X(X-1)\cdots (X-n+1))& = P_X^{(n)}(1)= \sum_{k=n}^{T-1} \left(\frac{(k+1)!}{(k+1-n)!}-\frac{k!}{(k-n)!}\right)\prod_{t=1}^k\rho(t) + \frac{dz}{dz^n} \\
&= \sum_{k=n}^{T-1} \frac{k!n}{(k+1-n)!}\prod_{t=1}^k\rho(t) + \frac{dz}{dz^n}.\end{align*}
\end{proof}

If $T$ is finite, it is clear that all moments are finite. Below we consider the case when $T$ is infinite.
\begin{prop} \label{prop:existenceOfMoments}
Let $X$ be a random variable with PMF $f$ as in \Cref{eqn:PMF} with $T= \infty$, and say \[
\lim_{k\to \infty} \rho(k) = r\] where we must have $r\in [0,1]$ as $\rho(k)\in[0,1]$. Then all moments exist whenever $r<1$. If $r = 1$ then some moments may exist. 
\end{prop}
\begin{proof}
The ratio test says that a series $\sum_{i=1}^\infty a_i$ for $a_i\in \Rbb$ converges whenever 
\[ \lim_{k\to \infty} \left | \frac{a_{k+1}}{a_k}\right | <1.\] It diverges if the limit is $>1$ and may or may not converge if this limit is $1$. Say we have $\lim_{k\to \infty} \rho(k) =r$, then we can consider the series 
\[\sum_{k=1}^\infty k^n \prod_{t=1}^k\rho(t)\] for some $n=1,2,\ldots, \infty$, which appears in the formula for the moments $E(X^n)$ in \Cref{prop:MGFPGF}.  Then  
\begin{align*}
\lim_{k\to \infty} \left | \frac{a_{k+1}}{a_k}\right |&= \lim_{k\to \infty} \left | \frac{(k+1)^n \prod_{t=1}^{k+1}\rho(t)}{k^n \prod_{t=1}^k\rho(t)}\right |\\
&= \lim_{k\to \infty} (1+\frac{1}{k^n}) \rho(k+1) \\
& \to 1\cdot r
\end{align*}
 therefore all moments centred at $0$ converge when $r<1$ by the ratio test. As all moments of order $n$ can be written in terms of moments centred at $0$ of up to order $n$, then all moments converge for $X$ when $r<1$. Similarly, they diverge if $r>1$ (which is not possible as we have $\rho(k)\in[0,1]$ for all $k$) and may or may not converge if $r=1$.
\end{proof}

Note that we do not necessarily have that $\lim_{k\to \infty} \rho(k)$ converges, so a more general statement would be than any subsequence of the $\rho(k)$ that converges must converge to something less than $1$ to guarantee existence of the moments. This is equivalent to the following stronger corollary of \Cref{prop:existenceOfMoments}.

\begin{cor}
Let $X$ be a random variable with PMF $f$ as in \Cref{eqn:PMF} with $T= \infty$, and say \[
\limsup_{k\to \infty} \rho(k) = r\] where we must have $r\in [0,1]$ as $\rho(k)\in[0,1]$. Then all moments exist whenever $r<1$. If $r = 1$ then some moments may exist. 
\end{cor}

If $\limsup_{k\to \infty} \rho(k) = 1$, that is, there are (sub)sequences of the $\rho(k)$ which converge to $1$, more complicated arguments must be considered as in the following example. 

\begin{ex}
If $\rho(k) = \frac{k}{k+1}$ for $k =1,2,\ldots, \infty$ then $\lim_{k\to \infty} \rho(k) =1$ and we have \[\sum_{k=1}^\infty\prod_{t=1}^{k}\rho(t) = \sum_{k=1}^\infty \frac{1}{k+1},\] which diverges and no moments exist. These $\rho$ correspond to PMF \[f(k)=\frac{1}{k}-\frac{1}{1-k} = \frac{1}{k(k+1)}.\]

More generally, if $\rho(k) = \frac{k^m}{(k+1)^m}$ for some positive integer $m$ then still $\lim_{k\to \infty} \rho(k) =1$, however we have 
\[\sum_{k=1}^\infty k^n\prod_{t=1}^{k}\rho(t) = \sum_{k=1}^\infty  \frac{k^n}{(k+1)^m}.\] When $n+1<m$, we have 
\begin{align*}
    \sum_{k=1}^\infty  \frac{k^n}{(k+1)^m} &= \sum_{k=1}^\infty \frac{1}{(1+\frac{1}{k})^n}\frac{1}{(1+k)^2}\frac{1}{(1+k)^{m-2}}\\
    &\le \sum_{k=1}^\infty \frac{1}{(1+k)^2} = \frac{\pi^2}{6} -1 
\end{align*}
so this sum converges and this means moments up to order $m-2$ exist. When $n = m-1$ we have 
\begin{align*}
    \sum_{k=1}^\infty  \frac{k^n}{(k+1)^m} &= \sum_{k=1}^\infty \frac{1}{(1+\frac{1}{k})^{m-1}}\frac{1}{1+k}\\
    &\ge \frac{1}{2^{m-1}} \sum_{k=1}^\infty \frac{1}{1+k} \\
    & \to \infty
\end{align*}
so this diverges, and all moments of order $m-1$ or higher do not exist. These $\rho$ correspond to PMF \[f(k) = \frac{1}{k^m} - \frac{1}{(1+k)^m}.\]
\end{ex}

\begin{ex}
For a random variable $X$ with the geometric distribution $\rho(k)=p$, $T\rightarrow \infty$ as in \Cref{eqn:PMFgeometric} then using \Cref{prop:MGFPGF} we have mean
\[ \E(X) = 1+\sum_{j=1}^{\infty}p^j = 1+\frac{1}{1-p}-1=\frac{1}{1-p},\]
and variance
\begin{align*}
    \Var(X) & = \E(X) +2\lim_{T\to \infty}\sum_{k=1}^{T-1} k\prod_{t=1}^{k}\rho(t)- \E(X)^2\\
    &=\frac{1}{1-p} + 2\lim_{T\to \infty}\sum_{k=1}^{T-1}kp^k - \frac{1}{(1-p)^2}\\
    & = \frac{1}{1-p} + 2\lim_{T\to \infty}p\frac{d}{dp}\sum_{k=1}^{T-1}p^k - \frac{1}{(1-p)^2}\\
    & = \frac{1}{1-p} + 2\lim_{T\to \infty}p\frac{d}{dp}\left(\frac{1-p^T}{1-p}-1\right) - \frac{1}{(1-p)^2}\\
    & = \frac{1}{1-p} + 2p\frac{d}{dp}\left(\frac{1}{1-p}-1\right) - \frac{1}{(1-p)^2}\\
    & = \frac{1}{1-p} + \frac{2p}{(1-p)^2} - \frac{1}{(1-p)^2}\\
    & = \frac{1-p+2p-1}{(1-p)^2}\\
    &=\frac{p}{(1-p)^2}.
\end{align*}
Both of these results are as expected from the geometric distribution.
\end{ex}

\subsection{Linear factor model} \label{subsec:linearex}
In this section we consider the case where we have a linear $\rho$, that is 
\[ \rho(k)=ak+b \quad \text{for all } k=1,2, \ldots, T.\] For this to be valid, we require that $ak+b \in [0,1]$ for all $k$. This then means if we send $T\to \infty$ we must make $a=0$, resulting in the geometric distribution. If we consider $T$ finite then we require the endpoints $k=1,T-1$ are in $[0,1]$, which means we need $a+ b\in [0,1]$, so $a\in [-b,1-b]$ and $a(T-1)+b \in [0,1]$, so $a\in [-\frac{b}{T-1},\frac{1-b}{T-1}]$. Combining these we have 
\[ \max\{-b,\frac{-b}{T-1}\} \le a \le \min\{1-b,\frac{1-b}{T-1}\}.\] To make sure the end points are non-overlapping, we then also require that $b\in [-\frac{1}{T},\frac{T-1}{T-2}]$ for $T\ge 3$. 

A natural choice in this case would be to set $a,b$ such that $a(T)+b = 0$, which gives $a=-b/T$ and $b\in (0,\frac{T}{T-1}]$. Then \[ \rho(k) = b\left(1-\frac{k}{T}\right)\] and the corresponding $f$ is as follows
\begin{align*}
    f(k)&=(1-\rho(k))\rho(k-1)\rho(k-2)\ldots \rho(1) \\
    &= (1+b(\frac{k}{T}-1))(b(1-\frac{k-1}{T}))(b(1-\frac{k-2}{T}))\cdots(b(1-\frac{1}{T}))\\
    &=(1+b(\frac{k}{T}-1))\left(\frac{b}{T}\right)^{k-1}(T-k+1)(T-k+2)\cdots (T-1)\\
    &= (1+b(\frac{k}{T}-1))\left(\frac{b}{T}\right)^{k-1}\frac{(T-1)!}{(T -k)!}.
\end{align*}
This is a natural extension of the geometric distribution over a finite domain.

Note that we can similarly use $a$ instead of $b$, where we have  \[ \rho(k) = a\left(k-T\right),\] with $b=-aT$, $a\in [\frac{1}{1-T},0)$. 

We can graph $f$ for certain values of $b$ (and corresponding $a$) and $T=10$, as in \Cref{fig:LinearGeomExtension}. 
\begin{figure}[!hpt]
    \centering
    \includegraphics[width=13cm]{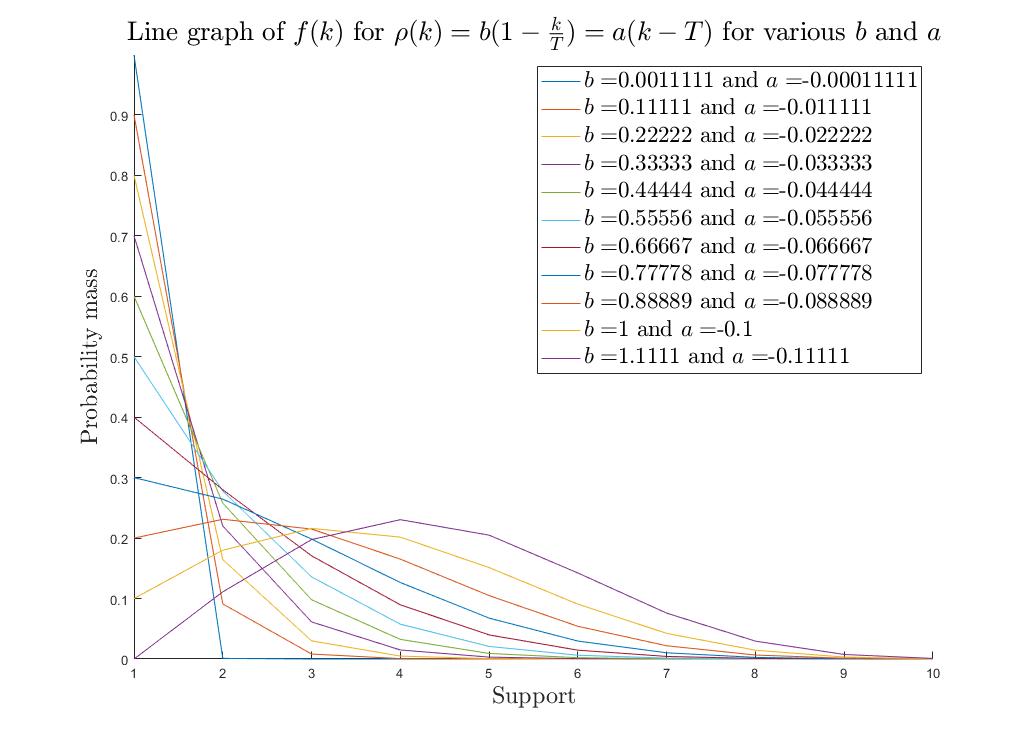}
    \caption{Plot of $f$ with $\rho(k)=b(1-\frac{k}{T})=a(k-T)$ with $T=10$ and various $b$.}
    \label{fig:LinearGeomExtension}
\end{figure}
We see quite different behaviours depending on the value of $b$. The larger $b$ exhibit a non-monotonic, uni-modal distribution with maximum near the centre of the support, while the smaller $b$ exhibit more geometric-like qualities, with maximum at $k=1$ and decreasing for larger $k$. Varying $T$ gives similar results.

\subsection{Simple polynomial factor models} \label{subsec:simplePolyexample}

Here we concentrate on models of the form \[\rho(k) = a(k-c)^n+b\] with $\rho(T)=0$. We require $n$ to be a fixed positive integer. 

We again set $\rho(T)=0$ which gives $a(T-c)^n+b=0$, so we must have $a=\frac{-b}{(T-c)^n}$ provided we assume $b\ne 0$. Note that if $b=0$ and $\rho(T) = 0$ then either $c=T$ and we are considering linear cases $\rho(T) = a(k-T)$ from the previous section, or $a =0$ and we have $\rho = 0$.

Given this we assume $c\ne T$ then writing in terms of $b$ we have 
\begin{align*} 
\rho(k)& = b\left(1-\frac{(k-c)^n}{(T-c)^n}\right),\\
f(k) &= \left(1+b\left(\frac{(k-c)^n}{(T-c)^n}-1\right)\right)\left(\frac{b}{(T-c)^n}\right)^{k-1}\prod_{t=1}^{k-1}((T-c)^n-(t-c)^n).
\end{align*}

Similarly, writing in terms of $a$ we have 
\begin{align*}
\rho(k) &= a\left((k-c)^n - (T-c)^n\right),\\ 
f(k) &= \left(1+a\left((T-c)^n-(k-c)^n\right)\right)a^{k-1}\prod_{t=1}^{k-1}((t-c)^n-(T-c)^n).
\end{align*}

We require $\rho(k)\in (0,1]$ for all $k=1,\ldots, T-1$. Whether $n$ is even or odd changes the requirements on $b$ and $c$ to ensure this. We get the conditions that $c\in \mathbb{R}\setminus\{T\}$ and that 
\begin{align*}
    n \text{ even}, \quad  c \le \frac{T+1}{2}, \quad&  0\le b \le \frac{(T-c)^n}{(T-c)^n -(k-c)^n}, \quad&   \frac{1}{(k-c)^n-(T-c)^n}\le a \le 0\\
    n \text{ even}, \quad  c > T,  \quad & \frac{(T-c)^n}{(T-c)^n-(k-c)^n}\le b \le 0,  \quad & 0\le a \le \frac{1}{(k-c)^n-(T-c)^n}\\
    n \text{ odd}, \quad  c < T,  \quad & 0\le b \le \frac{(T-c)^n}{(T-c)^n -(k-c)^n},  \quad & \frac{1}{(k-c)^n-(T-c)^n}\le a \le 0\\
    n \text{ odd}, \quad  c > T,  \quad & \frac{(T-c)^n}{(T-c)^n -(k-c)^n} \le b \le 0,  \quad & 0 \le a \le  \frac{1}{(k-c)^n-(T-c)^n}.
\end{align*}
Note that $n$ even and $c\in [\frac{T+1}{T},T)$ implies $b=0$, which we will again exclude here. With certain applications in mind as in \Cref{sec:experidata}, we are most interested in the odd $n$ cases particularly $n=3$ with $c<T$, however we will consider all cases in generality when possible. 

In \Cref{fig:PolyGeomExtension} we graph $f$ for certain values of $b$ (and corresponding $a$) with $n=3$, $c=4$ and $T=15$. 
\begin{figure}[!hpt]
    \centering
    \includegraphics[width=13cm]{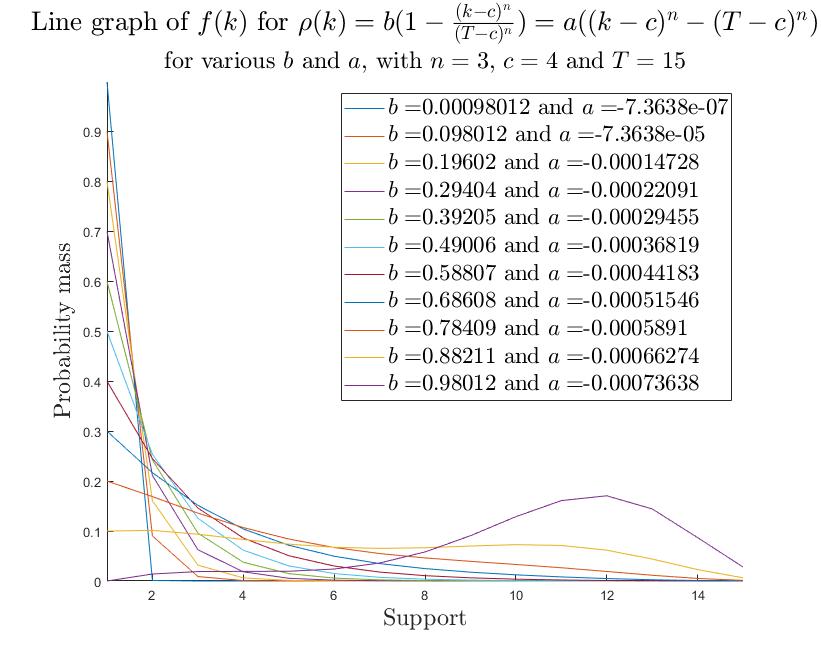}
    \caption{Plot of $f$ with $\rho(k)=b\left(1-\frac{(k-c)^n}{(T-c)^n}\right)=a\left((k-c)^n - (T-c)^n\right)$ with $T=15$, $n=3$, $c=4$ and various $b$.}
    \label{fig:PolyGeomExtension}
\end{figure}
Note that here we see quite different behaviours depending on the value of $b$. The larger $b$ exhibit a non-monotonic, uni-modal or even bi-modal distributions with maxima towards $T$, while the smaller $b$ exhibit more geometric-like qualities, with maximum at $k=1$ and decreasing for larger $k$. Varying $T$ gives similar results.

We also graph $f$ while varying $c$ and $b$ with fixed $n=3$, and $T=15$,  as in \Cref{fig:PolyGeomExtensionC} 
\begin{figure}[!hpt]
    \centering
    \includegraphics[width=13cm]{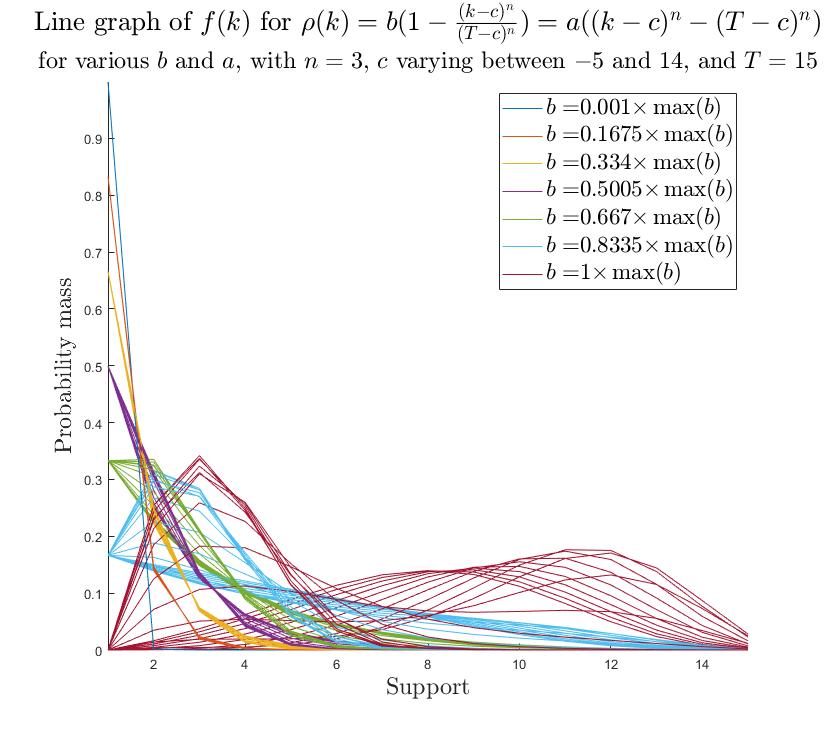}
    \caption{Plot of $f$ with $\rho(k)=b\left(1-\frac{(k-c)^n}{(T-c)^n}\right)=a\left((k-c)^n - (T-c)^n\right)$ with $T=15$, $n=3$, $c=-5,-4,\ldots,14$ and various $b$. Note that $\max(b) = \min_{k} \frac{(T-c)^n}{(T-c)^n -(k-c)^n}$.}
    \label{fig:PolyGeomExtensionC}
\end{figure}
Note that here we see quite different behaviours depending on the value of $b$. The larger $b$ exhibit a non-monotonic, uni-modal or even bi-modal distributions with maxima towards $T$, while the smaller $b$ exhibit more geometric-like qualities, with maximum at $k=1$ and decreasing for larger $k$. That such a varied behaviour is achieved through only two parameters will be helpful for fitting to real data as in \Cref{sec:experidata}.  Varying $T$ gives similar results.

In \Cref{fig:PolyGeomExtensionC2} we further separate out the different behaviours for fixed $b$ and changes in $c$. 
\begin{figure}[!hpt]
    \centering
    \includegraphics[width=15cm]{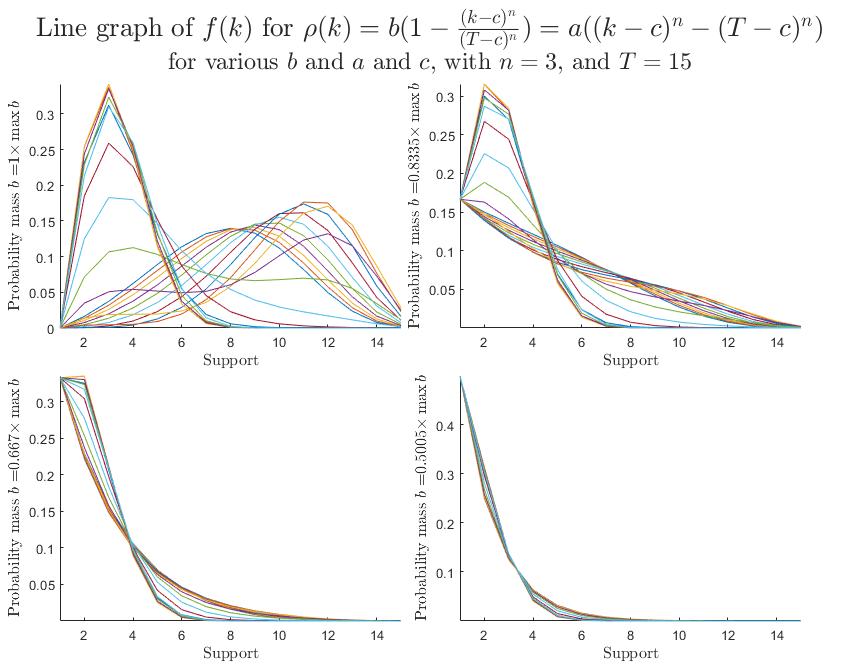}
    \caption{Plot of $f$ with $\rho(k)=b\left(1-\frac{(k-c)^n}{(T-c)^n}\right)=a\left((k-c)^n - (T-c)^n\right)$ with $T=15$, $n=3$, $c=-5,-4,\ldots,14$ and various $b$ that are fixed in each plot. Note that $\max(b) = \min_{k} \frac{(T-c)^n}{(T-c)^n -(k-c)^n}$.}
    \label{fig:PolyGeomExtensionC2}
\end{figure}

\subsection{Relationship to SMM} \label{sec: SemiMarkovRelationship}
In this section we describe the relationship between the family of distributions in \Cref{defn:discreteRho} and semi-Markov models. This is one of the motivations for this paper. This section also contains details on how semi-Markov models can be considered using Matrix analytic methods. This will be important in future applications of the sojourn time distributions studied so far.

\subsubsection{Sojourn time distributions arising from SMMs}
In a semi-Markov model, we have a random process $(X_n,\tau_n)$ $n=1,2,\ldots$. Here $X_n$ is a random variable for each $n$ with state space $S=\{1,2,3,\ldots,s\}$ and $\tau_n$ is a random variable for each $n$. Either $\tau_n$ has a continuous state space in $(0,T_{\max})$ for some $T_{\max}>0$ possibly infinite or a discrete state space in $1,2,\ldots,T_{\max}$ for some integer $T_{\max}>0$ possibly infinite. Here we will only consider discrete state spaces.

Note that we differ from some traditional presentations of semi-Markov models as in \cite[\S1.9]{Medhi2003} where often $X_n$ is only the transitions of the process and $\tau_n$ is the times these transitions occur, assumed to be continuous times, giving a \emph{Markov renewal process} $(X_n,\tau_n)$. Instead, here our $X_n$ will be a discrete random process that may remain in the same state as $n$ increments, and the $\tau_n$ will record the number of times the process has remained in the same state continuously. 

We require the independence condition that 
\begin{align*} P((X_{n+1},\tau_{n+1})=(j,t)|  (X_{n},\tau_{n}),(X_{n-1},\tau_{n-1}),\ldots, (X_{1},\tau_{1})) \\
= P((X_{n+1},\tau_{n+1})=(j,t)|(X_{n},\tau_{n})).
\end{align*} 
and the time-homogenous condition that 
\[ P((X_{n+m+1},\tau_{n+m+1})|(X_{n+m},\tau_{n+m})) = P((X_{n+1},\tau_{n+1})|(X_{n},\tau_{n}))\]
for all $m=1,2,\ldots$.

This means that $Z_n = (X_n,\tau_n)$ is a Markov chain.   We require that $\tau_n$ count the number of consecutive time steps $X_n=j$ has spent in state $j$, so that if we observe that both $X_n$ and $X_{n+1}$ are equal to $j$, then
\[P((X_{n+1},\tau_{n+1})=(j,t)|(X_{n} =j,\tau_{n})) = 0 \qquad \text{if }   \tau_{n}\ne t-1, t\ge 2,\] and otherwise if $X_n=i\ne j$ then
\[P((X_{n+1},\tau_{n+1})=(j,t)|(X_{n} =i,\tau_{n})) = 0 \qquad \text{unless }   t=1.\] 
Then $(X_n,\tau_n)$ is specified by the following probabilities
\begin{align*}
    P(X_{n+1}=j,\tau_{n+1}=t_2|X_{n+1}=i,\tau_{n+1}=t_1)=\begin{cases} A_{i,j}(t_1) & \text{if } t_2 =1, j\ne i,\\
    \rho_i(t_1) & \text{if } i=j, t_2 = t_1+1, \\
    0 & \text{otherwise.} \end{cases}
\end{align*}
Here the $\rho_i(k)$ can be considered as the margin probability that the chain stays in the same state for the next time step. With this set up, we call $X_n$ a \emph{semi-Markov process}. \Cref{fig:Z_trans2} shows the allowable state transitions for the process $Z_n$.

\begin{figure}[!htp]
\centering
\begin{tikzpicture}[line cap=round,line join=round,>=triangle 45,x=1.0cm,y=1.0cm]
\clip(4.,0.5) rectangle (12.5,7.5);
\draw [rotate around={0.:(7.83,6.38)},line width=2.pt] (7.83,6.38) ellipse (1.0453503525729777cm and 0.901031275608368cm);
\draw [rotate around={0.:(10.98,6.36)},line width=2.pt] (10.98,6.36) ellipse (1.0480183729244286cm and 0.9041252733925524cm);
\draw [rotate around={0.:(5.58,4.58)},line width=2.pt] (5.58,4.58) ellipse (1.0295649597816658cm and 0.8826686843942214cm);
\draw [rotate around={0.:(5.6,1.68)},line width=2.pt] (5.6,1.68) ellipse (1.0103414909727655cm and 0.8601685464960177cm);
\draw [->,line width=2.pt] (6.914880793461978,5.9444726038992215) -- (6.21075932820925,5.277623444471598);
\draw [->,line width=2.pt] (8.88,6.44) -- (9.937153222134956,6.4497087863121845);
\draw [->,line width=2.pt] (7.64678916372626,5.492915139373649) -- (6.14,2.42);
\draw (7.2195015639526835,6.773992246644394) node[anchor=north west] {$(j_1,t)$};
\draw (10.01583752600971,6.773992246644394) node[anchor=north west] {$(j_1,t+1)$};
\draw (5.000907793200267,5.000105511465399) node[anchor=north west] {$(j_2,1)$};
\draw (5.009896074623204,2.045910172333133) node[anchor=north west] {$(j_{s},1)$};
\draw (5.301685140236067,3.526996123322203) node[anchor=north west] {$\vdots$};
\draw (8.902493766325249,7.3816172521783715) node[anchor=north west] {$\rho_i(t)$};
\draw (5.208767170275509,6.413214899608595) node[anchor=north west] {$A_{j_1j_2}(t)$};
\draw (6.8058062469563265,3.754855500397444) node[anchor=north west] {$A_{j_1j_s}(t)$};
\end{tikzpicture}
\caption{Shows allowable transitions out of a state for the augmented chain $Z_n$. \label{fig:Z_trans2}}
\end{figure}
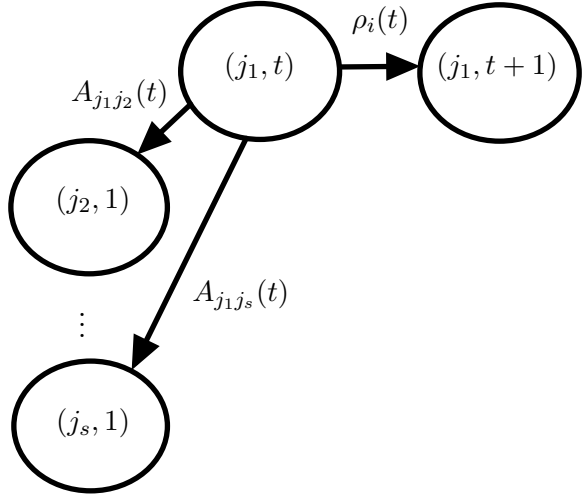

\begin{ex}\label{ex:MarkovSMM}
Note semi-Markov processes are an extension of Markov processes. Given a time homogenous Markov process $X_n$ specified by transition probabilities $P(X_{n+1}=j|X_n=i) = A_{i,j}$, then we can define $\tau_n$ as the consecutive time $X_{n}$ has spent in the current observed state since it entered this state. We then get that 
\[P(X_{n+1}=i,\tau_n=t+1|X_{n}=i,\tau_n=t)=P(X_{n+1}=i|X_n = i)=A_{i,i} = \rho_i(t)\] and for $j\ne i$ we get that
\[P(X_{n+1}=j,\tau_n=1|X_{n}=i,\tau_n=t)=P(X_{n+1}=j|X_n = i)=A_{i,j}.\] Note that neither of these probabilities depend on the time $t$. 
\end{ex}

We have suggestively used $\rho_i$ in the notation above. This is due to the $\rho_i$ satisfying the conditions in \Cref{defn:discreteRho} for each $i$. We have that 
\begin{align*}
   P(&\text{$X_n$ enters state $i$ then stays at state }i\text{ for exactly }k\text{ consecutive steps})\\
   &=\sum_{j_1\ne i, j_2\ne i}P(X_{n+k}=j_1\ne i,X_{n}=X_{n+1}=\ldots =X_{n+k-1}=i, X_{n-1}= j_2)\\
   &=\sum_{j_1\ne i, j_2\ne i}P(X_{n+k}=j|X_{n}=\ldots =X_{n+k-1}=i,X_{n-1}= j_2)\cdot\\
   & \quad P(X_{n}=\ldots =X_{n+k-1}=i, X_{n-1}= j_2)\\
   &=(1-\rho_i(k))\sum_{j_2\ne i}P(X_{n+k-1}=i|X_{n}=\ldots =X_{n+k-2}=i, X_{n-1}= j_2)\cdot\\&\quad P(X_{n}=\ldots =X_{n+k-2}=i, X_{n-1}= j_2)\\
       &=(1-\rho_i(k))\rho_i(k-1)\sum_{j_2\ne i}P(X_{n}=\ldots =X_{n+k-2}=i, X_{n-1}= j_2)\\
      &\quad \vdots\\
      &=(1-\rho_i(k))\rho_i(k-1)\rho_i(k-2)\ldots \rho_i(1)\sum_{j_2\ne i}P(X_n=i,X_{n-1}=j_2),
\end{align*}
Then we have that 
\begin{align*} P(&X_{n+k+1}\ne i, X_{n+k}=X_{n+k-1}=\ldots=X_{n+1}=i|X_n=i,X_{n-1}\ne i) \\
\quad & = (1-\rho_i(k))\rho_i(k-1)\rho_i(k-2)\ldots \rho_i(1).\end{align*}
This is a PMF and the corresponding $f_i$ to the $\rho_i$ from \Cref{defn:discreteRho}. We call this the \emph{sojourn time} distribution of a state $i$, with support $t=1,2,\ldots, T_i$. An aim of our research is to be able to characterise such distributions. Note that if the semi-Markov process is a Markov-process as in \Cref{ex:MarkovSMM} then the sojourn time distributions are geometric with parameters $A_{ii}$.

See \cite[\S1.9]{Medhi2003} and \cite[\S8]{Sahner1996} for further details on semi-Markov models, although they mainly focus on continuous models. Note that in that case a Markov-process would have an exponential sojourn time distribution. 

\subsubsection{Matrix analytic form} \label{subsec:MAM}

Matrix analytic methods (MAMs) \cite{LR99} are a convenient way of manipulating probabilities associated with Markov chains having a two-dimensional state space as is the case for the augmented chain $Z_n$. MAMs are generally associated with {\it quasi birth-death} processes but we can apply the ideas here. Here we associate the {\it levels} with the sojourn time variable, and the {\it phases} with the state variable. 

Here we assume $T_i=T$ are equal for all sojourn time distributions, although we can extend this when they are not equal. Then using the MAM approach, we let $z_n =$ Vec$(Z_n)$ so the transition matrix $\Ac$ for the $z_n$ process is specified by
\begin{align*}
\Ac^\prime & = \left[ \begin{array}{ccccc} A(1)^\prime & A(2)^\prime & \cdots & \cdots & A(T)^\prime \\
D(1) & 0&&& 0\\
0 & D(2) & 0& &0\\
\vdots&&\ddots&&\vdots\\
0&&& D(T-1) & 0\end{array} \right] \ .
\end{align*}
Each block in $\Ac$ is of size $s \times s$ and given by
\begin{align*}
\left[ A(t) \right]_{i,j} & = \left\{ \begin{array}{ll} A_{i,j}(t) & i \neq j \\ 0 & i =j \end{array} \right. \\
D(t) & = \text{diag} \lb A_{i,i}(t) \rb = \text{diag} \lb \rho_i(t)\rb \ .
\end{align*} 
The process can thus either (i) go up one level and stay in the same state (phase) or (ii) drop back to level 1 and go to some other state.
From any state $(i,T)$, the process must jump to some $(j,1), j \neq i$.
Note that $\Ac$ is row stochastic and all blocks of $\Ac$ are row sub-stochastic apart from $A(T)$ which is row stochastic.

To extend this definition when the $T_i$ are not equal, then set $T=\max_iT_i$. When $t\ge T_i$ then set $A(t)=A(T_i)$, otherwise each $A$ has the definition above. When $t\ge T_i$ then $[D(t)]_{i,i}=0$ and otherwise each $D$ has the definition above.

\subsubsection*{Stationary Distribution}

The MAM formulation admits a conceptually simple process for finding the stationary distribution $\Pi$ for $z_n$ via the usual approach of solving $\Ac^\prime \, \Pi = \Pi$. Since $\Ac$ is row stochastic, $\Pi$ exists and is unique. The block form of $\Ac$ leads to :
\begin{align*}
\Pi(t) & = D(t-1) \, \Pi(t-1), \ t = 2, \ldots, T \\
\Pi(1) & = \underbrace{\lb \sum_{t=1}^{T} \ Q(t)^\prime \rb}_{\Bc^\prime} \ \Pi(1) \,
\end{align*}
where
 \begin{align*}
 Q(t) = D(1) \, D(2) \, \cdots \, D(t-1) \, A(t) \ .
 \end{align*}
The quantities $\left[ Q(t)\right]_{i,j}$ are the conditional probabilities of the process $X_n$ staying in a state $i$ for exactly $t$ steps then going to $j \neq i$ given it entered $i$, i.e. 
\begin{align*}
\left[ Q(t)\right]_{i,j} & = \Pr \lbr X_{n+1} = X_{n+2} = \cdots X_{n+t-1} = i, X_{n + t} = j | X_n = i, X_{n-1}\ne i \rbr\\
&= A_{i,j}(t)\rho_i(t-1)\rho_i(t-2)\ldots \rho_i(1)
\end{align*}
and note that
\begin{align*}
\sum_{j\ne i} \left[ Q(t)\right]_{i,j} & = \sum_{j\ne i} \Pr \lbr X_{n+1} = X_{n+2} = \cdots X_{n+t-1} = i, X_{n + t} = j | X_n = i, X_{n-1}\ne i \rbr\\
&= \sum_{j\ne i} A_{i,j}(t)\rho_i(t-1)\rho_i(t-2)\ldots \rho_i(1)\\
&= (1-\rho_i(t))\rho_i(t-1)\rho_i(t-2)\ldots \rho_i(1)
\end{align*}
is our sojourn time distribution as in \Cref{sec:experidata}.
Thus the matrix $\Bc$ contains the probabilities over all possible sojourn times in state $i$, of $X_n$ transitioning from state $i$ to some state $j \neq i$ at a later time. So $\Bc$ is the transition probabilities for the {\it embedded Markov Chain}. 
Then $\Bc$ is row stochastic, and we can find a non-negative solution $\Pi(1)$ with $\Pi(1) = \Bc^\prime \, \Pi(1)$. The remaining levels of $\Pi$ are then easily computed. The process is computationally efficient and of order $O(s^2 \, T) + O(s^3)$. 

\subsubsection*{Dynamics and Forgetting}

The forgetting properties for the augmented process are determined by the second largest magnitude eigenvalue of $\Ac^\prime$. Consider the equation $\Ac^\prime \, \Phi = \lambda \, \Phi$, then from the block form we have
\begin{align*}
\Phi(t) & = \lambda^{-1} \, D(t-1) \, \Phi(t-1), \ t = 2, \ldots, T \\
\Phi(1) & = \lambda^{-1} \, \lb \sum_{t=1}^{T} \ \lambda^{1-t} \, Q(t)^\prime \rb \ \Phi(1) \,
\end{align*}
Consider the matrix polynomial
\begin{align*}
\Pc(\lambda) & = \lambda^{T} \, I - \sum_{t=1}^{T} \ \lambda^{T-t} \, Q(t)^\prime
\end{align*}
then the eigenvalues are given by det $\Pc(\lambda) = 0$. Whilst we have not as yet studied forgetting properties of SMCs, these remain an important issue and the above formulation permits their study with reduced computational complexity. We will return to this issue in future work.

\section{Maximum likelihood estimation and Cram\'{e}r-Rao bounds} \label{sec:MLEsandCRLBs}
In this section we study maximum likelihood estimators (MLEs) of the parameters for the linear factor model and the simple polynomial factor models in \Cref{subsec:linearex} and \Cref{subsec:simplePolyexample}. To do this, we first write formulae for the likelihoods, then determine their properties in an effort to understand how to solve for the MLEs. 

We then numerically simulate drawing from such distributions with known parameters and determining the MLEs numerically. We compare our results for different sample sizes and initial parameters, including discussing the Cram\'{e}r-Rao lower bound for the variance of estimates for the parameters and how closely this aligns with the variance of our MLEs.
\subsection{Linear factor model} \label{subsec:linearMLEs}
Recall the linear factor model from \Cref{subsec:linearex} with $\rho(k)=ak+b$. We will write this in terms of the parameter $b$ so that $\rho(k) = b(1-\frac{k}{T})$. Given samples $x_1,\ldots, x_n$ we wish to know the MLE for the parameters for $f$. We first consider $T$ fixed. We then have log-likelihood
\[ l(b;\mathbf{x}) = \sum_{i=1}^n \left(\log(1+b(\frac{x_i}{T}-1)) + (x_i-1)\log(\frac{b}{T}) + \log\frac{(T-1)!}{(T -x_i)!}\right) .\]
Differentiating with respect to $b$ we have
\begin{align}
    \frac{d l}{db} & =  \sum_{i=1}^n \left(\frac{\frac{x_i}{T}-1}{1+b(\frac{x_i}{T}-1)} + \frac{x_i-1}{b}\right)\nonumber\\
    & = \sum_{i=1}^n \frac{\frac{x_i}{T}-1}{1+b(\frac{x_i}{T}-1)} + \frac{n
    \bar{x}-n}{b}, \label{eqn: derivativeLinearb}
\end{align}
 where $\bar{x}=\frac{1}{n}\sum_{i=1}^nx_i$ denotes the sample mean.
When we send $b$ towards zero, we have that the right most term goes to infty and the left most term is finite, so this means we have a positive slope around $b=0$, while the sign at $b=\frac{T}{T-1}$ can be positive or negative.

We can check the second derivative of $l$ w.r.t. $b$ is
\begin{align*}
    \frac{d^2 l}{db^2} &  = \sum_{i=1}^n \frac{-(\frac{x_i}{T}-1)^2}{(1+b(\frac{x_i}{T}-1))^2} + -\frac{n\bar{x}-n}{b^2}.
    \end{align*}
which is negative for all values of $b$. This means that the likelihood is concave, so that the maximum occurs either at the end point with $b=\frac{T}{T-1}$ or at the (only) point (if it exists) where the first derivative is equal to zero for $b\in (0,\frac{T}{T-1}]$. 

We can use numerical methods to solve for this $b$, which is the MLE, using the formula for the first derivative. We can also reduce further to a polynomial form before solving.

To do this, set the first derivative equal to zero and then we have 
\begin{align*}
    \frac{n-n\bar{x}}{b}& = \sum_{i=1}^n \frac{\frac{x_i}{T}-1}{1+b(\frac{x_i}{T}-1)}\\
    & = \frac{\sum_{i=1}^n(\frac{x_i}{T}-1)\prod_{j\ne i}(1+b(\frac{x_j}{T}-1))}{\prod_{i=1}^n(1+b(\frac{x_i}{T}-1))}
\end{align*}
 so that we require to find $b$ such that it solves the polynomial equation
 \begin{align*}
    0&= b\sum_{i=1}^n(\frac{x_i}{T}-1)\prod_{j\ne i}(1+b(\frac{x_j}{T}-1))+ (n\bar{x}-n)\prod_{i=1}^n(1+b(\frac{x_i}{T}-1)) \\
    &= n\bar{x}-n + (n\bar{x}-n+1) (\sum_{i=1}^nb^i\sum_{j=1}^i\sum_{k_1\ne\ldots\ne k_j=1}^n c_{k_1}c_{k_2}\cdots c_{k_j})
    \end{align*}
with $b\in (0,\frac{T}{T-1}]$ and $c_{i}= \frac{x_i}{T}-1 $. As this is a $n$-order polynomial we can use numerical methods to solve. 

Solving for a root of the polynomial within $(0,\frac{T}{T-1}]$ works well in practice, for small sample sizes. Note that from our previous discussions there is at most 1 root in this domain, and if there is no root the MLE is $\frac{T}{T-1}$. However when the sample size is large, as $|c_i|<1$, then the coefficients of $b$ become exceptionally small for all low order terms, and then function may then return zero up to precision level of the programming language, causing an issue. Solving \Cref{eqn: derivativeLinearb} numerically instead does not have the same issues, although we note that the log-likelihood is often extremely flat around the MLE requiring the tolerance for numerical optimisation to be reduced. 

We can similarly do this with respect to the variable $a$, where $\rho(k) = a(k-T)$. Note that there are minimal changes as $b=-aT$ and $a$ is negative. With $T$ again fixed and $x_1,\ldots, x_n$ a sample, we have the log-likelihood
\[ l(a;\mathbf{x}) = \sum_{i=1}^n \left(\log(1+a(T-x_i)) + (x_i-1)\log(-a) + \log\frac{(T-1)!}{(T -x_i)!}\right) .\]

Differentiating with respect to $a$ we have
\begin{align}
    \frac{d l}{d a} & =  \sum_{i=1}^n \left(\frac{T-x_i}{1+a(T-x_i)} + \frac{x_i-1}{a}\right)\nonumber\\
    & = \sum_{i=1}^n \frac{T-x_i}{1+a(T-x_i)} + \frac{n\bar{x}-n}{a}. \label{eqn: derivativeLineara}
\end{align}
If we consider sending $a$ to zero, then the first term tends to a positive number while the second term tends to negative infinity, so the slope is negative as we send $a$ to zero. The second derivative is 
\begin{align}
    \frac{d^2 l}{d a^2} & =  \sum_{i=1}^n \frac{-(T-x_i)^2}{(1+a(T-x_i))^2} - \frac{n\bar{x}-n}{a^2}. \label{eqn: 2ndderivativeLineara}\\
    &\le 0
\end{align}
As this second derivative is negative, then the function must be concave everywhere. Given the slope is positive as $a$ tends to zero, then there is a unique MLE which occurs either where the derivative is equal to zero or at the lower boundary where $a=\frac{1}{1-T}$.

We can again solve for the MLE numerically. We can also apply the same technique as above to write this as a polynomial equation to find the roots, however the coefficients will now be of the form $(T-x_i)$ which for large sample sizes will create very large numbers that may not be computable by the operating system.

\subsubsection{Numerical analysis}
While we can solve the above equations numerically to determine the MLE for $a$ and/or $b$ for a given sample and given $T$, we do not know if the MLE is a good estimate or not, given the finite sample sizes. We cannot analytically determine the bias nor the variance of the MLE from the formulae above. However, we have simulated the bias and variance of the MLE as below using Matlab. For the definitions of  Cram\'{e}r-Rao Lower Bound and Fisher information used in this section see any standard statistics textbook such as \cite[pg.~517--527]{DeGroot}.

Note that the Cram\'{e}r-Rao Lower Bound $(\mathrm{CRLB})$ is a lower bound for the variance of that estimator under certain regularity conditions. If $T(X)$ is an estimator for a parameter $a$ then we have
\[ \Var(T) \ge \frac{\frac{d}{da}E_X(T(X))}{\mathcal{I}(a)}= \mathrm{CRLB}(a)\] where $\mathcal{I}$ is the Fisher information.
This gives a measure of how consistently the estimator estimates the parameter (up to the bias). If the estimator is unbiased, the expected value of the estimator is equal to the parameter and the derivative is equal to one. In this case, the lower bound is exactly the inverse Fisher information. So, for unbiased estimators, the CRLB provides a performance bound which is only dependent on the model, and not the form of estimator itself. 

In our case, we cannot determine the expected value of the MLE nor the derivative of this analytically, so we have chosen not to estimate this. We can determine the Fisher information as follows, where it is the negative of the expected value of the second derivative of the log-likelihood. 

With respect to $a$ this is 
\begin{align*} \mathcal{I}(a)=\E_X\left(\frac{d}{d a}\log(f(X;a))\right)^2 & = \E_X\left(\sum_{i=1}^n \left (\frac{(T-X_i)}{1+a(T-X_i)} + \frac{X_i-1}{a^2}\right)\right)^2
\end{align*}
which is an analytic formula given we know the true $a$ and $T$. Similarly for $b$.

For the numerical simulations, we set the parameters $T$ and $a$ (or $b$). We assume $T$ is known. Using Matlab, we sample $n$ points from the probability distribution $f$ and then numerically determine the MLE for $a$ using {\tt fmincon}. We iterate this 10,000 times. We determine the sample mean of the MLEs of the iterates and the sample variance of the MLEs of the iterates numerically. 

In \Cref{fig:linearExtensiona} we plot the sample bias and see that it is close to zero for all sample sizes, and decreases to zero as the sample size increases. We also plot the sample variance along with the inverse Fisher information and mean squared error. Note that the Fisher information is not a lower bound in all cases, which we understand is due to the bias of the MLE and also due to regularity conditions that may break for $a$ very small, as the likelihood is undefined for $a=0$.

\begin{figure}[]
\vspace{-14pt}
    \centering
    \begin{subfigure}[b]{\textwidth}
    \includegraphics[width=14.5cm]{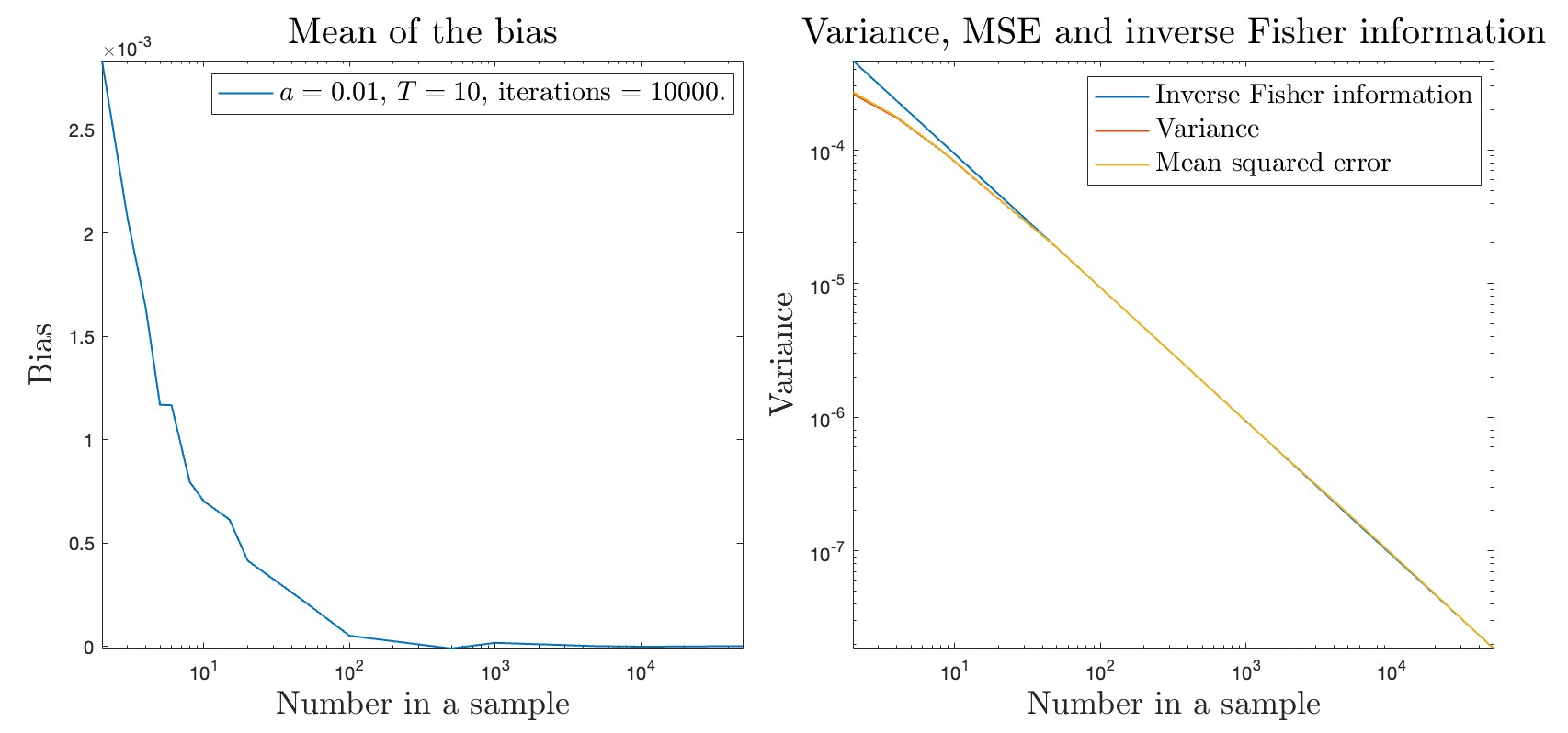}
    \end{subfigure}
    \begin{subfigure}[b]{\textwidth}
    \includegraphics[width=14.5cm]{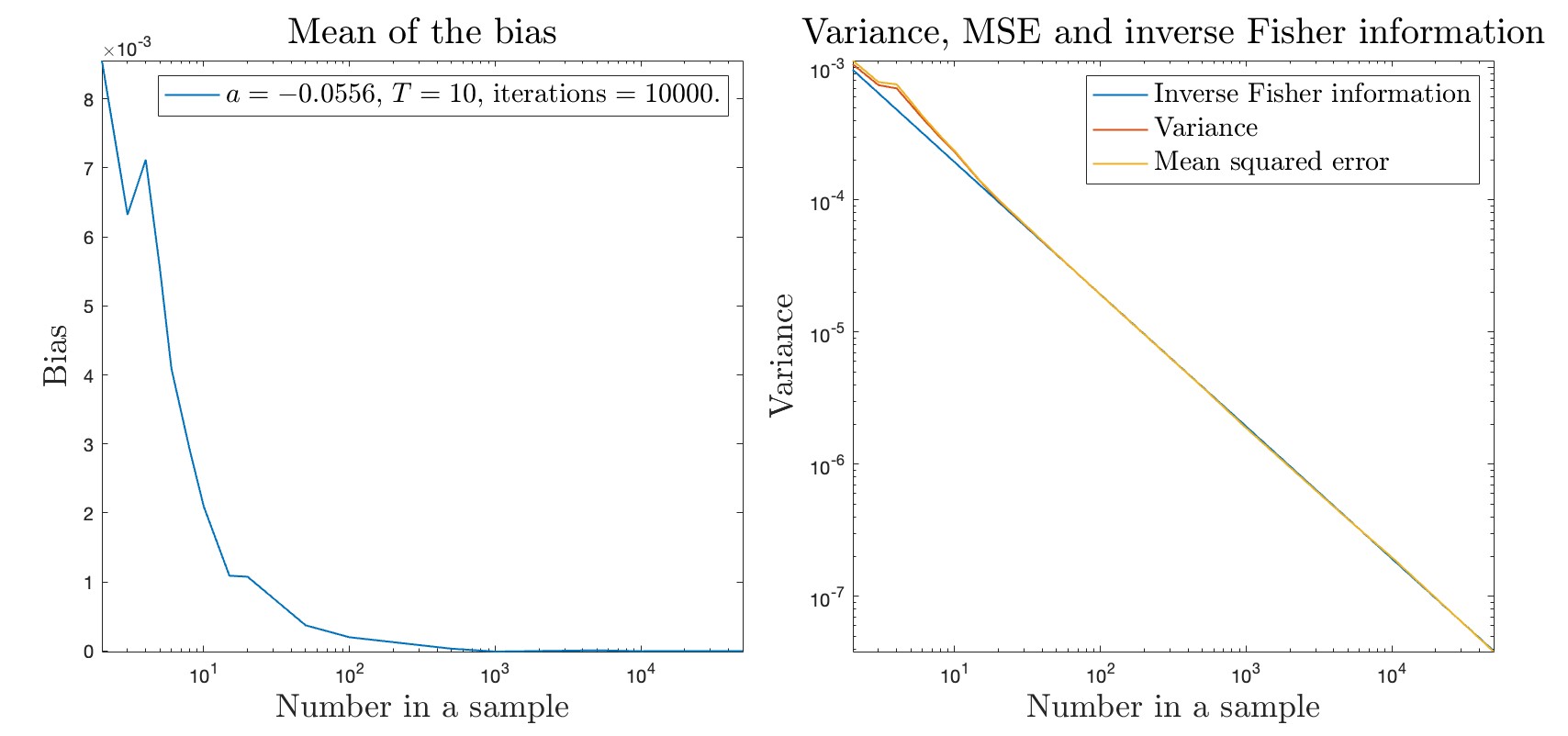}
    \end{subfigure}
    \begin{subfigure}[b]{\textwidth}
    \includegraphics[width=14.5cm]{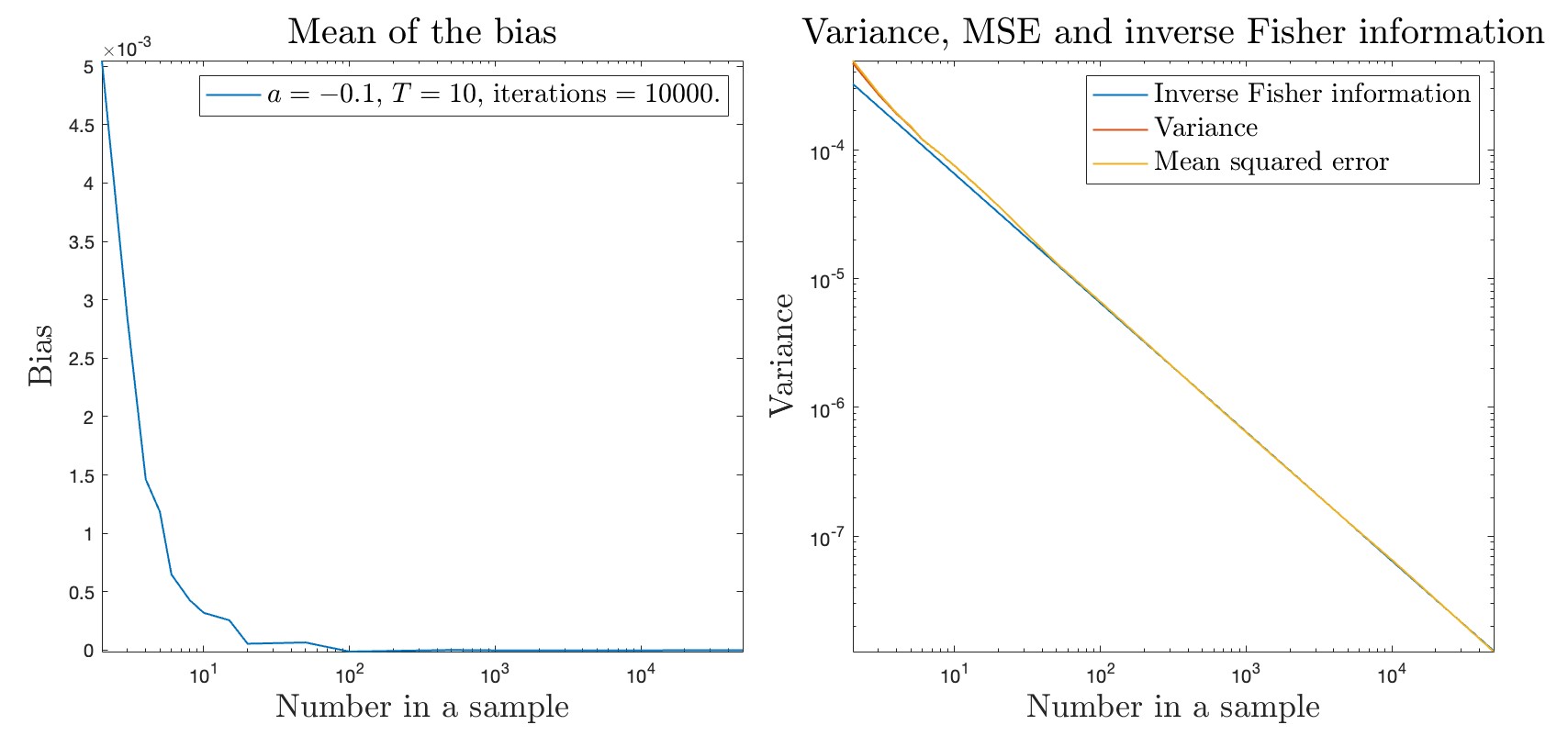}
    \end{subfigure}
    \caption{Mean of the sample bias of the MLE estimates for $a$ over all iterations, and separately the sample variance and mean squared error of the estimates with the Fisher information. Sample sizes ranged from $2$ to $50,000$, and samples were taken from the distribution $f$ with $\rho(k) = a(k-T)$ and true $a=-0.01,-0.0556,-0.1$, $T=10$. $10,000$ iterations for each sample size. We assumed $T$ was known. Note the log scales where used.}
    \label{fig:linearExtensiona}
\end{figure}
We see that small sample sizes for $a=-0.01$ result in the variance being lower than the inverse Fisher information, while for larger samples the Fisher information and the variance are almost the same. For the largest sample size of $50,000$ we have the inverse Fisher information slightly lower than the sample variance. Plotting for other $a$ as per \Cref{fig:linearExtensiona} we see that the larger $a$ is the more the Fisher information is below the sample variance, acting as a lower bound.

\subsubsection{\texorpdfstring{Unknown $T$}{Unknown T}} \label{subsec:linearunknownT}
In the previous section we assumed $T$ was known. In this section, we consider how to estimate $T$ when $T$ is unknown, as would be likely to occur in practice. Note that given a sample $x_1,\ldots, x_m$ we need at least $T\ge\max_i{x_i}$, as the support of the PMF must allow for the sample to be obtained. We numerically estimate the MLE for $T$ using a grid based method.

We sample from the linear model with $T=10$, and consider the true $a$ equal to $-0.1,-0.01,-0.0556$. We take a sample $x_1,\ldots, x_m$ of size $m$ and then for each of 
\[T =\max_i{x_i}, \max_i{x_i}+1,\max_i{x_i}+2,\ldots, T_{\max}\]
we estimate the MLEs using Matlab's {\tt fmincon}. Here we took $T_{\max} = \max_i{x_i}+200$.

We pick the value of $T$ and corresponding {\tt fmincon} output for $a$ which have the maximum likelihood for the sample as our MLEs. We iterate this 100 times and record the bias of the means of the MLEs of the iterates and the variance of the MLEs of the iterates. We do this for various different sample sizes $m$ including $m=2, 3, 4, 5, 6, 8, 10, 15, 20, 50, 100, 500,$ $1000, 5000, 10000, 50000$.  We plot in \Cref{fig:linearExtensionMLEUnknownTfora} and \Cref{linearExtensionMLEUnknownTforT} the bias of the MLEs against $m$, and we plot the variance, the inverse Fisher information and the mean squared error against $m$. We see that variance is larger when $m$ is small, but otherwise it decreases and begins to align with the inverse fisher information for $a$. 

Note that as we are undertaking a grid search for $T$, then $T$ is an integer and at large sample sizes for some $a$ we observe the MLE achieving the correct result of $T=10$ for every iteration, giving a variance of zero. This is unable to be plotted on the axes due to the log scale. Such is the case when $a=0.1$. In this case, we have also included the single parameter inverse Fisher information. Here, it is clear that the MLE estimates for $a$ tend to follow the single parameter inverse Fisher information when the sample size is large enough, which is when $T$ is correctly predicted at every iteration.  Note that such integer parameters such as $T$ are often analysed by considering how frequently they achieve the true value, and converting this into a probability --- in this case they are achieving the true value every time, so the probability would be 1 with some confidence level. We have not used this approach here.  

The inverse Fisher information is derived by assuming $T$ to be continuous and is only included as a guide. We use that the derivative of $l(a,T;x)$ for a sample $x_1,\ldots, x_m$, with respect to $T$ is 
\[ \frac{dl}{dT} = \sum_{i=1}^m \left(\frac{a}{1+a(T-x_i)} + \sum_{k=1}^{x_i-1}\frac{1}{T-k} \right) .\] 

\begin{figure}[]
\vspace{-8pt}
    \centering
    \begin{subfigure}[b]{\textwidth}
    \includegraphics[width=14cm]{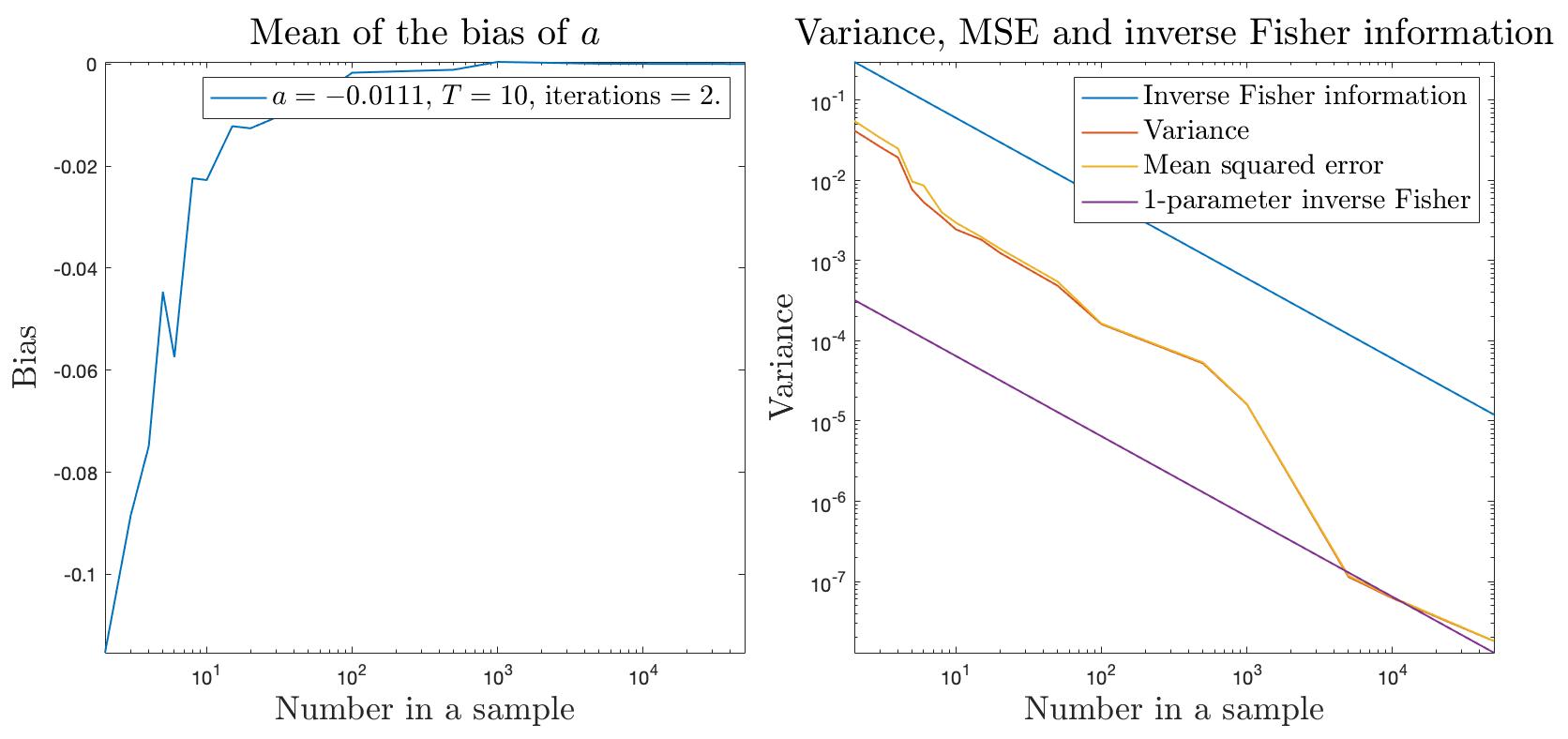}
    \end{subfigure}
    \begin{subfigure}[b]{\textwidth}
    \includegraphics[width=14cm]{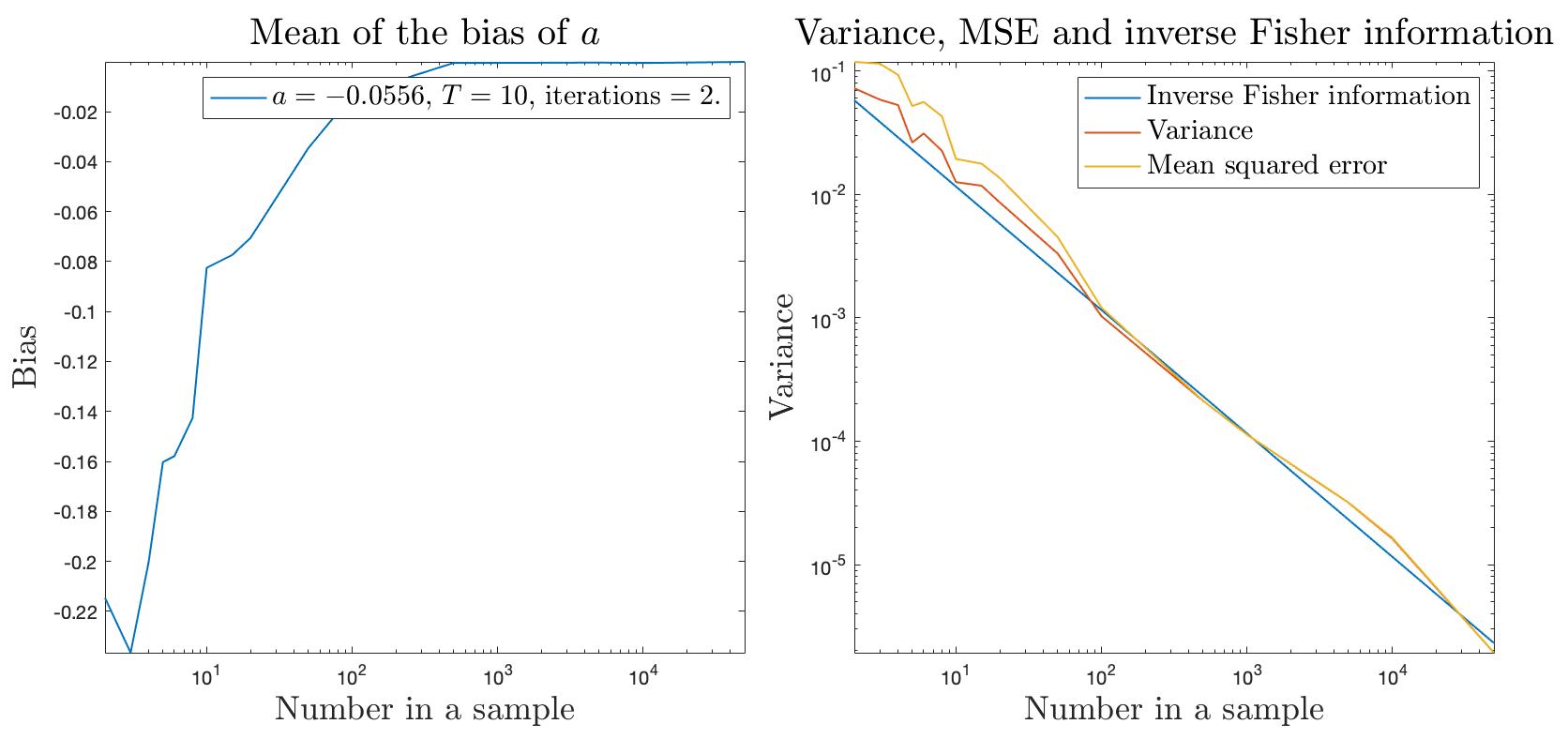}
    \end{subfigure}
    \begin{subfigure}[b]{\textwidth}
    \includegraphics[width=14cm]{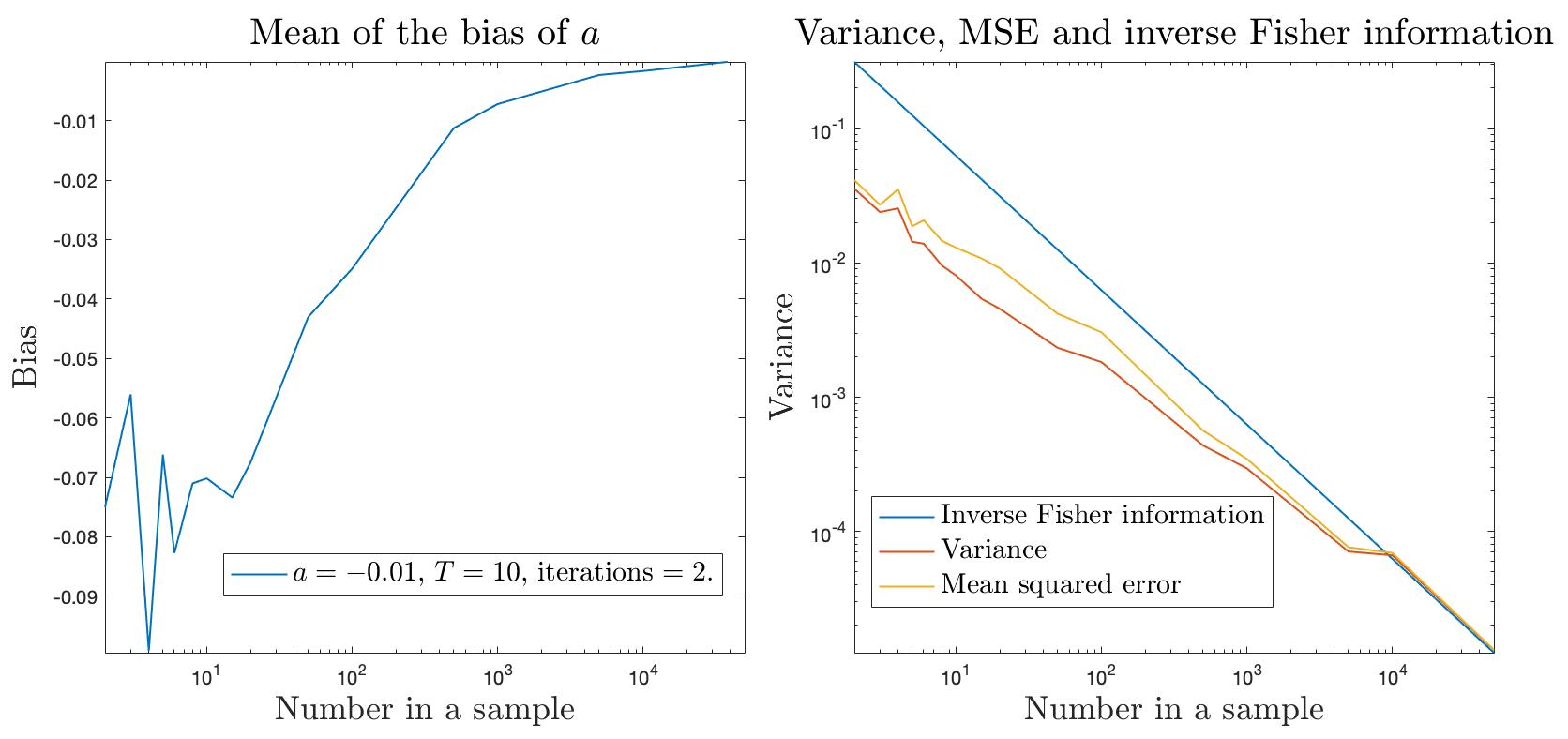}
    \end{subfigure}
    \caption{Mean of the sample bias of the MLE estimates for $a$ over all iterations, and separately the sample variance and mean squared error of the estimates against the Fisher information. Sample sizes ranged from $2$ to $50,000$, and samples were taken from the distribution $f$ with $\rho(k) = a(k-T)$ and true $a=-0.1,-0.0556,-0.01$, $T=10$ and $100$ iterations for each sample size. We assume $T$ is unknown. Note the log scales where used.}
    \label{fig:linearExtensionMLEUnknownTfora}
\end{figure}

\begin{figure}[]
\vspace{-8pt}
    \centering
    \begin{subfigure}[b]{\textwidth}
    \includegraphics[width=14cm]{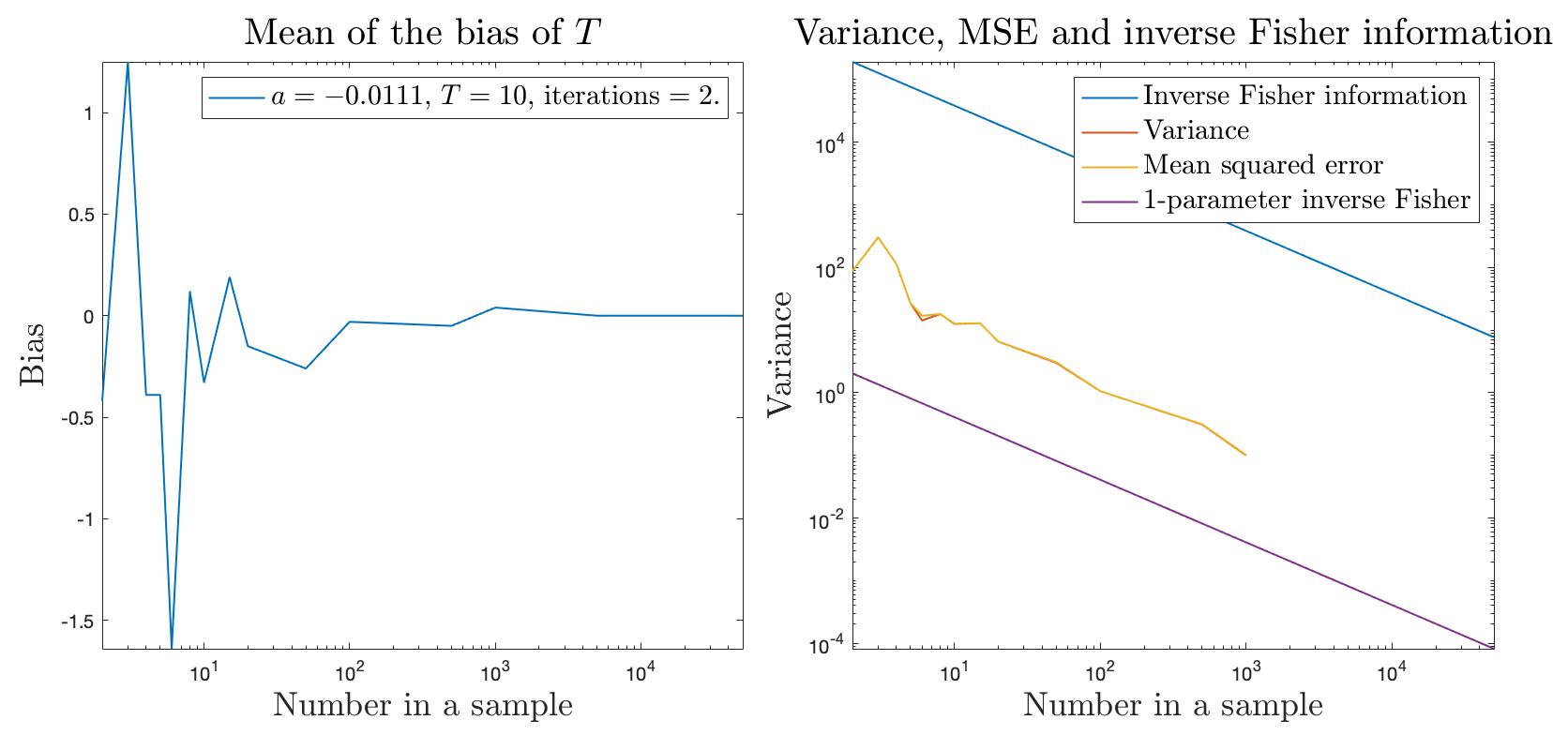}
    \end{subfigure}
    \begin{subfigure}[b]{\textwidth}
     \includegraphics[width=14cm]{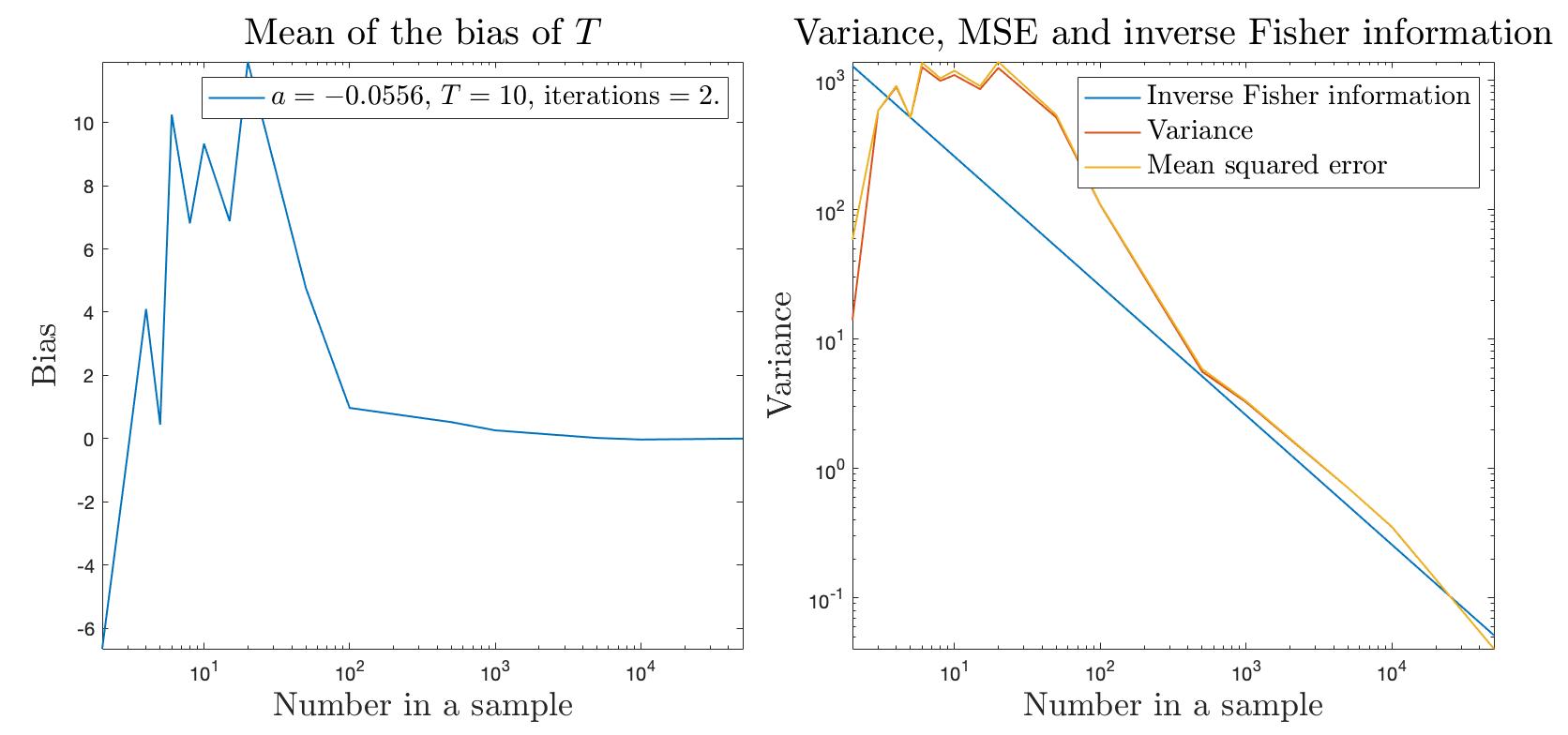}
    \end{subfigure}
    \begin{subfigure}[b]{\textwidth}
    \includegraphics[width=14cm]{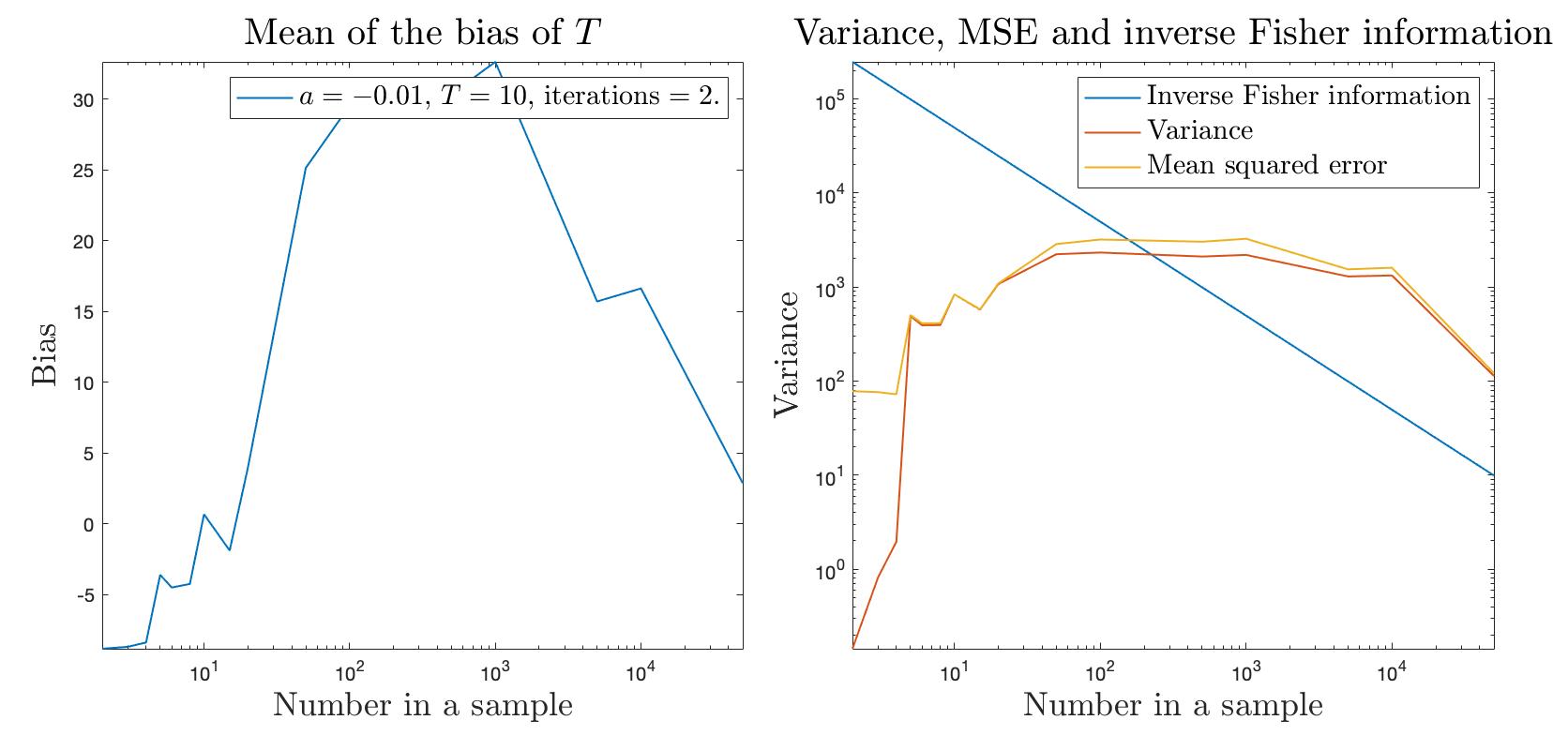}
    \end{subfigure}
    \caption{Plot of the mean of the sample bias of the MLE estimates for $T$ over all iterations, and separately the sample variance and mean squared error of the estimates against the Fisher information. Sample sizes ranged from $2$ to $50,000$, and samples were taken from the distribution $f$ with $\rho(k) = a(k-T)$ and true $a=-0.1,-0.0556,-0.01$, $T=10$ and $100$ iterations for each sample size. Note the log scales where used.}
    \label{linearExtensionMLEUnknownTforT}
\end{figure}

Note that for $a$ close to zero, then for small sample sizes the sample that contains all ones is frequently selected. In this case, the MLE is $T=1$ and $a=0$. Given the frequency of this result, the variance for T is reduced for these small sample sizes (although the bias is large).

\subsection{Simple polynomial factor model} \label{subsec:polyMLE}
We now repeat the process in \Cref{subsec:linearMLEs} for the simple polynomial factor model in \Cref{subsec:simplePolyexample}. With respect to $b$ the log-likelihood for an i.i.d. sample $x_1,\ldots, x_m$ is
\begin{align*} l(b,c, T) =& \sum_{i=1}^m\Bigg(\log\left(1+b\left(\frac{(x_i-c)^n}{(T-c)^n}-1\right)\right)+(x_i-1)\log(b) -n(x_i-1)\log(T-c) \\
&+ \sum_{t=1}^{x_i-1}\log((T-c)^n-(t-c)^n)\Bigg).
\end{align*}
The derivative with respect to $b$ is as follows 
\[ \frac{\partial l}{\partial b} = \sum_{i=1}^m\left( \frac{\frac{(x_i-c)^n}{(T-c)^n}-1}{1+b\left(\frac{(x_i-c)^n}{(T-c)^n}-1\right)} + \frac{x_i-1}{b}\right).\]
The second derivative with respect to $b$ is as follows
\[ \frac{\partial^2 l}{\partial b^2} = \sum_{i=1}^m\left( \frac{-(\frac{(x_i-c)^n}{(T-c)^n}-1)^2}{(1+b\left(\frac{(x_i-c)^n}{(T-c)^n}-1\right))^2} - \frac{x_i-1}{b^2}\right).\]

The derivative with respect to $c$ is as follows
\[ \frac{\partial l}{\partial c} = \sum_{i=1}^m \left(\frac{\frac{-nb(x_i-c)^{n-1}}{(T-c)^n} +\frac{nb(x_i-c)^n}{(T-c)^{n+1}} }{1+b\left(\frac{(x_i-c)^n}{(T-c)^n}-1\right)} + \frac{n(x_i-1)}{T-c} + \sum_{t=1}^{x_i-1}\frac{-n(T-c)^{n-1}+n(t-c)^{n-1}}{(T-c)^n-(t-c)^n}\right).\]
The second derivative with respect to $c$ is as follows 
\begin{align*} \frac{\partial^2 l}{\partial c^2} &= \sum_{i=1}^m \left(\frac{nb\left(\frac{-(n-1)(x_i-c)^{n-2}}{(T-c)^n}  -\frac{2n(x_i-c)^{n-1}}{(T-c)^{n+1}} +\frac{(n+1)(x_i-c)^n}{(T-c)^{n+2}}\right) }{1+b\left(\frac{(x_i-c)^n}{(T-c)^n}-1\right)} - \frac{\left(\frac{-nb(x_i-c)^{n-1}}{(T-c)^n} +\frac{nb(x_i-c)^n}{(T-c)^{n+1}} \right)^2}{\left(1+b\left(\frac{(x_i-c)^n}{(T-c)^n}-1\right)\right)^2} \right.\\
&+\frac{n(x_i-1)}{(T-c)^2}\\
& \left.+ \sum_{t=1}^{x_i-1}\left(\frac{n(n-1)((T-c)^{n-2}-(t-c)^{n-2})}{(T-c)^n-(t-c)^n}- \frac{(-n(T-c)^{n-1}+n(t-c)^{n-1})^2}{((T-c)^n-(t-c)^n)^2}\right)\right).\end{align*}

The mixed partial derivative with respect to $b$ and $c$ is as follows
\[ \frac{\partial^2 l}{\partial b \partial c} =  \sum_{i=1}^m \left(\frac{\frac{-n(x_i-c)^{n-1}}{(T-c)^n} +\frac{n(x_i-c)^n}{(T-c)^{n+1}} }{1+b\left(\frac{(x_i-c)^n}{(T-c)^n}-1\right)}-\frac{\frac{-nb(x_i-c)^{n-1}}{(T-c)^n} +\frac{nb(x_i-c)^n}{(T-c)^{n+1}} }{\left(1+b\left(\frac{(x_i-c)^n}{(T-c)^n}-1\right)\right)^2}\left(\frac{(x_i-c)^n}{(T-c)^n}-1\right)\right).\]

The Hessian of the likelihood is 
\[ H(b,c) = \begin{bmatrix} \frac{\partial^2 l}{\partial b^2} &\frac{\partial^2 l}{\partial b \partial c} \\\frac{\partial^2 l}{\partial b \partial c} &  \frac{\partial^2 l}{\partial c^2} \end{bmatrix}. \] Noting that $\frac{\partial^2 l}{\partial b^2}<0$ for all values of $b,c$ then at any point where the first partials are equal to zero, this point is either a maximum or a saddle point. Then the MLE must be found either at such a point that is a maximum or along the boundary. Numerical solvers can be used to find the maximum, solving for the maximum of the likelihood function in terms of $b$ and $c$ to determine the MLEs.

Similarly, in terms of $a$, then $b = -a(T-c)^n$ and $\rho(k) = a((k-c)^n - (T-c)^n)$. 
We have
\[ f(k) = \left(1+a\left((T-c)^n-(k-c)^n\right)\right)(-a)^{k-1}\prod_{t=1}^{k-1}((T-c)^n-(t-c)^n).\]

With respect to $a$ the log-likelihood for a sample $x_1,\ldots, x_m$ is
\begin{align}\label{eqn:loglikepolyac} l(a,c, T) =&\\\nonumber \sum_{i=1}^m\Bigg(\log&\left(1+a\left((T-c)^n-(x_i-c)^n\right)\right)+(x_i-1)\log(-a)  
+ \sum_{t=1}^{x_i-1}\log((T-c)^n-(t-c)^n)\Bigg).
\end{align}

The derivative with respect to $a$ is as follows 

\[ \frac{\partial l}{\partial a} = \sum_{i=1}^m\left( \frac{(T-c)^n-(x_i-c)^n}{1+a\left((T-c)^n-(x_i-c)^n\right)} + \frac{x_i-1}{a}\right).\]

The derivative with respect to $c$ is as follows

\[ \frac{\partial l}{\partial c} = \sum_{i=1}^m \left(\frac{-na(T-c)^{n-1} + na(x_i-c)^{n-1} }{1+a\left((T-c)^n-(x_i-c)^n\right)}  + \sum_{t=1}^{x_i-1}\frac{-n(T-c)^{n-1}+n(t-c)^{n-1}}{(T-c)^n-(t-c)^n}\right).\]

These derivatives are required to write the Fisher information which we use in \Cref{subsec:PolyNumAnaly}

\subsubsection{Numerical analysis}\label{subsec:PolyNumAnaly}

Similar to the linear factor model case, here we take various samples from the simple polynomial factor model in terms of a true $a$ and $c$. We consider the case $n=3$ assumed to be known and $T=10$ also assumed to be known, and $c=5$ assumed to be unknown while $c<T$ is known. We then consider the different cases of $a=0.1\min(a), 0.5\min(a),0.9\min(a)$. For a set sample size $m$ and a fixed true $a$, we take samples from the true distribution and numerically determine the MLEs of $a$ and $c$ using the Matlab command {\tt fmincon}. We iterate this 10,000 times and then take the mean of the bias and variance of these MLEs over the iterations.  We consider $m=2, 3, 4, 5, 6, 8, 10, 15, 20, 50, 100, 500, 1000, 5000, 10000, 50000$.

In \Cref{fig:PolyGeomExtensionMLEa} and \Cref{fig:PolyGeomExtensionMLEc} the mean of the bias against the sample sizes for $a$ and $c$ respectively. We see that the bias is large when $a$ is small,  however it approaches zero as $a$ gets larger. 

In \Cref{fig:PolyGeomExtensionMLEa} and \Cref{fig:PolyGeomExtensionMLEc} we also plot the variance, mean squared error (MSE) and the Fisher information against the sample sizes, noting that in general where there is no bias the Fisher information is a lower bound for the variance of an estimator. In our case, it is clear that there is a bias of the estimator, however the Fisher information remains informative, with the variance and MSE tracking closely to the Fisher information as the sample size $m$ increases. We also plot a corrected variance that considers the variance with the outliers removed.

In general, we see that when the true $a$ is close to zero, the sample variance for $a$ is lower than the inverse Fisher information for small sample sizes. As $a$ becomes larger in absolute value, the sample variance starts to become larger than the inverse Fisher information for small sample sizes. Similar results hold for the MLEs of $c$.

All results in this section used the following options for {\tt fmincon}: {\tt options =} \\{\tt \allowbreak optimoptions("fmincon","Algorithm","interior-point","EnableFeasibilityMode",} \\{\tt true,"SubproblemAlgorithm","cg","SpecifyObjectiveGradient",true,}\\ {\tt "StepTolerance",1e12,"ConstraintTolerance",1e15).} We supplied the gradient of the likelihood using the derivatives of the likelihood given at the start of \Cref{subsec:polyMLE}.

\begin{figure}[]
\vspace{-20pt}
    \centering
    \begin{subfigure}[b]{\textwidth}
    \includegraphics[width=14cm]{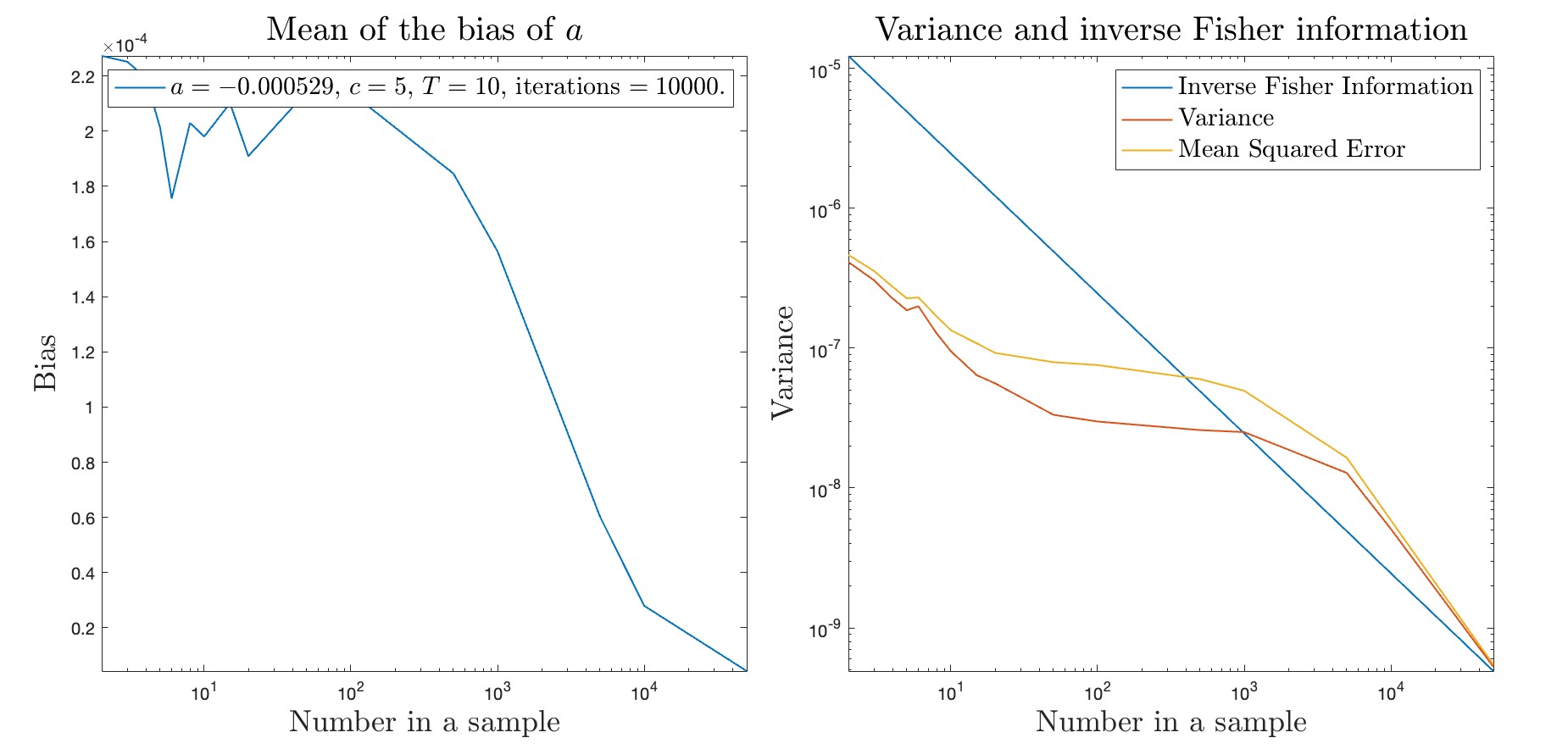}
    \end{subfigure}
    \begin{subfigure}[b]{\textwidth}
    \includegraphics[width=14cm]{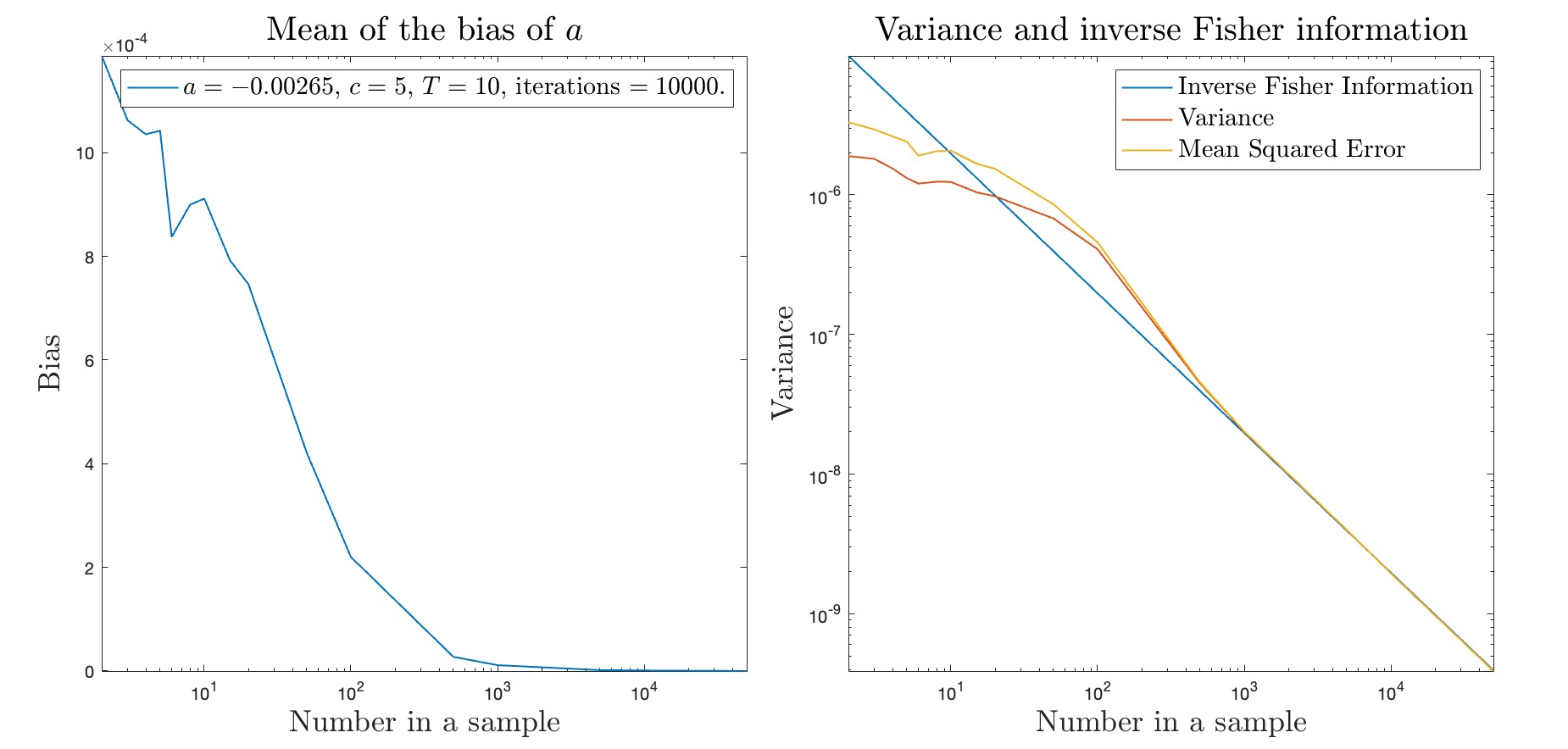}
    \end{subfigure}
    \begin{subfigure}[b]{\textwidth}
    \includegraphics[width=14cm]{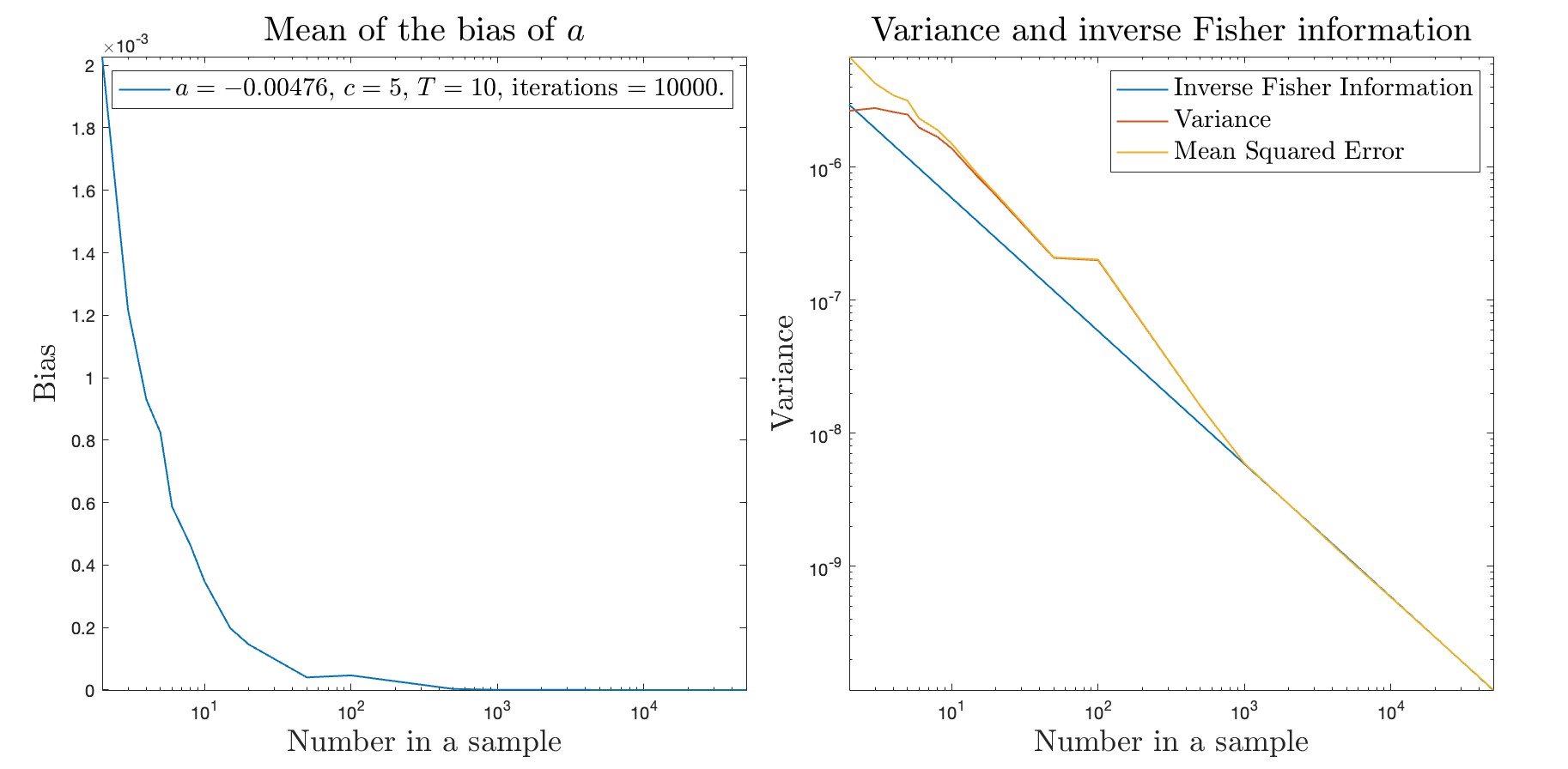}
    \end{subfigure}
    \caption{Mean of the sample bias, sample variance and Fisher information for the MLE estimator for $a$ for various sample sizes from distribution $f$ with $\rho(k) = a((k-c)^n-(T-c)^n)$ with $n=3$, $a=0.1\min(a),0.5\min(a),0.9\min(a)$, $T=10$ and $10,000$ iterations taken for each sample size. Note the log scales where used. Assumed $T$ known.}
    \label{fig:PolyGeomExtensionMLEa}
\end{figure}

\begin{figure}[]
\vspace{-20pt}
    \centering
    \begin{subfigure}[b]{\textwidth}
    \includegraphics[width=14cm]{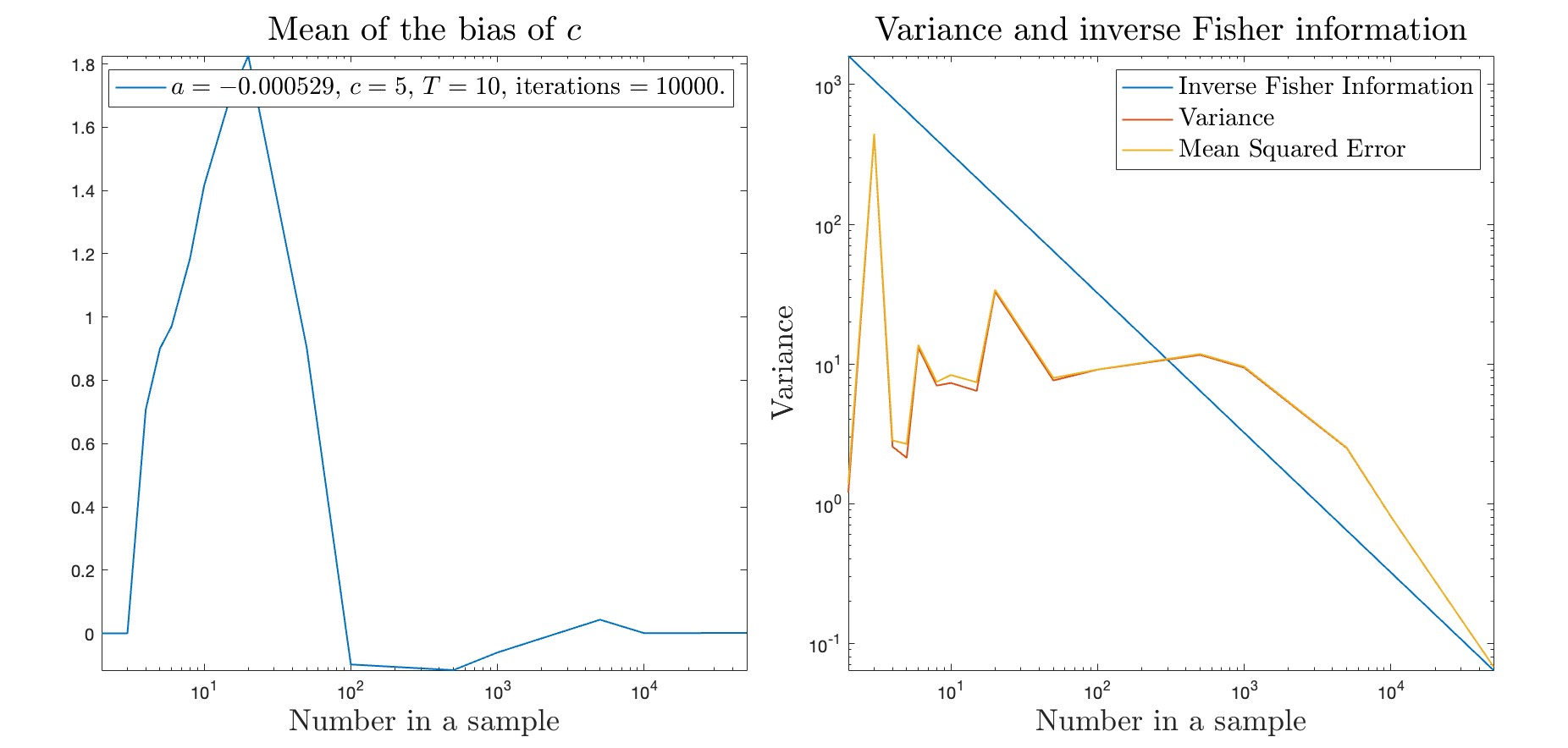}
    \end{subfigure}
    \begin{subfigure}[b]{\textwidth}
        \includegraphics[width=14cm]{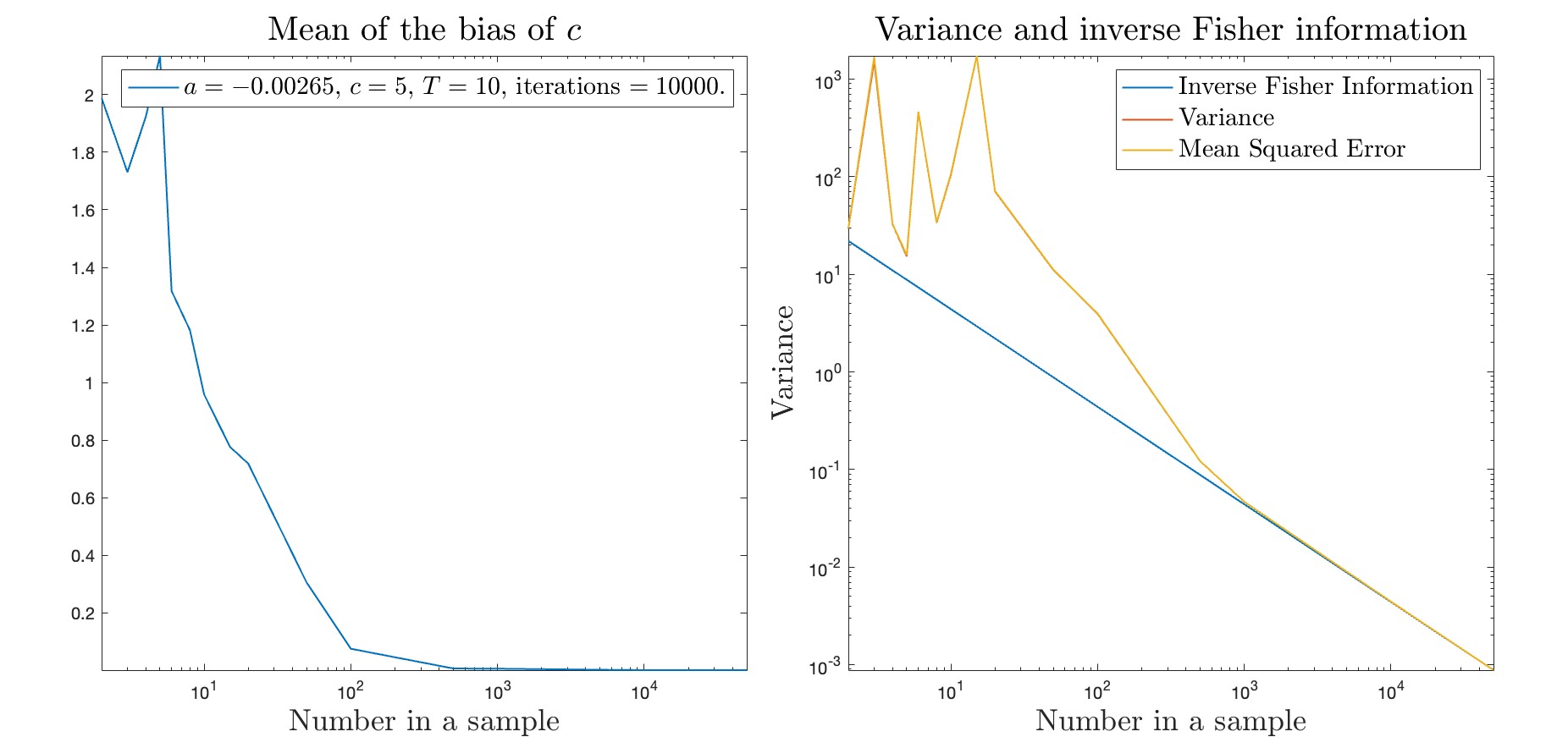}
    \end{subfigure}
    \begin{subfigure}[b]{\textwidth}
        \includegraphics[width=14cm]{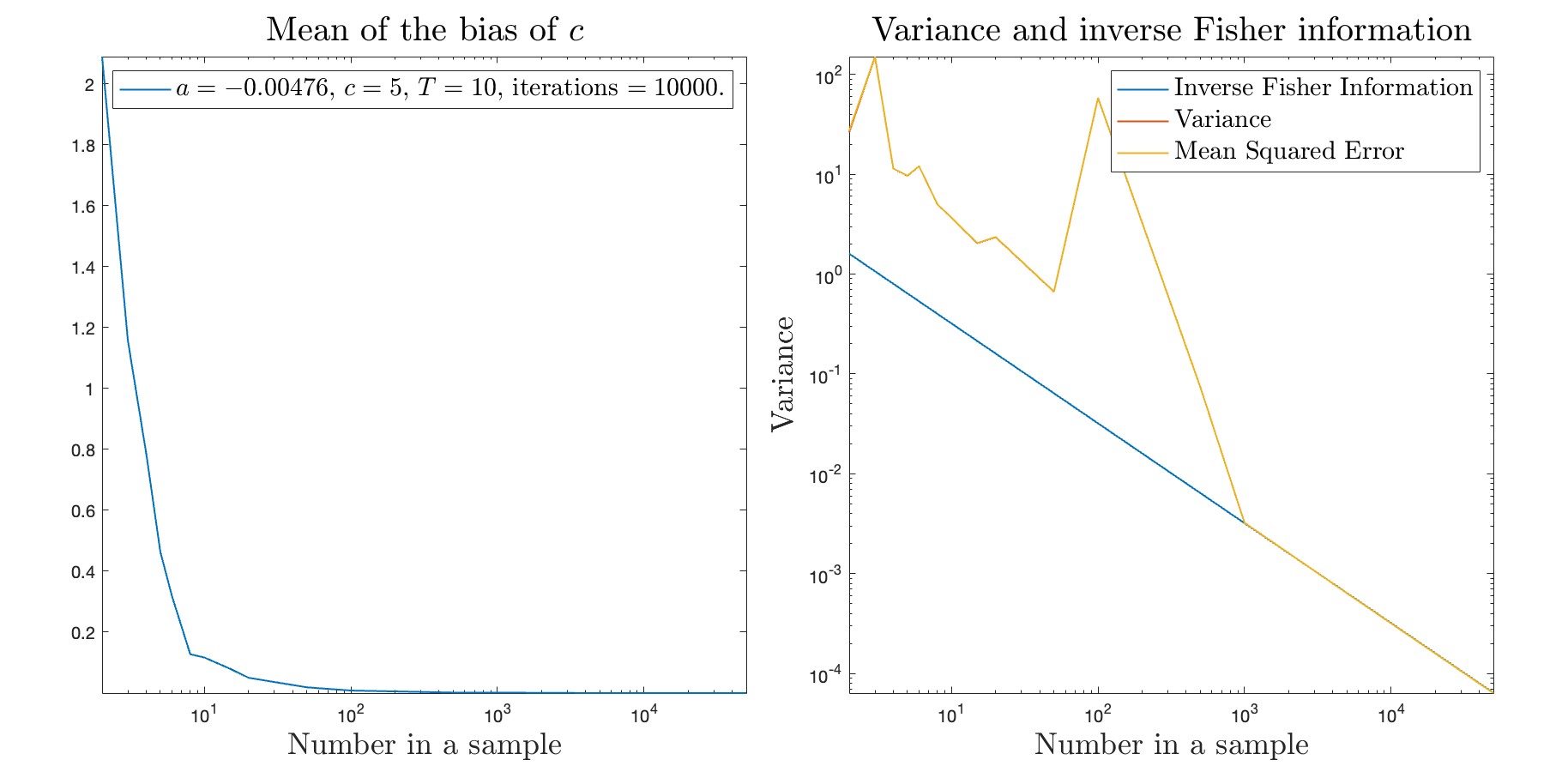}
    \end{subfigure}
    \caption{Mean of the sample bias, sample variance and Fisher information for the MLE estimator for $c$ for various sample sizes from distribution $f$ with $\rho(k) = a((k-c)^n-(T-c)^n)$, $n=3$, $a=0.1\min(a),0.5\min(a),0.9\min(a)$, $T=10$ and $10,000$ iterations taken for each sample size. Note the log scales where used. Assumed $T$ known.}
    \label{fig:PolyGeomExtensionMLEc}
\end{figure}

\subsubsection{\texorpdfstring{Unknown $T$}{Unknown T}}
Similar to \Cref{subsec:linearunknownT} we can apply a grid search to determine the MLE for $T$ from the data, along with $b$ and $c$. We do not do this with simulations here, however we do apply this method to experimental data as in \Cref{sec:experidata}. Note that as usual we require that $T\ge \max_i\{x_i\}$ for a sample $x_1,\ldots, x_m$. Helpful formulae for the derivative of the log-likelihoods with respect to $T$ are given below, noting that $T$ is assumed to be integer so these are only intended as a guide to be used for the inverse Fisher information. 

We leave further simulations for future work.

In terms of $b$, the derivative with respect to $T$ is as follows
\[ \frac{\partial l}{\partial T} = \sum_{i=1}^m \left(\frac{\frac{-(x_i-c)^n}{(T-c)^{n+1}} }{1+b\left(\frac{(x_i-c)^n}{(T-c)^n}-1\right)} - \frac{n(x_i-1)}{T-c} + \sum_{t=1}^{x_i-1}\frac{n(T-c)^{n-1}}{(T-c)^n-(t-c)^n}\right).\]

In terms of $a$, the derivative with respect to $T$ is as follows
\[ \frac{\partial l}{\partial T} = \sum_{i=1}^m \left(\frac{na(T-c)^{n-1} -na(x_i-c)^{n-1}}{1+a\left((T-c)^n-(x_i-c)^n\right)} + \sum_{t=1}^{x_i-1}\frac{n(T-c)^{n-1}}{(T-c)^n-(t-c)^n}\right).\]

\subsection{Discussion}
In this section we further discuss the results of \Cref{subsec:linearMLEs} and \Cref{subsec:polyMLE}. We note that for the linear factor model, overall we observe the variances of MLEs for the parameters when $T$ is known closely follow the inverse Fisher information in all cases, and we only observe it deviating below the inverse Fisher information when the $a$ is particularly small and the sample size is small. 

However, when estimating both $a$ and $T$ using the grid search method, we do see deviations of the variance away from the inverse Fisher information for small samples sizes, and in the case where $T$ becomes correctly identified in every sample. In this latter case, the variance for the MLE for $T$ becomes zero, while the variance for $a$ starts to track the one-parameter inverse Fisher information instead, which shows the limitations for using the inverse Fisher information (or CRLB) for integer variables. Further investigation of the PMFs of the distributions produced from these parameters shows little visual difference. In \Cref{appendix:unknownT} the second author discusses this further. 

In the simple polynomial factor model with $n=3$ we see similar behaviour although slightly more extreme than the 1-parameter case, when estimating $a$. We observe the variance for $a$ tracking closer to the inverse Fisher information for large sample size, however we see some deviation away from the inverse Fisher information for small sample sizes. The behaviour for the MLE estimates of $c$ is also similar, although deviates further from the inverse Fisher information in small sample sizes. We suspect that some of this may result from the programming behind the Matlab function {\tt fmincon}, and we discuss this further in \Cref{appn: NumericalPolyIssues}. Additionally, the true CRLB requires knowing the square of the derivative of the bias of the variable, which we have not estimated --- we note that the bias is larger in these cases and this may contribute to the differences observed. Further examination we leave for future work.

\section{Applying to experimental data} \label{sec:experidata}
In this section we apply the linear and polynomial factor models to experimental data. This experimental data consists of the consecutive times a number of participants took to complete three screen based tasks. We consider these times to be the sojourn times and each screen based task to be a separate state. A key objective of our analysis is to charaterise the distributions of the times taken. The data is available in \Cref{appn:ExperimentalData}. 

After collecting the data and plotting histograms of the data, we do not see distributions that appear to be geometric. Instead, the distributions more closely match the linear and simple polynomial factor models studied in the previous sections. We can determine the empirical PMF and cummulative density functions (CDFs) through bin counting, and then determine the empirical $\rho(k)$ using the formulae in \Cref{lem:relationship}. Plotting this in \Cref{fig:RhoPlotsExperi} we see immediately that the plot is not constant, so not indicating a geometric distribution, and is non-linear, which indicates a simple polynomial factor model may be appropriate.

\begin{figure}[]
    \centering
    \begin{subfigure}[b]{0.5\textwidth}\centering
            \includegraphics[width=1\textwidth]{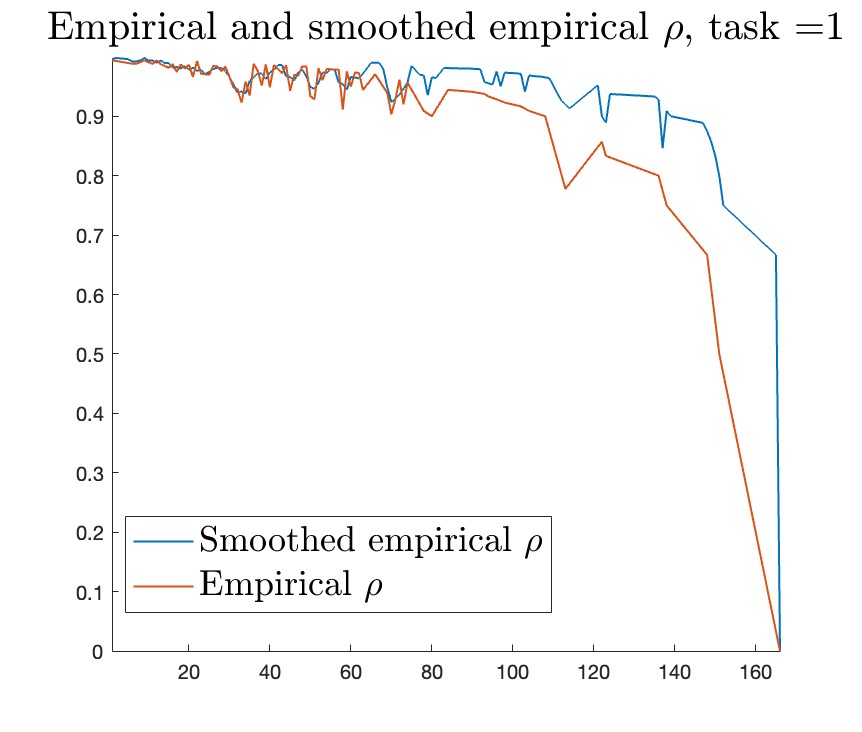}
    \end{subfigure}\hfill
        \begin{subfigure}[b]{0.5\textwidth}\centering
            \includegraphics[width=1\textwidth]{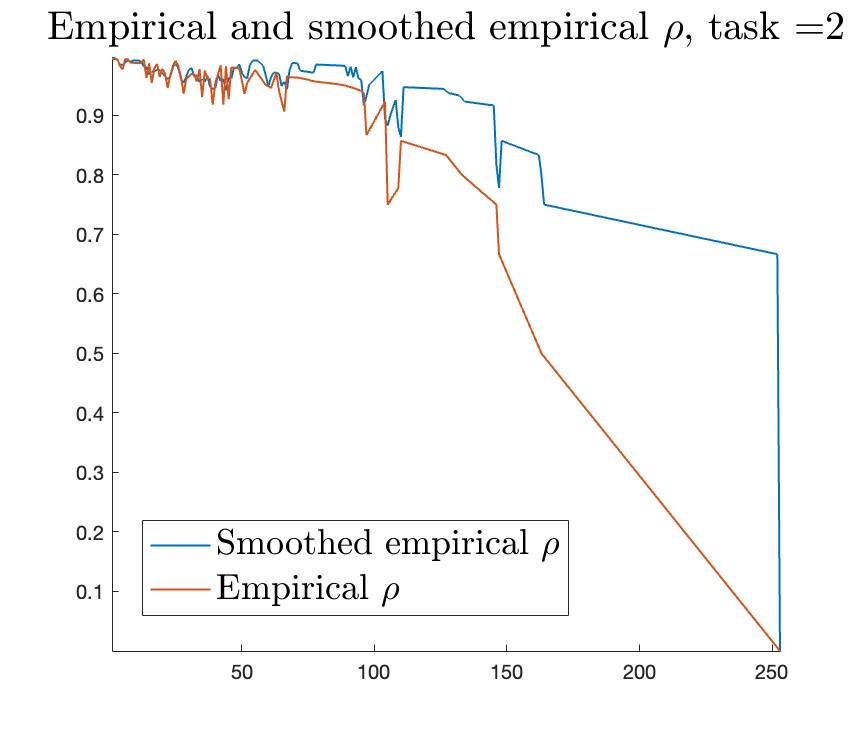}
    \end{subfigure}
        \begin{subfigure}[b]{0.5\textwidth}\centering
            \includegraphics[width=1\textwidth]{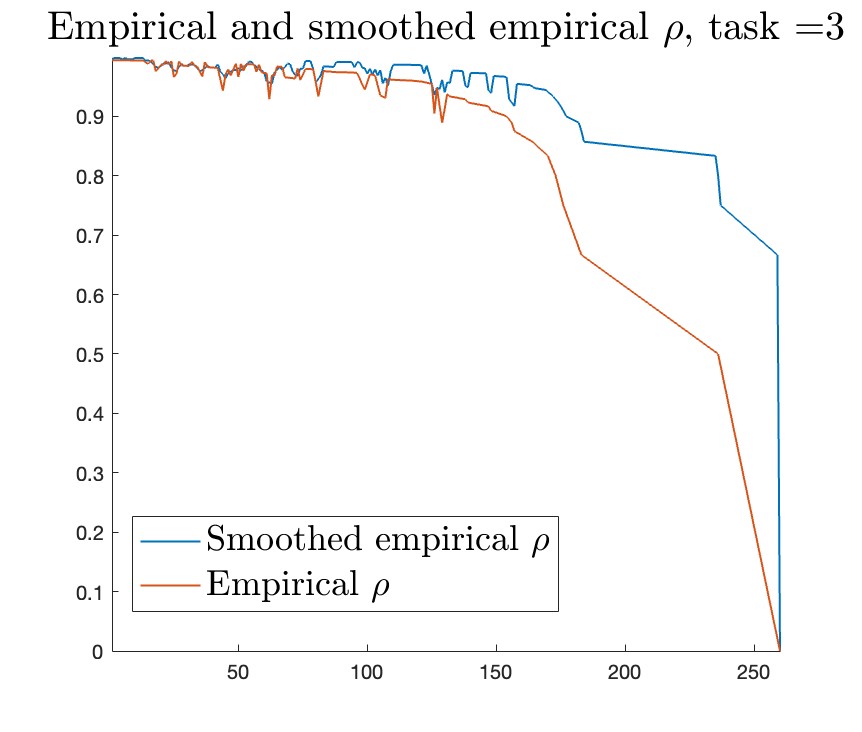}
    \end{subfigure}
        \caption{Empirical $\rho$ of the data from tasks 1, 2 and 3. Only values not equal to 1 were plotted. A smoothed version was also plotted, where the PMF was smoothed so that each bin was the average of it and the two bins either side of it. As each task displays non-constant and non-linear for each task $\rho$ this suggests that a simple polynomial factor model may be a good fit.}
        \label{fig:RhoPlotsExperi}
\end{figure}

Below, we plot histograms of the consecutive time spent on task for each task. We apply the linear factor model as in \Cref{fig:PDFExperi} and the simple polynomial factor model with $n=3,c<T$ as in \Cref{fig:PDFexperipoly} to each task. We note that the number of data points for each task were 178, 187, 187 respectively.

A first observation of the data is that the data do not have minimum of 1 for each task, instead the minimums are 4,3 and 12 respectively. This makes sense as the tasks had a certain level of complexity that would ensure more than one time unit would need to be counted for any attempt at the task.

Given this, we extended our model to consider distributions of the linear and polynomial factor models that were also shifted to start at an integer $t_1>1$ with $t_1$ less than or equal to the minimum of the sample. We compared the different possible shifts by comparing the value of the likelihood for the MLEs. Whichever one had the largest likelihood we considered to be the estimated MLE. A table of the results is included for each analysis, see \Cref{tab:linearExperi200} and \Cref{tab:polyExperi200}. In this case we used a grid search to determine $T$ by searching up to 200 integers higher than the minimum possible value for $T$, i.e. up to 200 higher than the max of the data.

\begin{table}
\begin{center}
\begin{tabular}{ |c|c|c|c|c|c| } 
 \hline
 Task &Maximum &Minimum & Shift & MLE $a$ & MLE $T$  \\ \hline
 1 &164& 4 & 3 & $-0.002786$& $360$\\ 
 2 &201& 3 & 2 & $-0.002519$ & $398$\\ 
  3 &221& 12 & 11 & $-0.002451$& $409 $ \\ 
 \hline
\end{tabular} 
\caption{This table details the output from our optimisation process for each task from applying the simple linear factor model. We have also included the maximum and the minimum of the data for that task, as well as the shift of the distribution that gave the largest likelihood. Here $T$ was allowed to range up to 200 integers higher than the maximum data value after the shift.}
\label{tab:linearExperi200}
\end{center}
\end{table}

We use PMF, CDF and QQ-plots to compare the empirical distribution with the MLE estimates, as in \Cref{fig:CDFexperi} and \Cref{fig:CDFexperipoly}. Here the QQ-plots plot the quantiles of the empirical distribution against the distribution using the MLEs estimates. To indicate a good fit we expect that this plot should be very close to the straight line drawn on the graph, which is the quantiles of the MLE against itself. 

We note that the simple linear factor model was not a good fit when comparing to the plots of the PMFs as well as to the QQ-plots and CDF plots. We observed this as well in the optimisation procedure that the $T$ value was at the far end of the range we were searching. We increased the grid search to include $T$ values up to 1000 higher than the maximum data point, which greatly increased the visual fit of the data, however the optimisation procedure still picked the largest value of $T$ possible. This suggests that this particular linear factor model may not be the right fit for this data. A table of results for this is include in \Cref{tab:linearExperi1000}.

\begin{table}[]
\begin{center}
\begin{tabular}{ |c|c|c|c|c|c| } 
 \hline
 Task &Maximum &Minimum & Shift & MLE $a$ & MLE $T$  \\ \hline
 1 &164& 4 & 3 & $-8.6204\times 10^{-4}$& $1160$\\ 
 2 &201& 3 & 2 & $-8.3542\times 10^{-4}$ & $1198$\\ 
  3 &221& 12 & 11 & $-9.1503\times 10^{-8}$& $409 $ \\ 
 \hline
\end{tabular}
\caption{This table details the output from our optimisation process for each task from applying the simple linear factor model. We have also included the maximum and the minimum of the data for that task, as well as the shift of the distribution that gave the largest likelihood. Here $T$ was allowed to range up to 1000 integers higher than the maximum data value after the shift.}
\label{tab:linearExperi1000}
\end{center}
\end{table}

\begin{table}[]
\begin{center}
\begin{tabular}{ |c|c|c|c|c|c|c| } 
 \hline
 Task &Maximum &Minimum & Shift & MLE $a$ & MLE $c$ & MLE $T$  \\ \hline
 1 &164& 4 & 3 & $-9.1503\times 10^{-8}$& $74.8998$& $294$\\ 
 2 &201& 3 & 3 & $-4.2160\times 10^{-8}$& $92.9539$  & $377$\\ 
  3 &221& 12 & 9 & $-3.6970\times 10^{-8}$& $89.4028$ & $387 $ \\ 
 \hline
\end{tabular}
\caption{This table details the output from our optimisation process for each task from applying the simple polynomial factor model with $n=3$. We have also included the maximum and the minimum of the data for that task, as well as the shift of the distribution that gave the largest likelihood. Here $T$ was allowed to range up to 200 integers higher than the maximum data value after the shift. Note however that the MLE for T was less than the maximum minus the shift.}
\label{tab:polyExperi200}
\end{center}
\end{table}

\begin{figure}[]
    \centering
    \begin{subfigure}[b]{0.5\textwidth}\centering
            \includegraphics[width=1.1\textwidth]{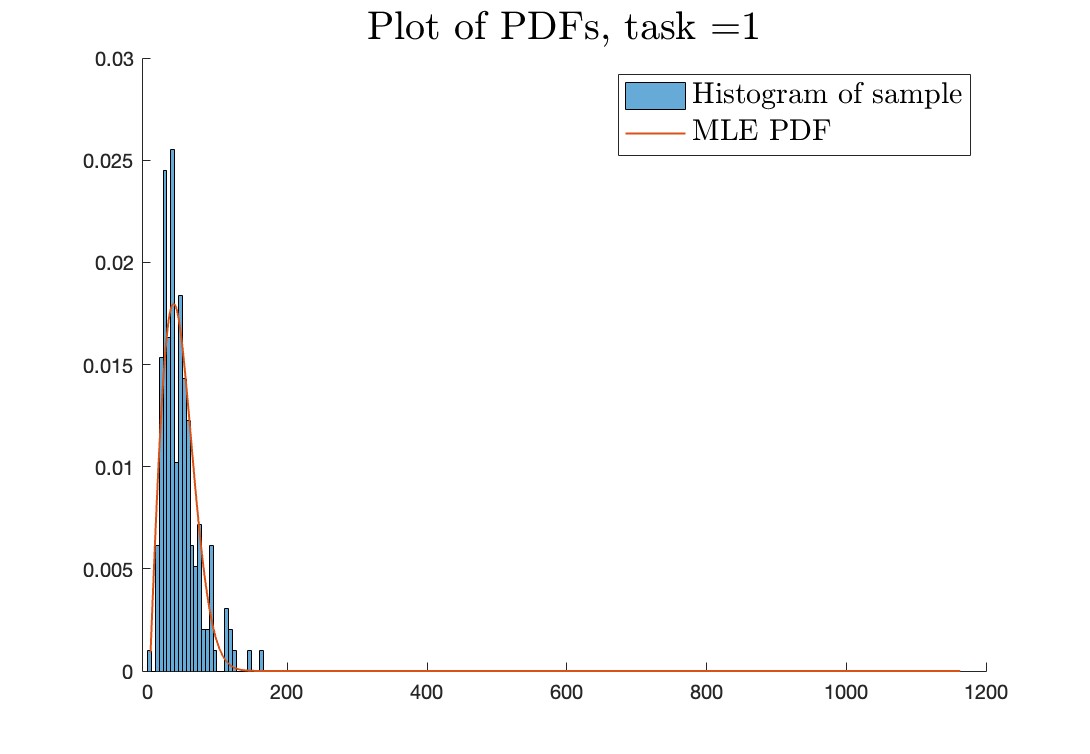}
    \end{subfigure}\hfill
        \begin{subfigure}[b]{0.5\textwidth}\centering
            \includegraphics[width=1.1\textwidth]{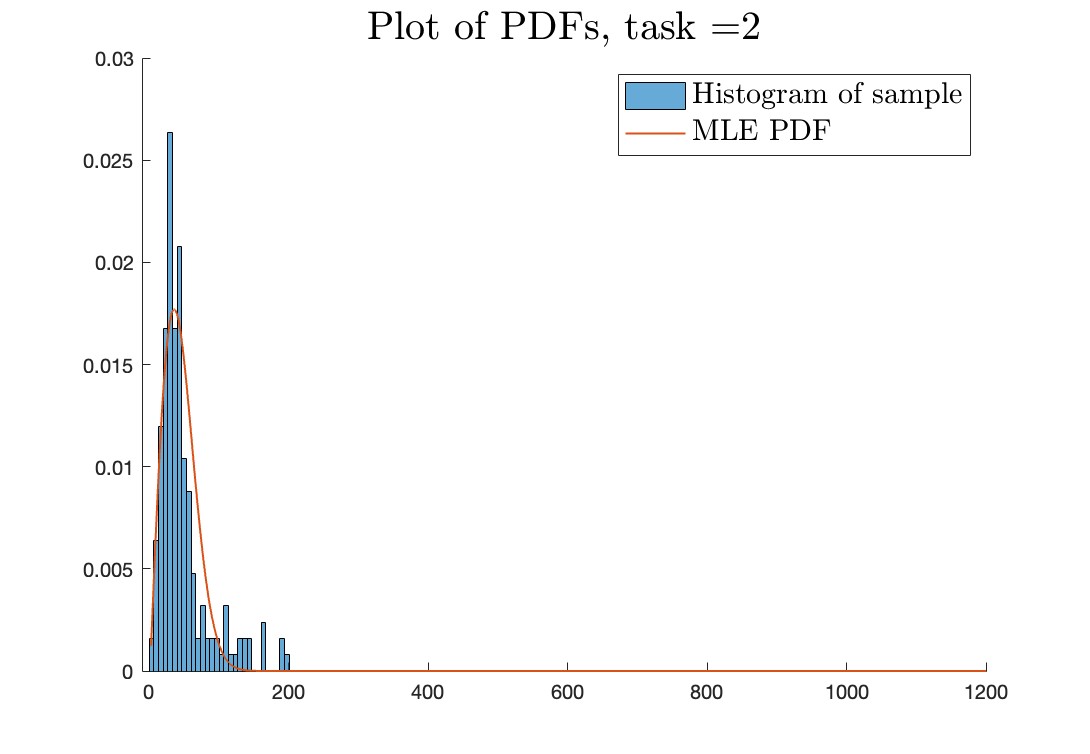}
    \end{subfigure}
        \begin{subfigure}[b]{0.5\textwidth}\centering
            \includegraphics[width=1.1\textwidth]{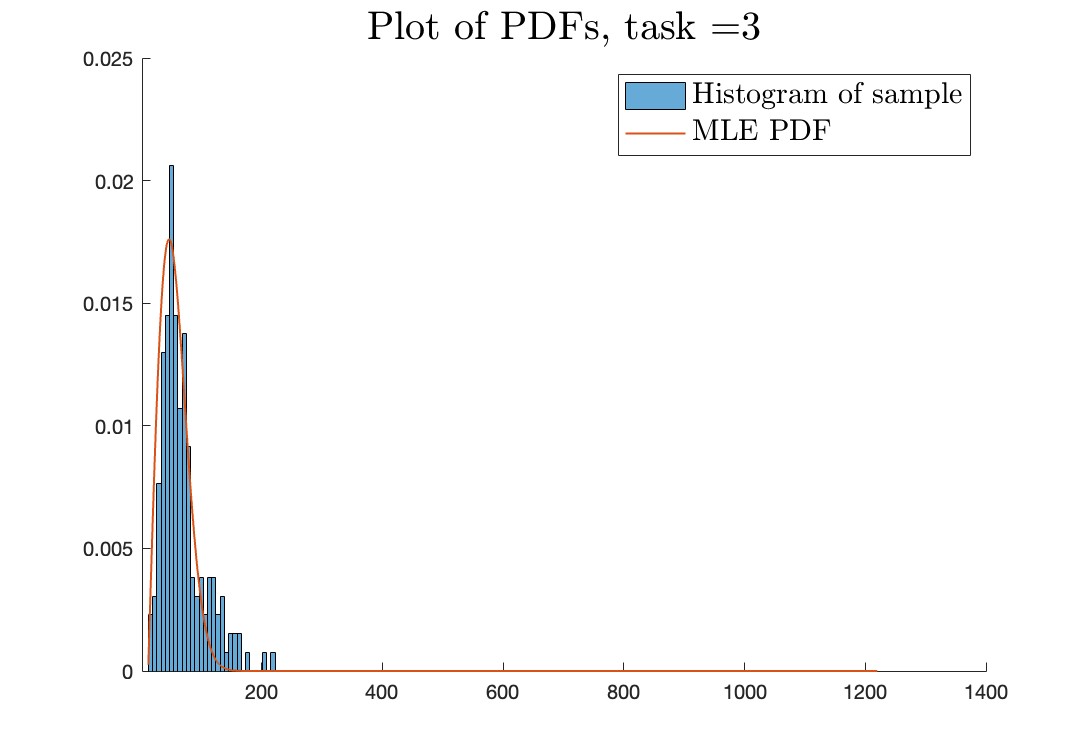}
    \end{subfigure}
        \caption{Histogram of the data from tasks 1, 2 and 3, along with the MLE estimated PMFs using the linear factor model. Here $T$ was allowed to range up to 1000 integers higher than the maximum data value after the shift.}
        \label{fig:PDFExperi}
\end{figure}

\begin{figure}[]
    \centering
        \begin{subfigure}[b]{\textwidth}\centering
    \includegraphics[width=14cm]{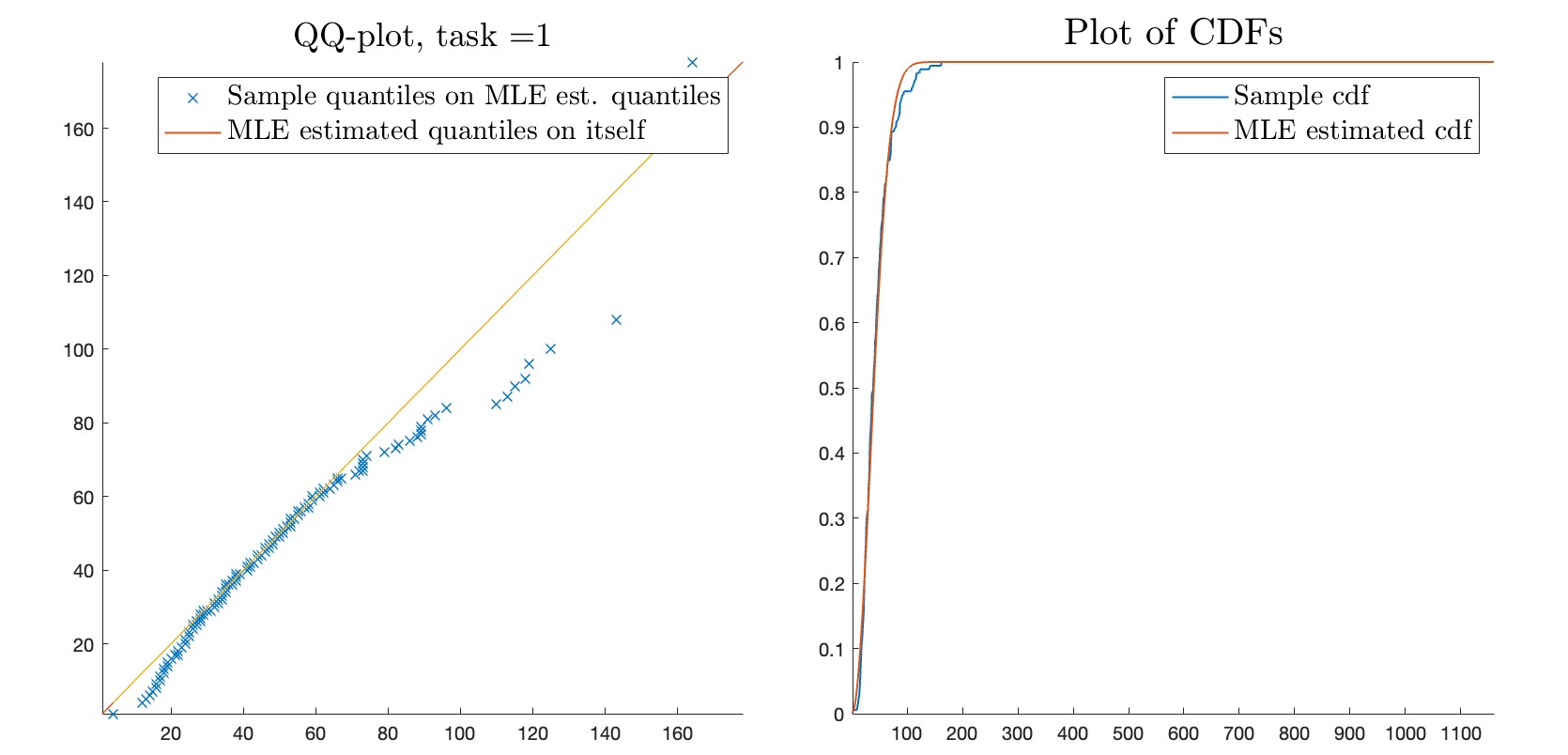}
    \end{subfigure}
        \begin{subfigure}[b]{\textwidth}\centering
    \includegraphics[width=14cm]{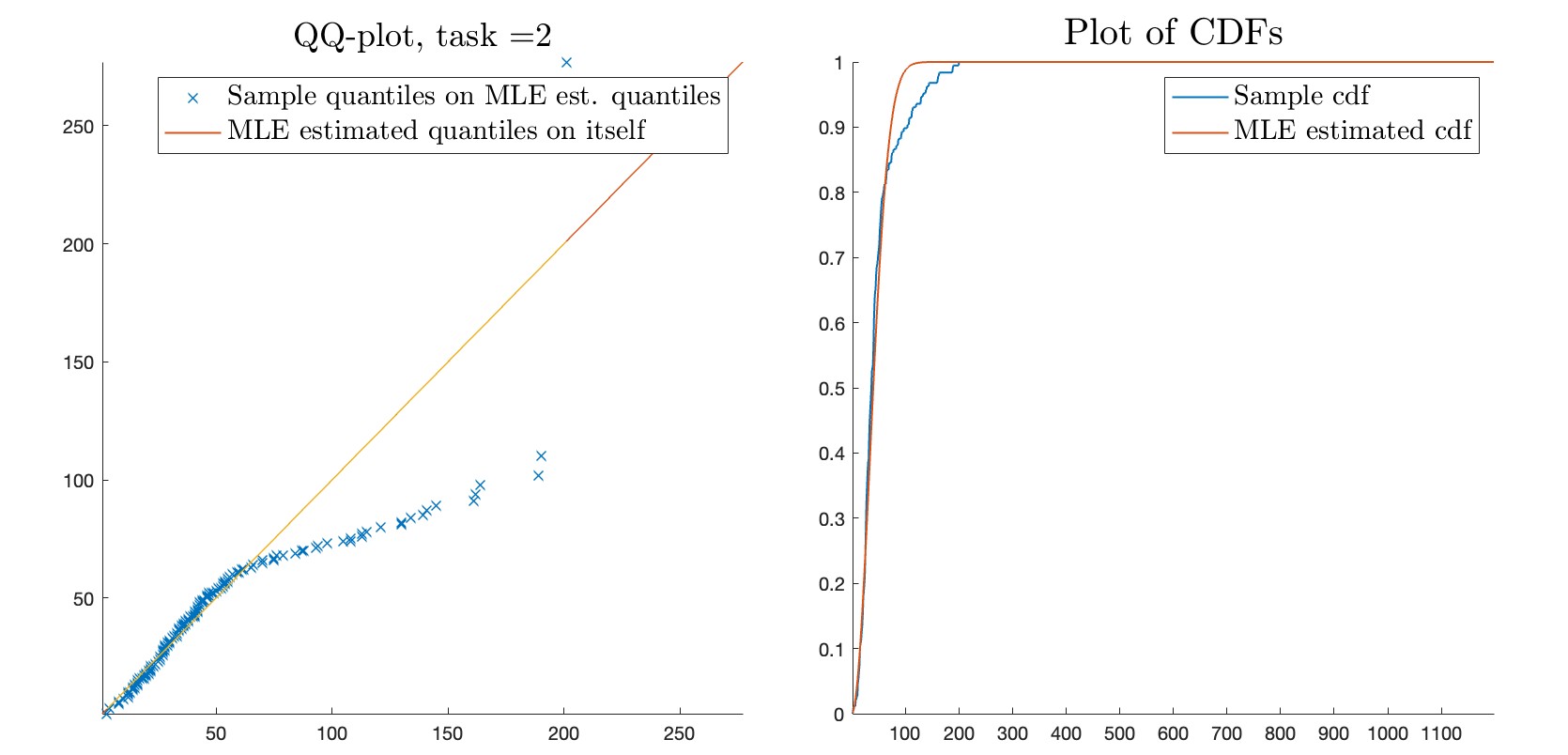}
    \end{subfigure}
        \begin{subfigure}[b]{\textwidth}\centering
    \includegraphics[width=14cm]{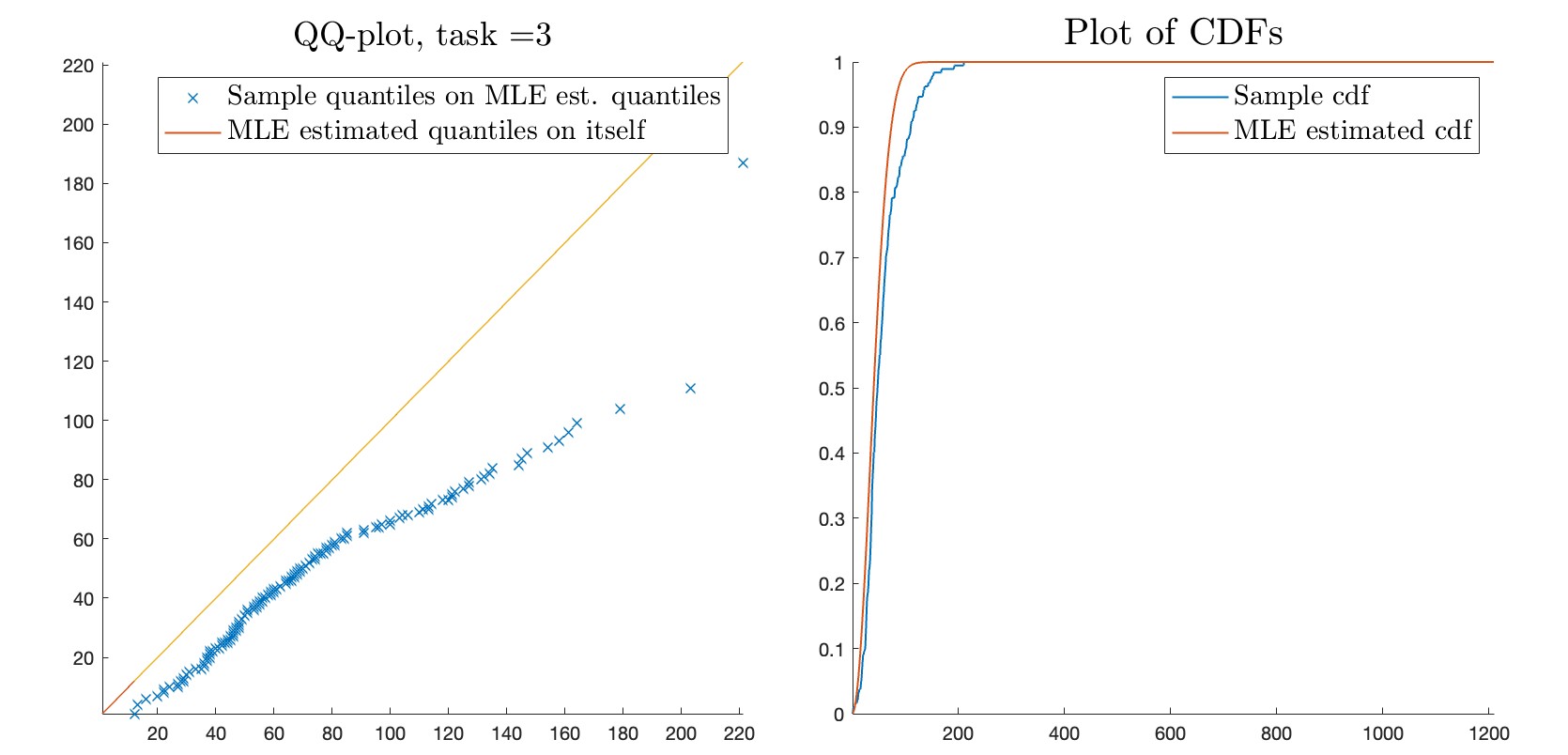}
    \end{subfigure}
        \caption{QQ plots of the quantiles of the data from task 1,2, and 3 against the quantiles of the MLE estimated distribution, along with plots of the CDF of the sample against the CDF of the MLE estimated distribution. This used the linear factor model. Here $T$ was allowed to range up to 1000 integers higher than the maximum data value after the shift.}
        \label{fig:CDFexperi}
\end{figure}

\begin{figure}[]
    \centering
    \begin{subfigure}[b]{0.5\textwidth}\centering
            \includegraphics[width=1.1\textwidth]{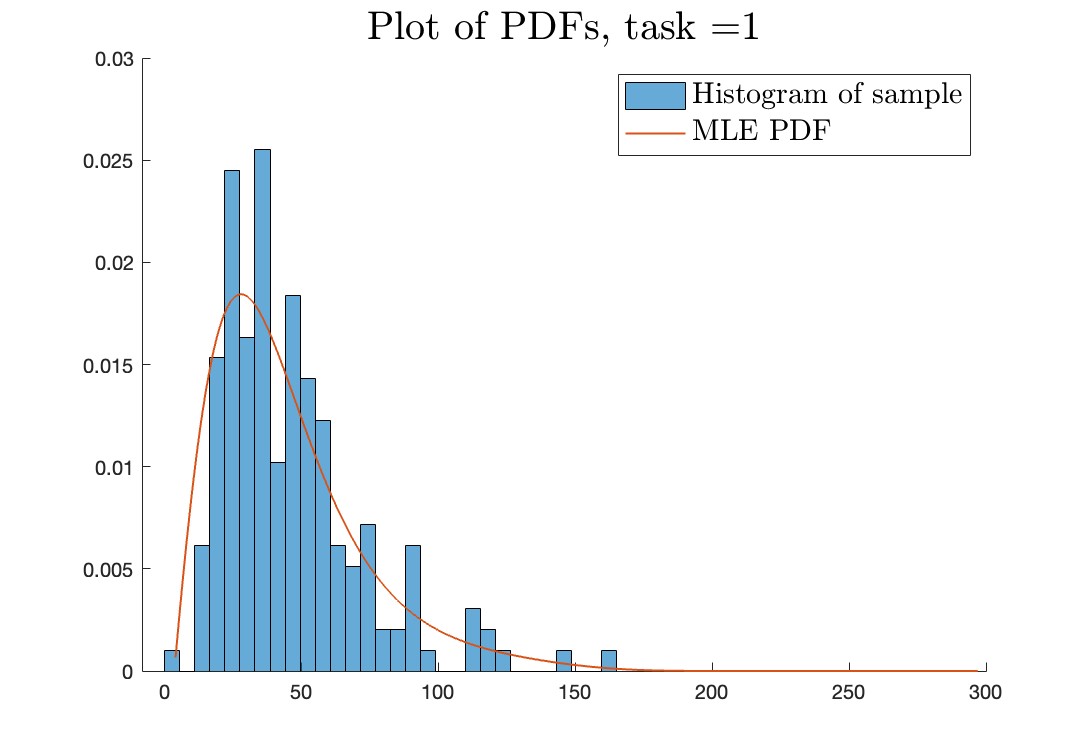}
    \end{subfigure}\hfill
        \begin{subfigure}[b]{0.5\textwidth}\centering
            \includegraphics[width=1.1\textwidth]{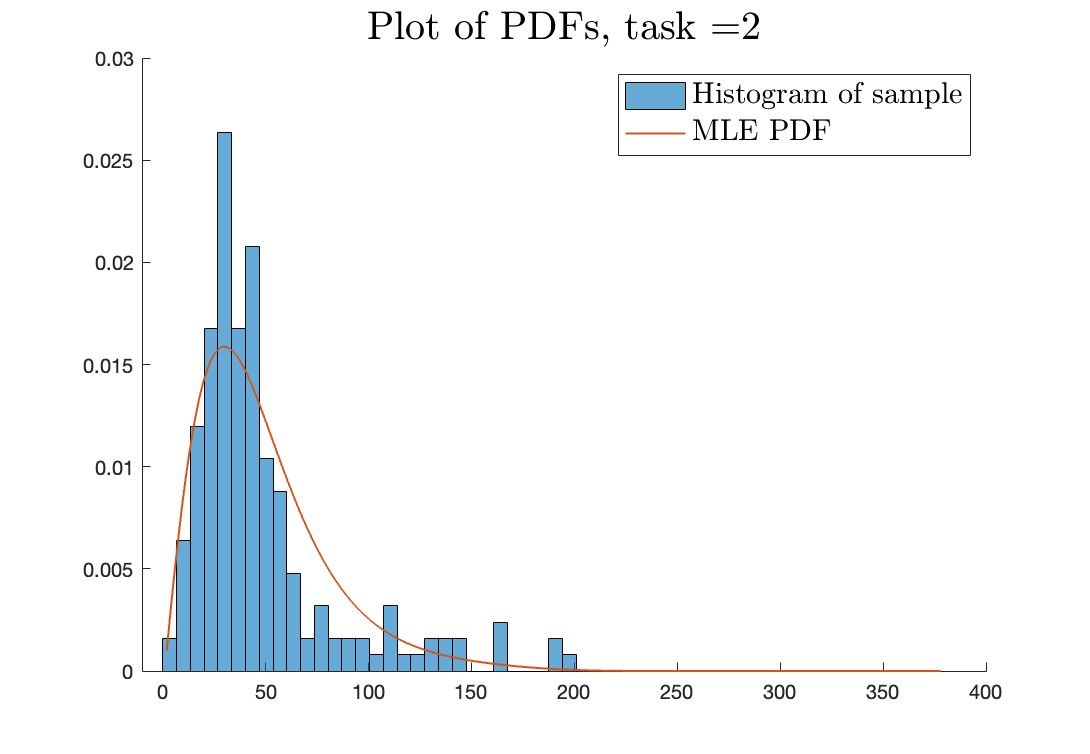}
\hfill
    \end{subfigure}
        \begin{subfigure}[b]{0.5\textwidth}\centering
            \includegraphics[width=1.1\textwidth]{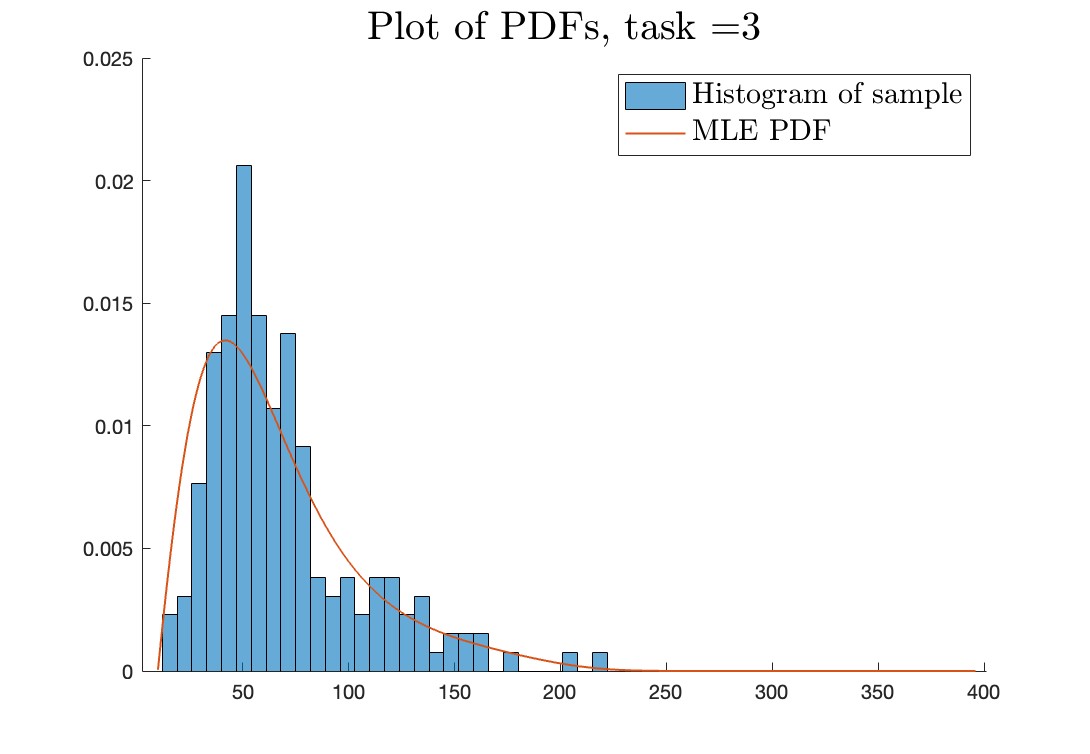}
    \end{subfigure}
        \caption{Histogram of the data from tasks 1, 2 and 3, along with the MLE estimated PMFs using the simple polynomial factor model with $n=3$. Here $T$ was allowed to range up to 200 integers higher than the maximum data value after the shift.}
        \label{fig:PDFexperipoly}
\end{figure}

\begin{figure}[]
    \centering
        \begin{subfigure}[b]{\textwidth}\centering
    \includegraphics[width=14cm]{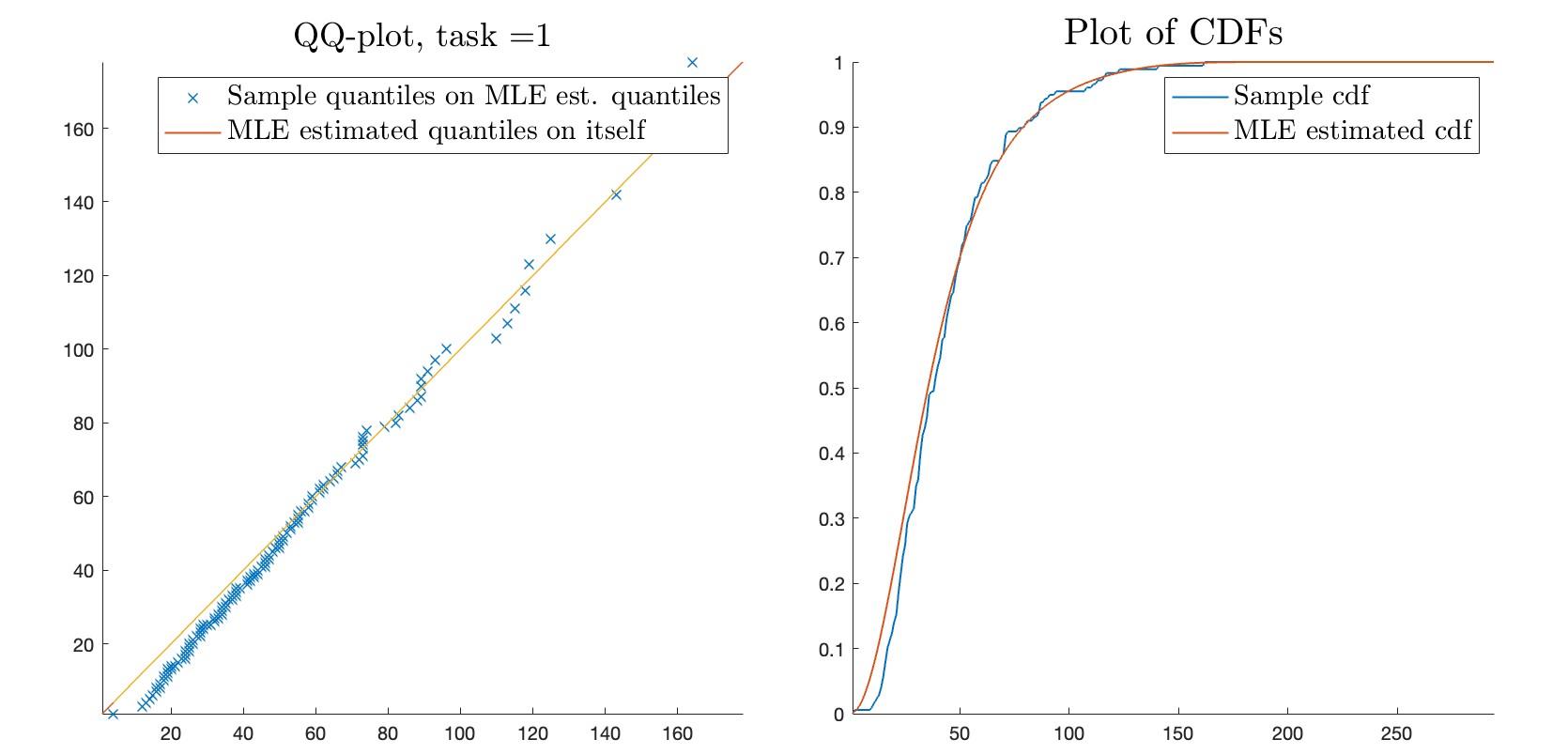}
    \end{subfigure}
        \begin{subfigure}[b]{\textwidth}\centering
    \includegraphics[width=14cm]{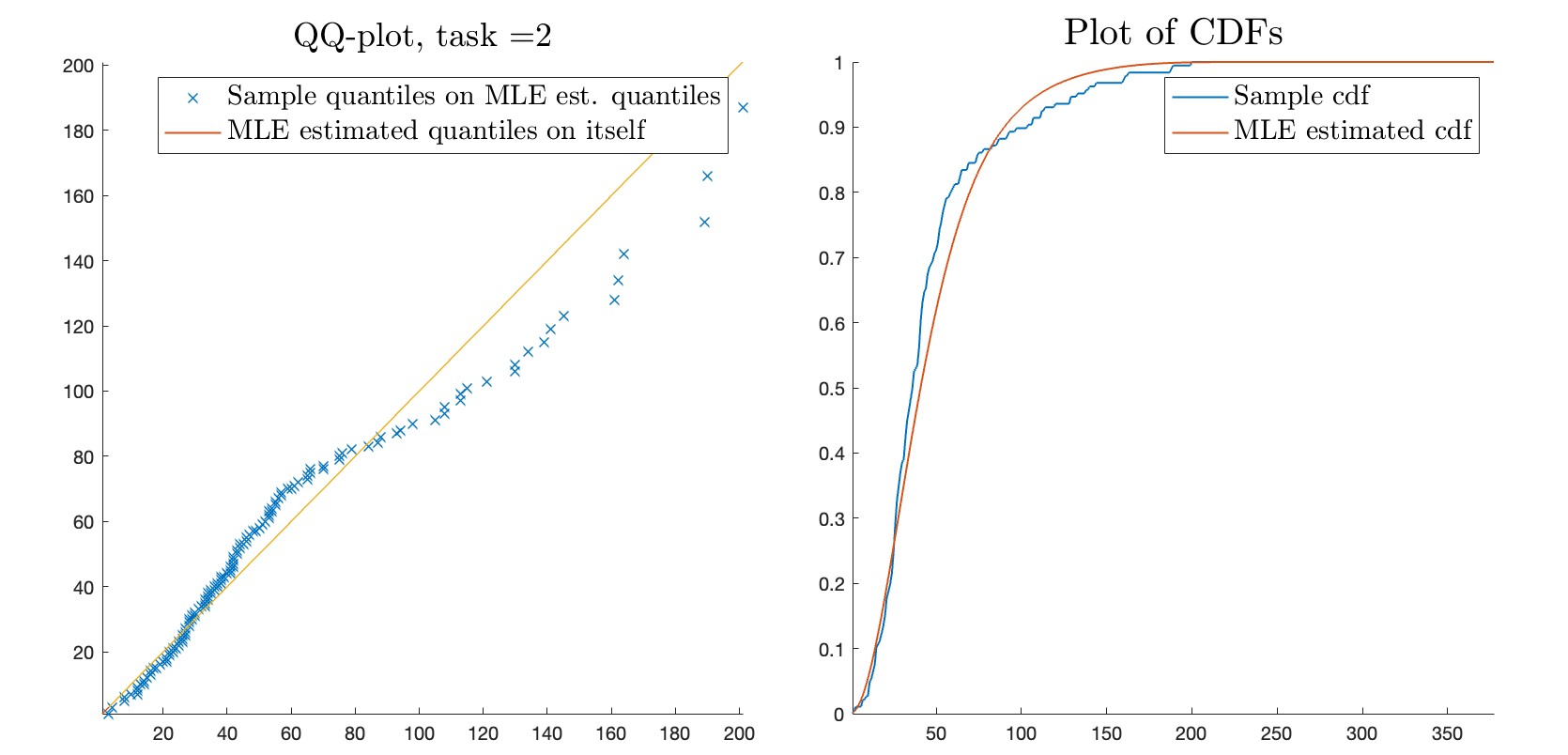}
    \end{subfigure}
        \begin{subfigure}[b]{\textwidth}\centering
    \includegraphics[width=14cm]{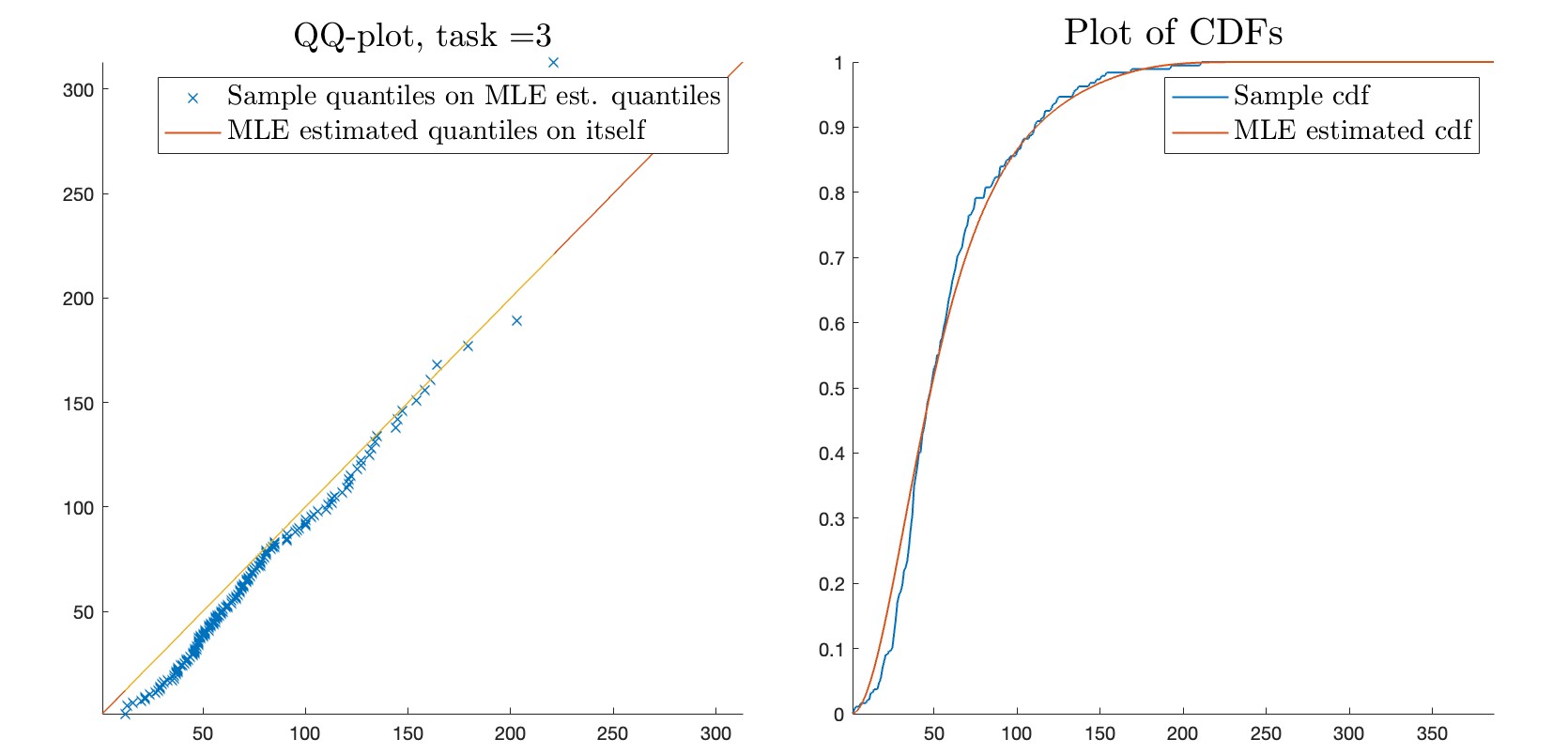}
    \end{subfigure}
        \caption{QQplots of the quantiles of the data from task 1,2, and 3 against the quantiles of the MLE estimated distribution, along with plots of the CDF of the sample against the CDF of the MLE estimated distribution. Here MLE parameters were estimated for the simple polynomial factor model with $n=3$.}
        \label{fig:CDFexperipoly}
\end{figure}

For the simple polynomial factor model with $n=3$ visual observation of the PMF plot seems to show close alignment between the recorded data and the MLE estimated PMF. In the QQ-plots we see mostly good alignment between the recorded data quantiles and the MLE estimated quantiles except perhaps task 2, where the data appears slightly more peaked than the distribution.

Note that applying to the simple polynomial factor model with $n=5$ did not markedly change the results of the plots, particularly for task 2. This may suggested trying a more complicated model if needed, with additional parameters. The other suggestion is that there is a relatively small data set (under 200 points) that this is applied to, so collecting more data could help improve the models. Overall we are happy with the fit for the simple polynomial factor model with $n=3$ given this small data set.

\section{Conclusion and future work} \label{sec:conclusion}
In this report we have described a new way to extend the geometric distribution and studied specific parametrised subfamilies including the linear factor model and the simple polynomial factor models in \Cref{sec:genGeometric}. We have studied maximum likelihood estimation of the parameters for these subfamilies, and also considered how to estimate the maximum sojourn time $T$ in the linear factor model case as in \Cref{sec:MLEsandCRLBs}. Finally in \Cref{sec:experidata} we applied our techniques to experimental data, where we conclude that the polynomial factor model of order three appears to offer very good performance in estimating sojourn time distributions for these experiments. 

While already a contribution to research in this area, with the aim to be applicable to semi-Markov models as in \Cref{sec: SemiMarkovRelationship}, there are many avenues of future research. These include analysis on estimating unknown $T$ in the simple polynomial factor models and examining the bias, variance and mean squared errors of the MLE estimators in greater detail.

Another avenue of research is to consider other generalisations, perhaps with additional parameters for the $\rho$. Examples include more general polynomial factor models for $\rho$, having $\rho$ of the form $e^{-g(k)}$ where $g:\mathbb{N}\to\mathbb{N}$, or a trignometic $\rho$ such as $\rho(k) = \frac{1}{2}\cos(k)+\frac{1}{2}$, provided $\rho\in [0,1]$ is guaranteed.  This may improve the fit for task 2 of the experimental data, as discussed in \Cref{sec:experidata}.

We would also like to understand more about these sojourn time distributions and how the MLEs behave in general. For example, in the simple polynomial factor models we found that there were larger variances observed for small sample sizes, however our observations of the different distributions generated by these MLEs showed little visual differences. Future work can include making this precise, for example by considering the KS-Statistics generated by the MLE distributions, as in \cite{DAgostino}, which give an measure of the distance between the CDFs (in this case using a $\sup$ norm). 

The second author has further explored this further in the appendix, where he replicates some of the results for the linear factor model then explores further modelling behaviour. In \Cref{appendix:unknownT} where he considers the $\ell_1$ distance between the PDFs for the linear factor model, showing that while the variance may be elevated the true distributions are converging as desired. This leaves open further questions around the connection of this approach to KS-statistics as well as how this applies in the simple polynomial factor model. In \Cref{appendix:exploglik} he also considers how the expected log-likelihood can also give indications of how the MLE procedure will results in different outcomes for different parameters of $a$. This we aim to explore this further in future work.

Finally, we intend to broaden this research by considering how to determine the parameters in online settings, including when estimating parameters for hidden semi-Markov models. This will likely mean adapting the expectation-maximisation algorithm usually used by these models to account for these sojourn time distribution parameters. We also seek to expand on the matrix analytic methods of \Cref{subsec:MAM}, including studying forgetting times and the general dynamics of the models. Part of this is to continue with the work in \Cref{sec:experidata} that seeks to determine when a sojourn time distribution is not geometric (i.e. non-Markovian) implying that we should use semi-Markov methods. 

\section{Acknowledgements}
This research is supported by the Commonwealth of Australia as represented by the Defence Science and Technology Group of the Department of Defence. We also thank Prof. Anna Ma-Wyatt and Dr Jessica O'Rielly for their involvement in collecting the data used in \Cref{sec:experidata}.


\appendix

\section{Linear transitions --- alternative model} \label{appendix}

In this appendix, the results of an independent study conducted by the second author (LBW) are presented. The focus is on the linear factor model (as in  \Cref{subsec:linearex}) with a slightly different parametrisation. Also, additional simulations for estimating the unknown $T$ are presented. One interesting outcome of these simulation studies is that identifiability issues can arise, namely a larger, incorrect estimate for $T$ is accompanied by a small estimate of $a$ but in such a way that the resulting estimated PMF remains close to the true one. This suggests that inference procedures based upon the estimated PMF (e.g. generalised likelihood ratio tests) may still be effective even in the MLEs are incorrect. So a comparison of the estimated PMF and the true one (via the $\ell_1$ norm) is provided (see \Cref{fig:est_pdf_dist}). Asymptotic consistency is observed. 

\subsection{MLEs for discrete product form distributions}

Let $\rho_i \in [0,1]$ for $i = 1,\ldots,T-1$, and define
\begin{align}
f_k & = \left\{ \begin{array}{ll} 1-\rho_1 & k = 1 \\
\rho_1 \, \rho_2 \, \cdots \, \rho_{k-1} \, \lb 1 - \rho_k \rb & k =2, \ldots, T , \end{array} \right.
\label{eq:prod_form}
\end{align}
where we set $\rho_T = 0$. We can assume that $\rho_i > 0$ for $i=1, \ldots, T-1$, otherwise we could appropriately truncate the model, since $\rho_i = 0 \Rightarrow f_k = 0$ for all $k > i$.  

It can be shown that 
\begin{align*}
\rho_i & = \frac{1 - F(i)}{1-F(i-1)} \ ,
\end{align*}
where $F : \Rbb \ra [0,1]$ is the distribution function for $K$, noting that $F(0) = 0$. Thus any PMF supported on $\{1, \ldots, T\}$ can be represented in the form \eqref{eq:prod_form}. 

One simple parametrisation of $f$ arises from the form $\rho_i = b-ai$ with $a,b > 0$. To set the support at $T$ samples, we let $\rho_T = b-aT = 0$ which yields $b = aT$, and thus $\rho_i = a(T-i)$. To ensure $\rho_i \leq 1$, we require $a \leq 1/(T-1)$. In the following, we assume that $T$ is known, so the model for $f$ contains only a single parameter $a \in (0,1/(T-1)]$. 

\subsection*{Maximum likelihood estimates}

Suppose we have a set of independent realisations $k_n, n = 1, \ldots, N$, of a random variable $K \in \{1, \ldots, T\}$, distributed according to $f$. Then from \eqref{eq:prod_form}, the log-likelihood for $a$ is
\begin{align*}
\log \Pr \lbr K_1 = k_1, \ldots, K_N = k_N ; a \rbr & = \sum_{n=1}^N \log f_{k_n} \\ & = \sum_{n=1}^N \ \sum_{i=1}^{k_n-1} \log \rho_i + \log \lb 1 - \rho_{k_n} \rb \ ,
\end{align*}
where we adopt the usual convention that an empty sum takes the value zero (i.e. for cases where $k_n = 1$).  Substituting the linear form for $\rho_i$, we have
\begin{align}
\log \Pr \lbr K_1 = k_1, \ldots, K_N = k_N ; a \rbr & = \sum_{n=1}^N \ \sum_{i=1}^{k_n-1} \log a + \log (T-i) + \log \lb 1 - a(T-k_n)  \rb \ .
\label{eq:LL_lin}
\end{align}
Taking the derivative w.r.t. $a$ gives
\begin{align*}
\frac{d\log \Pr \{.\}}{da} & = \sum_{n=1}^N \ \frac{k_n-1}{a} - \frac{T-k_n}{1 - a(T-k_n)} \ .
\end{align*}
There does not appear to be a simple closed form for a value $a$ resulting in this derivative being zero, so we will rely on numerical methods to maximise the log-likelihood over $a$ given some set of measurements. 

\subsection*{Cram\'{e}r-Rao bound}

For a single observation $k$ of $K$, we have
\begin{align}
\log \Pr \lbr K = k ; a \rbr & = \sum_{i=1}^{k-1} \log a + \log (T-i) + \log \lb 1 - a(T-k)  \rb \Rightarrow \nonumber \\
\frac{d \log \Pr \lbr K = k ; a \rbr}{da} & = \frac{k-1}{a} - \frac{T-k}{1 - a(T-k)} \Rightarrow  \nonumber \\
- \frac{d^2 \log \Pr \lbr K = k ; a \rbr}{da^2} & = \frac{k-1}{a^2} + \frac{(T-k)^2}{(1 - a(T-k))^2} .
\label{eq:FI}
\end{align}
Observe that the final quantity in \eqref{eq:FI} is strictly positive for any $k = 1, \ldots, T$, i.e. it is almost surely positive. This property has implications for MLEs. 
 
The Fisher information (FI) for $N$ i.i.d. samples can then be computed according to
\begin{align*}
\Phi(a) & = N \ \sum_{k=1}^T f_k \ \lb \frac{k-1}{a^2} + \frac{(T-k)^2}{(1 - a(T-k))^2} \rb \ .
\end{align*}
The Cram\'{e}r-Rao theorem states that the variance of any {\it unbiased} estimator for $a$ determined from $N$ i.i.d. samples of $K$ cannot be lower than $\Phi^{-1}(a)$. Caution however may be required in application of this result because the FI has a singularity at $a = 0$. 

\subsection*{An example}

We chose $T=10$ and $a = 0.9/(T-1) = 0.1$. The form of the PMF $f$ is shown on \Cref{fig1}(L). When we take $a = 0.5/(T-1) = 0.0556$, we obtain the PMF shown in \Cref{fig1}(R) and for $a = 0.1/(T-1) = 0.0111$, we obtain the PMF in \Cref{fig1}(B).  So for this parametrisation, there is a transition from a monotonically decreasing PMF for smaller $a$ values, to a unimodal PMF for larger $a$ values. The ability of this simple model to capture non-monotonic behaviour in the PMF is notable. 

\begin{figure}[!hpt]
\centering
\includegraphics[scale=0.4]{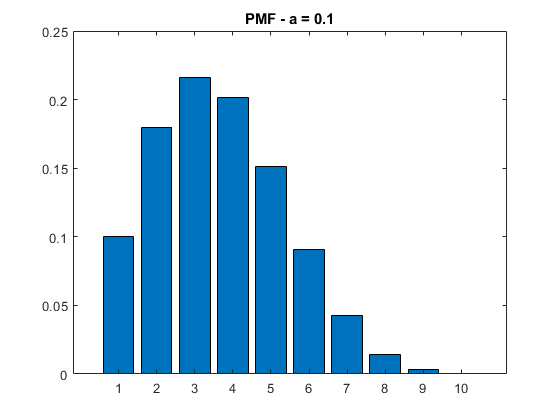}
\includegraphics[scale=0.4]{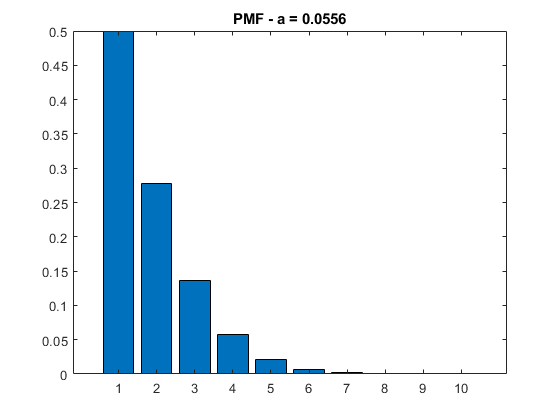}
\includegraphics[scale=0.4]{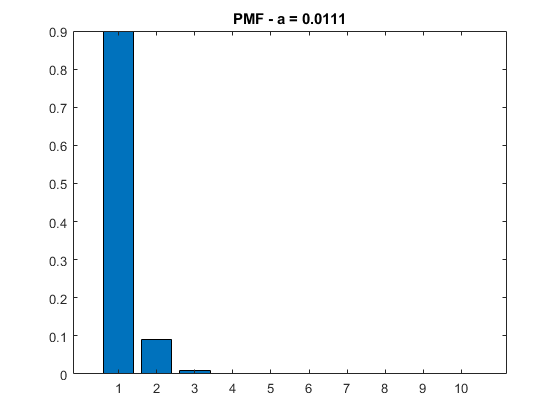}
\caption{Shows the PMFs for the linear product model with $T=10$, and $a = 0.1$ (L), $a = 0.0556$ (R), and $a = 0.0111$ (B). \label{fig1}}
\end{figure}

In \Cref{fig2}, the expected one-sample \footnote{The expected log-likelihood scales linearly with the number of i.i.d. observations.} log-likelihood functions (neglecting terms not dependent upon the parameter $a$) for the cases where the real parameter is $a = 0.0889$ (L) and $a = 0.556$ (R) respectively. 

\begin{figure}[!hpt]
\centering
\includegraphics[scale=0.4]{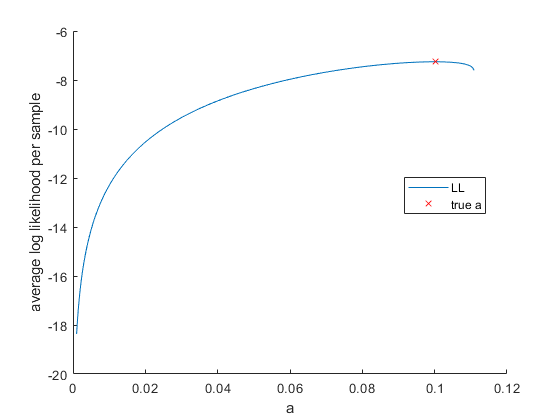}
\includegraphics[scale=0.4]{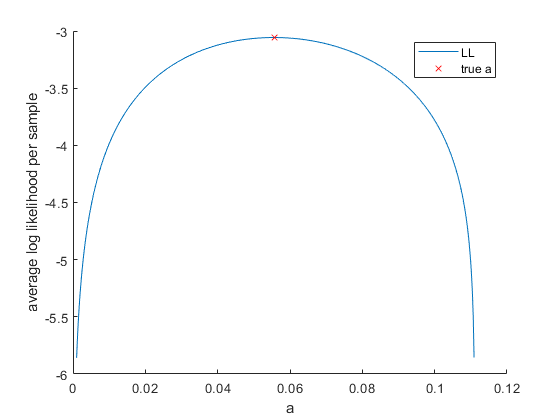}
\includegraphics[scale=0.4]{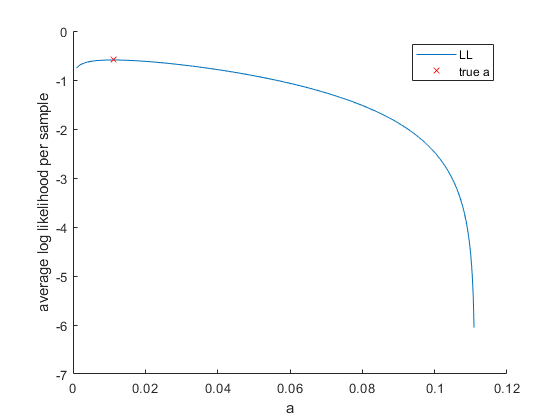}
\caption{Shows the expected log-likelihood function with $T=10$, and $a = 0.1$ (L), $a = 0.0556$ (R), and $a = 0.0111$. (B). \label{fig2}}
\end{figure}

We now examine the bias and variance of the MLE for $a$ as a function of the number of i.i.d. samples of $K$. The variance is compared to the Cram\'{e}r-Rao bound (for unbiased estimators). \Cref{fig3a}(L) shows the estimated MLE bias for the case $T=10, a = 0.1$, whilst \Cref{fig3}(R) shows the estimated MLE variance and the CRB (for an unbiased estimator) for the same parameters. In these simulations, 10,000 independent trials were averaged. Maximisation of the log-likelihood is done here on a grid of 500 values of $a$ uniformly spaced between 0.001 and the maximum value $1/(T-1)$. \footnote{The corresponding quantisation error variance for uniform sampling is $\Delta^2/12 = [(1/9-0.001)/500]^2/12 \approx 4e-9$ which is negligible.}

\begin{figure}[!hpt]
\centering
\includegraphics[scale=0.4]{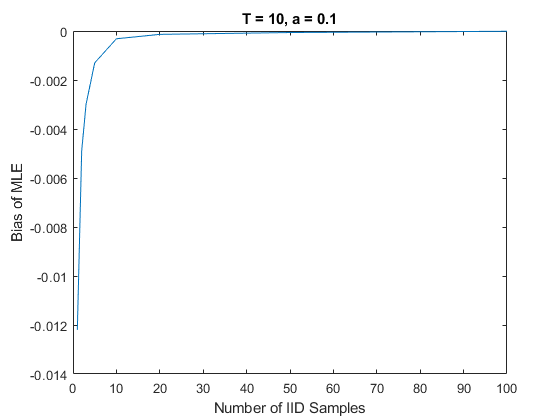}
\includegraphics[scale=0.4]{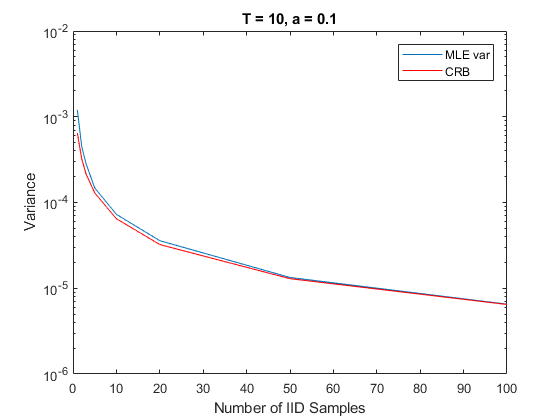}
\caption{Shows the bias (L) and variance (R) of the MLE for $a$ when $T=10$ and $a = 0.1$. The RHS figure also shows the CRB for an unbiased estimator. \label{fig3a}}
\end{figure}

\Cref{fig3} repeats this for the case where $a = 0.0556$. 

\begin{figure}[!hpt]
\centering
\includegraphics[scale=0.4]{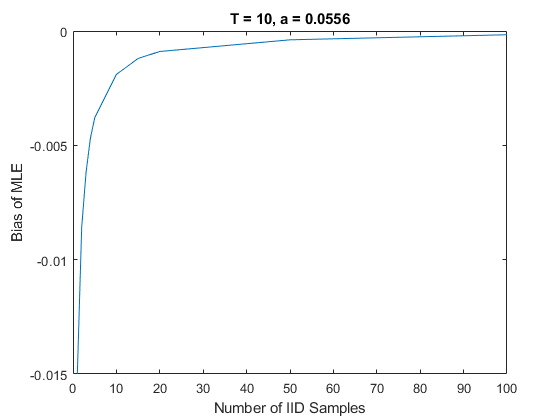}
\includegraphics[scale=0.4]{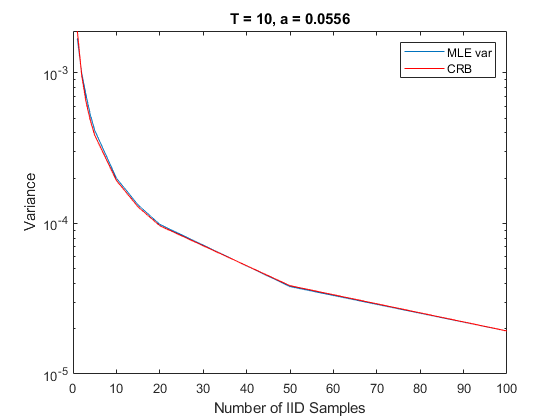}
\caption{Shows the bias (L) and variance (R) of the MLE for $a$ when $T=10$ and $a = 0.0556$. The RHS figure also shows the CRB for an unbiased estimator. \label{fig3}}
\end{figure}

\Cref{fig4} repeats this for the case where $a = 0.0111$. 

\begin{figure}[!hpt]
\centering
\includegraphics[scale=0.4]{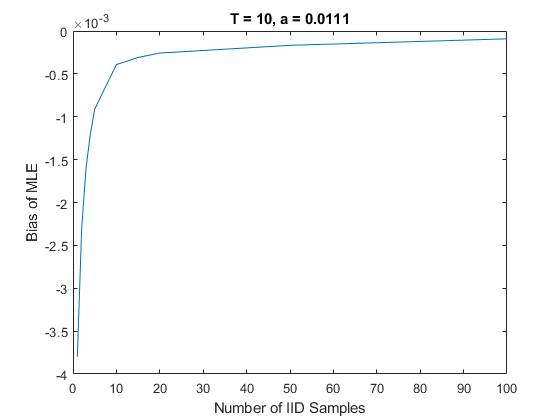}
\includegraphics[scale=0.4]{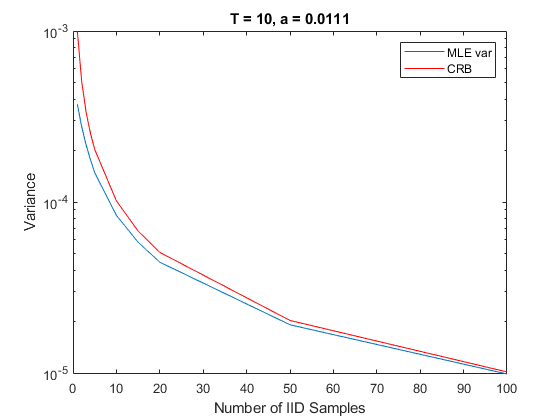}
\caption{Shows the bias (L) and variance (R) of the MLE for $a$ when $T=10$ and $a = 0.0111$. The RHS figure also shows the CRB for an unbiased estimator. \label{fig4}}
\end{figure}

\underline{Comments :}

\begin{enumerate}
\item
\Cref{fig4}(R) shows that the MLE variance is smaller than the CRB. This behaviour is not observed in \Cref{fig3}(R). In order to determine whether this is a fundamental issue or is a result of numerical effects, a higher fidelity simulation corresponding to the same parameters ($a = 0.0111, T = 10$) but with the minimum value of the range of $a$ for the log-likelihood maximisation reduced to 0.00001 with 10,000 discretisation points. Also, the number of independent trials used was increased to 100,000. \Cref{fig5}(L) compares the estimated bias for the original and high fidelity simulations, whilst \Cref{fig5}(R) does this for the estimated variance. 

\item For larger values of parameter $a$, we observe the MLE is close to unbiased and efficient even for relatively small numbers of data samples. 

\end{enumerate}
 
\begin{figure}[!hpt]
\centering
\includegraphics[scale=0.4]{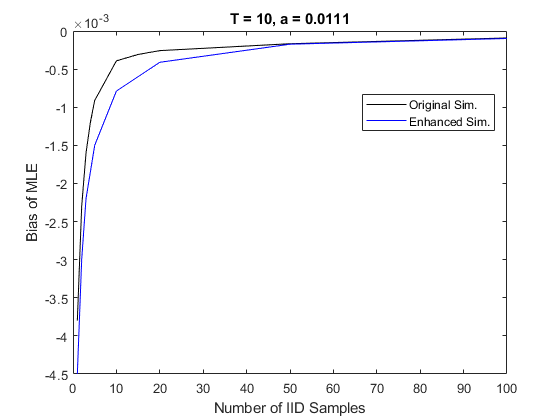}
\includegraphics[scale=0.4]{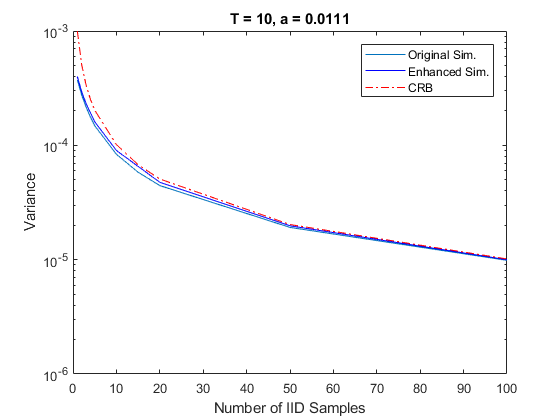}
\caption{Shows the bias (L) and variance (R) of the MLE for $a$ when $T=10$ and $a = 0.0111$. The RHS figure also shows the CRB for an unbiased estimator. These figures also show the results for the higher fidelity simulation.\label{fig5}}
\end{figure}

\subsection{Expected log-likelihood} \label{appendix:exploglik}

It is often useful to examine the form of the expected log-likelihood as a function of its parameter $a$, given some true value $a_0$. As noted, this quantity scales linearly with the number of measurements $N$, so it suffices to consider its form for $N=1$. From \eqref{eq:LL_lin}, we have
\begin{align*}
\Er_{K;a_0} \lbr \log f_K(.;a) \rbr & = \sum_{k=1}^T \ f_K(k;a_0) \ \sum_{i=1}^{k-1} \log a + \log (T-i) + \log \lb 1 - a(T-k) \rb \ ,
\end{align*}
where the expectation is with respect to the PMF $f(.;a_0)$. This quantity can be regarded as a function of the parameter $a$. These are the quantities shown in \Cref{fig2} (ignoring parts of the log-likelihood which do not depend on the functional parameter $a$). From asymptotic properties of the MLE, we know the maximum of this function will occur at $a = a_0$. The FI is a measure of the curvature of the expected log-likelihood at this point. Intuitively, a small FI means a ``flat'' curve, so it is arguable that this causes higher variation in the curve (when using finite samples) around the true parameter value, hence larger estimation error variance.

\subsection{\texorpdfstring{The case of unknown $T$}{The Case of Unknown T}} \label{appendix:unknownT}

In the above development, we considered the maximum support $T$ of the distribution of $K$ to be known. In many practical cases, this may not be true and, for completeness, we address the issues of dealing with unknown $T$. 

As a first investigation, we consider the expected log-likelihood function of variables $a$ and $T$, when the actual values of these parameters are $a_0$ and $T_0$. The variables $T$ and $T_0$ are each positive integers. Consider the log-likelihood function \eqref{eq:LL_lin}. For a given $T$, this function can be maximised over $a$, as above, yielding an estimate $\hat{a}(T)$. This is done for a range of possible values of $T$, and the corresponding $\hat{a}(T)$ is substituted back into \eqref{eq:LL_lin} which is then a function of $T$. That value of $T$ yielding a maximum is the MLE $\hat{T}$ for $T$ and the corresponding MLE for $a$ is $\hat{a}(\hat{T})$. 

If we examine the form of \eqref{eq:LL_lin} with the term not including $a$ present, we require $T \geq k_n$ for all $k_n > 1$, so we need to ensure $T \geq k_n$ for all measurements $k_n > 1$ to ensure the log-likelihood is well-defined. This places a lower bound on the admissible values of $T$, given a sequence of measurements. 

\subsubsection{An example}

When computing MLEs for $a_0$ and $T_0$, identifiability issues can arise. This is because when there are not many measurements $k_n$ close to the maximum $T$, the log-likelihood function tends to become a function of the product $a \, T$. So rather than assessing estimation performance by examining the MLEs $\hat{a}$ and $\hat{T}$, we compare the PMF for the observations obtained by substituting the MLEs into the functional form \eqref{eq:LL_lin} in place of the variables $a$ and $T$, and compare this PMF to the true one (\eqref{eq:LL_lin} with $a$ and $T$ substituted with $a_0$ and $T_0$ respectively. The $\ell_1$ norm between the estimated and true PMFs is used as a distance measure. 

For the search over candidate values of $T$, the range was restricted the set $\{T_0, \ldots, T_{\max}$, where $T_{\max}$ is given. This assumption is somewhat unrealistic since we do not know $T_0$ (indeed, the objective is to estimate it). Anecdotal evidence suggests setting the minimum $T$ for searching at the maximum sample value in a batch may be effective. We leave a detailed study on this issue to later work.

In all the subsequent examples, we chose $T_0 = 10$, $T_{\max} = 20$. The range of $a$ searched over for a specified $10 \leq T \leq 20$ is given by $[0.001, 1/(T-1))$. In these examples we used a discrete grid for $a$ with this region divided into 200 test points. In all examples, 10,000 independent trials were used. The number $N$ of i.i.d. samples per batch was varied and the corresponding MLEs and thus the estimated PMF was determined. It should be noted that in the case of erroneous estimates for $T_0$, the support of the estimated PMF $\{1, \ldots, T_0\}$ is larger than that for the true PMF, but as we shall see, incorrect estimates $\hat{T} > T_0$ result in smaller values for $\hat{a}$, and thus there tends to be few non-zero values for the estimated PMF in $\{T_0+1, \ldots, \hat{T}\}$. \Cref{fig:est_pdf_dist} shows the average (over 10,000 trials) of the distance $\|f - \hat{f} \|_1$ plotted against the batch size. 

\begin{figure}[!htp]
\centering
\includegraphics[scale=0.72]{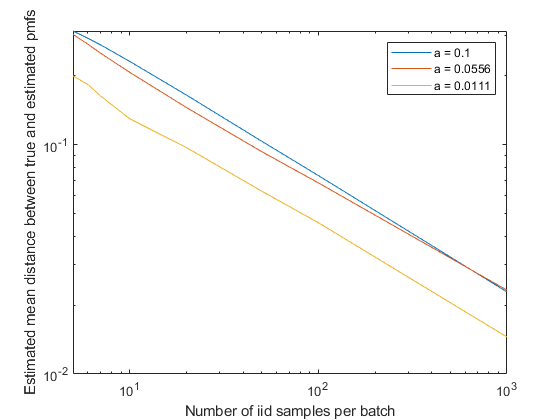}
\caption{Shows the average (over 10,000 independent realisations) of the $\ell_1$ distance between the true PMF and that PMF derived from MLEs, versus the number of independent samples per batch. Note the logaritmic scales on both axes. \label{fig:est_pdf_dist}}
\end{figure}

\Cref{fig:est_pdf_dist} displays the kind of behaviour that would be expected as the sample size $N$ increases, that is the ``estimated'' PMF approaches the true PMF. We can further illustrate this by comparing the true and estimated PMFs for a couple of different values of $N$. In \Cref{fig:comp_PMFs}, this is done for the values of $a_0 = 0.0111$ and $a_0 = 0.1$. We can observe close matches but we need to keep in mind that the estimates PMF is the sample mean over 10,000 trials and significant variability can be observed. We can illustrate this variability by adding ``error bars'' to the data shown on \Cref{fig:est_pdf_dist}. For reasons of clarity, we show the results for $a_0 = 0.1$ and $a_0 = 0.0111$ in \Cref{fig:pdf_dist_err_bars}. The error bars represent $\pm 1$ sample standard deviation about the sample mean (of the $\ell_1$ norm of the error). \\

\begin{figure}[!htp]
\centering
\includegraphics[scale=0.50]{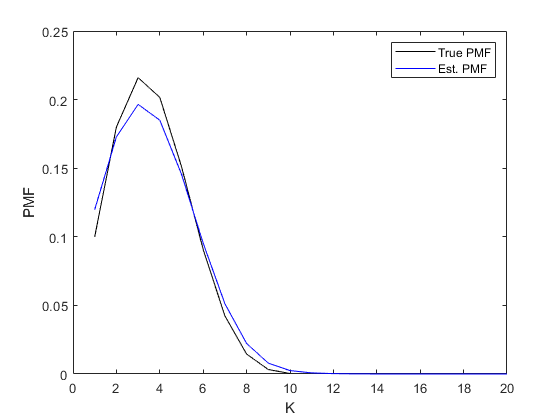}
\includegraphics[scale=0.50]{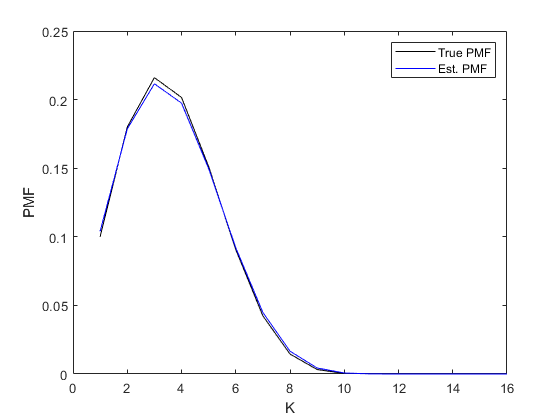}
\includegraphics[scale=0.50]{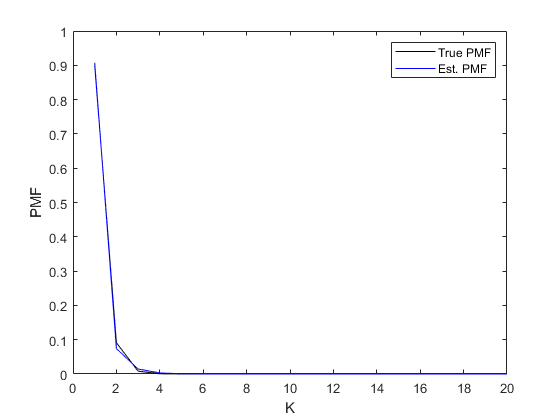}
\includegraphics[scale=0.50]{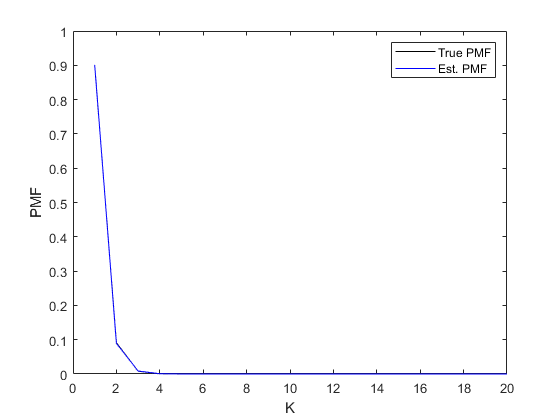}
\caption{Shows the true and estimated PMFs for $a_0 = 0.1$ (top) and $a_0 = 0.0111$ (bottom). The figures on the left are for sample size $N=5$, whilst the figures on the right are for $N=100$.  \label{fig:comp_PMFs}}
\end{figure}

\begin{figure}[!htp]
\centering
\includegraphics[scale=0.50]{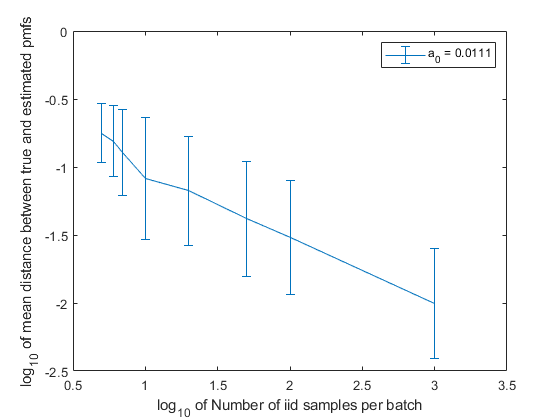}
\includegraphics[scale=0.50]{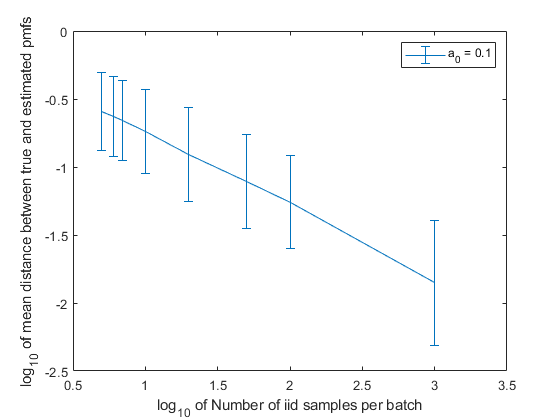}
\caption{Shows the sample average, and $\pm 1$ sample standard deviations for the $\ell_1$ distance between the true and estimated PMFs. Note the logarithmic scales on both axes. \label{fig:pdf_dist_err_bars}}
\end{figure}

\begin{figure}[!htp]
\centering
\includegraphics[scale=0.50]{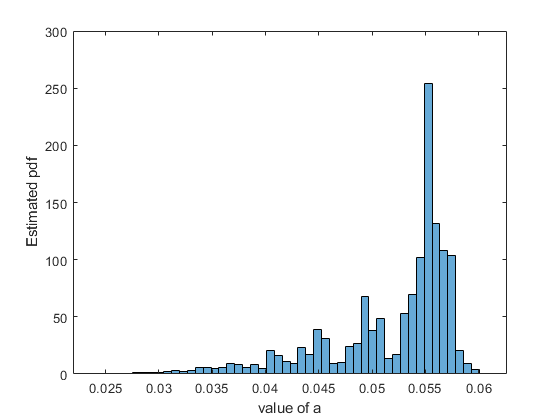}
\includegraphics[scale=0.50]{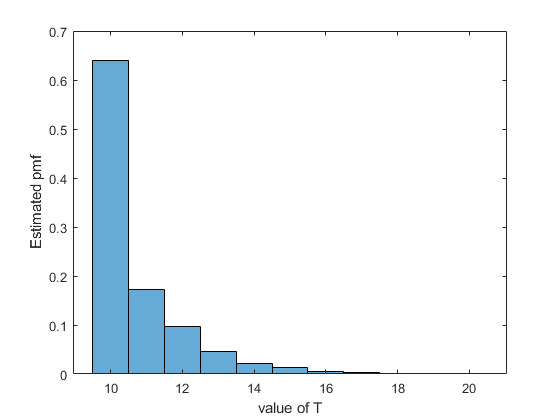}
\includegraphics[scale=0.50]{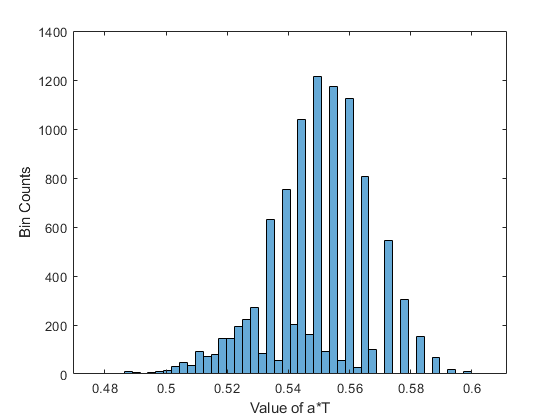}
\caption{Shows the histograms for the MLEs $\hat{a}$ (L), $\hat{T}$ (C), and $\hat{a}*\hat{T}$ (R). Here $a_0 = 0.0556$, and $N = 1000$. Observe the concentration of the product term $\hat{a}*\hat{T}$ around the true value of $a_0 \, T_0 = 0.556$. \label{fig:par_hists}}
\end{figure}

\pagebreak
\newpage
\section{Experimental data}\label{appn:ExperimentalData}

\begin{table}[!htp]
    \centering
\begin{tabular}{|l|l|l|l|l||l|l|l|l|l||l|l|l|l|l|}
\hline
\multicolumn{5}{|c||}{\textbf{Task 1}} &\multicolumn{5}{|c||}{\textbf{Task 2}}&\multicolumn{5}{|c|}{\textbf{Task 3}}\\\hline
25&24&44&24&38&33&79&47&41&25&67&106&147&48&69\\\hline
14&46&41&53&88&57&65&130&27&50&81&16&68&35&42\\\hline
24&19&125&73&28&16&31&49&53&108&60&39&127&53&85\\\hline
26&46&66&32&45&4&44&46&22&29&64&127&96&100&30\\\hline
17&93&50&61&37&26&14&70&29&32&71&45&62&54&161\\\hline
19&42&58&48&23&53&27&201&16&24&48&66&74&38&59\\\hline
39&62&89&113&20&8&28&42&27&27&60&57&79&50&41\\\hline
32&47&73&73&24&20&30&53&66&23&72&77&64&37&55\\\hline
34&164&86&74&27&17&22&34&16&94&13&110&179&46&60\\\hline
4&62&33&29&28&3&141&34&44&28&91&64&58&53&47\\\hline
22&66&115&71&22&41&29&37&30&190&78&73&50&49&131\\\hline
28&47&35&46&41&14&38&54&30&41&33&51&203&46&91\\\hline
24&44&52&58&110&20&61&42&8&42&45&69&84&100&85\\\hline
16&34&25&32&73&28&26&42&13&42&56&81&44&47&125\\\hline
18&118&42&67&27&23&113&162&30&12&38&66&121&46&71\\\hline
17&59&50&38&59&28&59&84&46&33&158&221&132&55&29\\\hline
15&73&66&34&28&16&145&93&43&161&22&72&68&42&48\\\hline
18&38&44&35&48&26&55&43&35&43&27&62&58&53&56\\\hline
13&44&53&64&31&19&113&39&41&60&104&80&75&42&78\\\hline
17&91&79&25&89&19&54&38&26&108&56&54&36&76&29\\\hline
26&41&27&82&35&21&139&42&53&22&59&114&46&145&20\\\hline
54&51&37&58&20&12&88&105&37&31&74&134&51&48&113\\\hline
46&51&119&25&55&43&56&98&28&25&120&59&48&27&95\\\hline
19&26&33&50&56&27&75&46&35&33&12&118&64&49&61\\\hline
18&59&26&43&53&40&66&33&65&121&39&135&70&44&43\\\hline
65&42&25&53&32&15&41&130&76&50&30&53&37&51&36\\\hline
21&143&34&26& &26&38&42&51&47&122&48&69&154&62\\\hline
12&38&36&51& &42&37&15&36&28&28&81&47&78&38\\\hline
52&35&21&57& &12&52&56&75&21&50&144&111&48&71\\\hline
19&28&35&22& &36&35&10&27&24&51&48&97&74&38\\\hline
55&55&28&32& &27&189&26&36&38&49&65&66&47&121\\\hline
24&48&34&41& &33&48&27&34&62&37&164&57&83&79\\\hline
23&83&32&44& &16&57&87&38&21&24&56&37&103&22\\\hline
16&46&36&43& &18&55&70&134&34&31&100&40&91&46\\\hline
50&72&55&49& &31&44&29&115&45&85&55&36&42&36\\\hline
30&18&29&38& &46&52&25&43& &68&73&113&41& \\\hline
47&96&34&89& &33&164&22&34& &45&70&31&38& \\\hline
34&61&37&38& &12&55&22&28& &45&53&67&29& \\\hline
\end{tabular}
   \caption{Experimental data collected of the time taken for participants to complete three screen-based tasks as detailed in \Cref{sec:experidata}.}
    \label{tab:experimentaldata}
\end{table}

\section{Comments on numerical optimisation for the polynomial factor model} \label{appn: NumericalPolyIssues}
In this appendix we comment on issues that arose from the numerical determination of the MLEs for the simple polynomial factor models in \Cref{subsec:PolyNumAnaly} for samples of size $m$. 
In particular, we had some issues with the Matlab function {\tt fmincon} for certain small samples. Note that {\tt fmincon} is a Matlab optimiser which we used to find the MLEs. For these small samples sizes, {\tt fmincon} returned the following
\begin{verbatim}
    Warning: Matrix is singular to working precision.
\end{verbatim}
When this occurred, it returned an MLE for $a$ that was slightly positive (but within constraint tolerance), even though we require that $a$ is negative.  This was in the specific case where $a=0.1\min(a)$ and $m=20$, and where the sample contained a single 3 and nineteen 1s. It also occurred in the case $a=0.5\min(a)$, $m=5,6,8$ where the sample contained one, two or three $3$s and the rest 1; and the case $a=0.9\min(a)$, $m=3,4,5$ where the samples were $(1,3,5),(1,3,3);(1,3,3,5),(1,3,3,7);(1,3,3,3,5),(1,3,3,3,9)$. 
This occurred in the first authors code. Due to this, significantly larger variance than the inverse Fisher information was initially recorded for these sample sizes. 

The second author's code did not have this issue with these samples. The differences in programming that lead to this were that the first author calculated the log-likelihood by finding the PMF $f$ for the given estimates of $a,c$, and then taking the sum of the log of this at the sample elements, which is the definition of the log-likelihood. However, the second author used the formula appearing in \Cref{eqn:loglikepolyac}, which uses the log laws such as $\log((-a)^{x-1})=(x-1)\log(-a)$, valid provided $a<0$. In the second author's case, whenever $a>0$ the use of these log laws meant that {\tt fmincon} returned an imaginary number, which this optimiser did not allow. In the first author's case, taking $\log((-a)^{x-1})$ when $x$ is odd and $-a<0$ results in a real number, and if $a$ is small enough to be within constraint tolerance, then {\tt fmincon} did allow this result. This occurred precisely for the samples listed above, which contain all odd numbers.

This appears to be a programming quirk of {\tt fmincon}. Editing the both authors' code to return an imaginary number for the log-likelihood whenever $a$ was positive improved the ability of {\tt fmincon} to return a feasible solution for the MLEs and improved the variance calculated.

\end{document}